\documentclass[draft,onecolumn]{IEEEtran}[12pt]
\date{}
\usepackage{cite}
\usepackage{graphicx}
\usepackage{subfigure}
\usepackage{amsmath}
\usepackage{amssymb}
\usepackage{amsfonts}

\usepackage{rotating}

\usepackage{cite}
\usepackage{graphicx}
\usepackage{subfigure}
\usepackage{url}
\usepackage{amsmath}
\usepackage{amsfonts}
\usepackage{amssymb}
\usepackage{bm}
\usepackage{array}
\usepackage{algorithmic}

\usepackage{longtable}

\allowdisplaybreaks

\newcommand{\SetR}{\mathbb{R}}
\newcommand{\SetZ}{\mathbb{Z}}

\newcommand{\relbd}[1]{\mathrm{rb}(#1)}

\newcommand{\ri}[1]{\mathrm{ri}(#1)}

\newcommand{\aff}[1]{\mathrm{aff}(#1)}

\newcommand{\ball}[2]{\mathrm{B}(#1,#2)}
\newcommand{\N}{\mathcal{N}}
\newcommand{\h}{\mathbf{h}}
\newcommand{\wh}{\widetilde{\h}}
\newcommand{\g}{\mathbf{g}}
\newcommand{\x}{\mathbf{x}}

\newcommand{\s}{\mathbf{s}}
\newcommand{\hsp}{\mathcal{H}}
\newcommand{\pmr}{\Gamma_n}
\newcommand{\ef}{\Gamma^*_n}
\newcommand{\efc}{\overline{\Gamma^*_n}}

\newcommand{\fix}{\mathrm{fix}}

\newcommand{\rb}{\mathbf{r}}
\newcommand{\ub}{\mathbf{u}}

\newcommand{\mi}{\mb{i}}

\newcommand{\mr}{\mathrm}
\newcommand{\mb}{\mathbf}
\newcommand{\mc}{\mathcal}
\newcommand{\mf}{\mathfrak}

\newcommand{\bs}{\boldsymbol}

\newcommand{\bsl}{\boldsymbol{\lambda}}

\newtheorem{theorem}{Theorem}
\newtheorem{lemma}{Lemma}
\newtheorem{corollary}{Corollary}
\newtheorem{proposition}{Proposition}

\newtheorem{definition}{Definition}
\newtheorem{example}{Example}

\newcommand{\QQQ}{$\Box$}

\newenvironment{proof}{\textit{Proof}}{\hfill\QQQ

}

\begin{document}

\title{Partition-Symmetrical Entropy Functions}


\author{
Qi~Chen
\IEEEmembership{Member IEEE}
 and Raymond W.~Yeung~\IEEEmembership{Fellow,~IEEE}\footnote{Qi Chen
   is with the Institute of Network Coding and the Department of
   Information Engineering, The Chinese University of Hong Kong, N.T.,
   Hong Kong (e-mail: qichen@ie.cuhk.edu.hk). Raymond W. Yeung is with
   the Institute of Network Coding and the Department of Information
   Engineering, The Chinese University of Hong Kong, N.T., Hong Kong,
   and with the Key Laboratory of Network Coding Key Technology and
   Application and Shenzhen Research Institute, The Chinese University
   of Hong Kong, Shenzhen, China (e-mail:
   whyeung@ie.cuhk.edu.hk). This paper was presented in part at ITW2013\cite{CY13}.}
}

\maketitle

\begin{abstract}
Let $\N=\{1,\cdots,n\}$. The entropy function $\h$ of a set of $n$
discrete random variables $\{X_i:i\in\N\}$
is a $2^n$-dimensional vector whose entries are $\h(\mc{A})\triangleq H(X_{\mc{A}}),\mc{A}\subset\N$, the (joint) entropies
of the subsets of the set of $n$ random variables with
$H(X_\emptyset)=0$ by convention.
The set of all entropy functions for $n$ discrete random variables,
denoted by $\Gamma^*_n$, is called the entropy region for
$n$. Characterization of $\Gamma^*_n$ and its closure
$\overline{\Gamma^*_n}$ are well-known
open problems in information theory. They are  important not
only because they play key roles in information theory problems
but also they are related to other subjects in mathematics and physics.

In this paper, we consider \emph{partition-symmetrical entropy
  functions}. Let $p=\{\N_1,\cdots, \N_t\}$
be a $t$-partition of $\N$. 
An entropy function $\bf h$ is called 
$p$-symmetrical if for all ${\cal A},{\cal B} \subset {\cal N}$, $\h({\cal A}) = \h({\cal B})$ whenever 
$|{\cal A} \cap {\cal N}_i| = |{\cal B} \cap {\cal N}_i|$, $i = 1,
\cdots,t$. The set of all the $p$-symmetrical entropy functions,
denoted by $\Psi^*_p$, is called $p$-symmetrical entropy function region.
We prove that $\overline{\Psi^*_p}$,
the closure of $\Psi^*_p$, is completely characterized by
Shannon-type information inequalities if and only if $p$ is the
$1$-partition or a $2$-partition with one of its blocks being a
singleton. 

The characterization of the partition-symmetrical entropy
functions can be useful for solving some information theory 
and related problems where symmetry exists in the structure of the
problems.

\end{abstract}

\section{Introduction}
\label{sec:1}
Let $\mc{N}=\{1,\cdots,n\}$. 
For a set of  (discrete) random variables
$X_\mc{N}=\{X_i:{i\in\mc{N}}\}$, 
we define a function $\h:2^{\N}\rightarrow\SetR$ by
\begin{equation*}
  \h(\mc{A})=H(X_\mc{A}),\ \mc{A}\subset\N
\end{equation*}
with $H(X_\emptyset)=0$ by convention.
Then $\h$ is called the \emph{entropy function} of $X_\N$.

Let $\hsp_n=\SetR^{
2^\mc{N}}$ be the \emph{entropy space} for $n$
random variables.\footnote{For a field $\mathbb{F}$ and a set $S$,
$\mathbb{F}^S$ denotes an $|S|$-dimensional space $\mathbb{F}^{|S|}$
whose coordinates are
labeled by $s\in S$.}
A vector $\h\in\hsp_n$ is called \emph{entropic}
if $\h$ is the entropy function for some set of $n$ random
variables, otherwise, it is called \emph{non-entropic}. The region in
$\hsp_n$ of all entropy functions 
is denoted by
$\Gamma^*_n$, called the \emph{entropy region}. As $\h(\emptyset)=0$ for any $X_\N$,
$\Gamma^*_n\subset\hsp^0_n\triangleq\{\h\in\hsp_n:\h(\emptyset)=0\}$
which is a subspace of $\hsp_n$\cite{Y97}.

It is well known that the entropy function satisfies the following
polymatroidal axioms: for all $\mc{A}, \mc{B}\subset\mc{N}$,
 \begin{align*}
 &\h({\emptyset})=0,\\
 &\h({\mc{A}})\leq \h({\mc{B}}), \quad \text{if }\mc{A}\subset\mc{B},\\
 &\h({\mc{A}})+\h({\mc{B}})\geq
 \h({\mc{A}\cap\mc{B}})+\h({\mc{A}\cup\mc{B}}),
 \end{align*}
that is, any entropy function $\h$ is  (the rank
function of) a polymatroid \cite{F78}. The
polymatroidal axioms are equivalent to the basic information
inequalities \cite[App. 14A]{Y08} from which all Shannon-type
information inequalities can be derived. The set of polymatroids, or
equivalently, the region bounded by
Shannon-type information inequalities,
is denoted by $\Gamma_n$, called the \emph{polymatroidal region}. Then
$\Gamma^*_n\subset\Gamma_n$.

Now the question is whether all polymatroids are entropic, or
whether $\Gamma^*_n=\Gamma_n$. It can be shown that
$\Gamma^*_2=\Gamma_2$, while $\Gamma^*_3\subsetneq\Gamma_3$ due to
the existence of non-entropic polymatroids on the boundary of
$\Gamma_3$\cite{ZY97,M06,CY12}.
However, by taking the closure of
$\Gamma^*_3$, we have $\overline{\Gamma^*_3}=\Gamma_3$. The vectors
in $\overline{\Gamma^*_n}$ are called \emph{almost entropic}. Thus,
all polymatroids are almost entropic when $n=3$\cite{ZY97}.
This was proved not
to be true for $n\ge 4$, i.e.
$\overline{\Gamma^*_n}\subsetneq\Gamma_n$, due to
the existence of unconstrained non-Shannon-type information
inequalities\cite{ZY98}. For a comprehensive treatment of the
subject, we refer the readers to \cite[Chapter 13-15]{Y08}.

Following the discovery of the first unconstrained non-Shannon-type information
inequality in \cite{ZY98}, many such inequalities have been
found, e.g., \cite{YYZ01,MMRV02,Z03, DFZ06,XWS08,DFZ11}. The region $\Gamma^*_n$ was proved to be
``solid inside''\cite[Theorem 1]{M07b}, that is, for any
$\h\in\overline{\Gamma^*_n}$ and $\h\notin\Gamma^*_n$, $\h$ must be on
the boundary of $\overline{\Gamma^*_n}$. It
was further proved in \cite{M07a} that there exist infinitely many
independent linear
non-Shannon-type inequalities.

Characterizations of $\Gamma^*_n$ and its closure
$\overline{\Gamma^*_n}$ are important not only because information
inequalities play key roles in the proof of converse coding theorems
but also they are related to probability theory, quantum mechanics\cite{NC00} and
matrix theory\cite{CGY12}, and have one-to-one correspondence with network coding\cite{CG08},
group theory\cite{CY02}, Kolmogorov complexity\cite{HRSV00} and
combinatorics\cite{C01}. For a comprehensive treatment of the relations between
entropy region $\Gamma^*_n$ and other subjects, readers are referred to
\cite{Y12, C11}. However, full characterizations of
$\Gamma^*_n$ and $\overline{\Gamma^*_n}$
are extremely difficult. To obtain partial characterizations of these regions,
constraints can be added to the boundary of the region and
corresponding constrained
non-Shannon-type inequalities have been discovered\cite{ZY97,KR11,KR12,KR13}.

In this paper, we consider \emph{partition-symmetrical entropy
  functions} defined as follows. A
partition $p$ of $\N$ is a set of nonempty subsets $\{\N_1,\cdots,\N_t\}$ of $\N$ such that distinct blocks $\N_i$ and
$\N_j$ are disjoint and $\cup^t_{i=1}\N_i=\N$. It induces a permutation group
$\Sigma_p$ whose members are those
permutations that keep the elements of each block $\N_i$ in the same
block. We define an action of group $\Sigma_p$ on
the entropy space $\hsp_n$, for any
$\sigma\in\Sigma_p$ and for $\h\in\hsp_n$,
$\sigma(\h)(\mc{A})=\h(\sigma(\mc{A})),\ \mc{A}\subset\N$.
This group action can naturally induce an action on the power set of
$\hsp_n$, i.e., for any $T\subset\hsp_n$,
$\sigma(T)=\{\sigma(\h):\h\in T\}$.
Then the
fixed set of $\hsp_n$, $\fix_p=\{\h\in\hsp_n:\sigma(\h)=\h,\ \text{for all
}\sigma\in\Sigma_p\}$, is a subspace of
$\hsp_n$ and is called the $p$-symmetrical subspace. 
It can be seen
that, $\h\in \fix_p$ if and only if $\h(\mc{A})=\h(\mc{B})$ whenever
$|\mc{A}\cap\N_i|=|\mc{B}\cap\N_i|$ for all $i=0,\cdots,t$. We call
a polymatroid or an entropy function $p$-symmetrical if it is in
the $p$-symmetrical subspace. The $p$-symmetrical polymatroids and
$p$-symmetrical entropy functions form the $p$-symmetrical
polymatroidal region,
denoted by $\Psi_p$, and $p$-symmetrical entropy function region,
denoted by $\Psi^*_p$,
respectively.

We prove in Theorem \ref{bfdp} that $\Psi_p=\overline{\Psi^*_p}$ if and only if $p$
is the 1-partition or a 2-partition with one of its blocks being a
singleton when $n\ge 4$. To prove the ``if'' part of Theorem \ref{bfdp}, we analyze the
extreme rays of $\Psi_p$ for the two cases of $p$ such that
$\Psi_p=\overline{\Psi^*_p}$ and show that these extreme rays contain
factors of uniform matroids which are almost entropic. To prove the ``only
if'' part of Theorem \ref{bfdp}, we show that $\Psi_p$ for other $p$
contain polymatroids that can be restricted to a factor of the V\'amos
matroid, which is known to be not almost entropic. Toward establishing
Theorem \ref{bfdp}, we prove some symmetrical properties pertaining to
$\Gamma_n$, $\Gamma^*_n$, $\Psi_p$ and $\Psi^*_p$. In particular, we
prove in Theorem \ref{bafoi} that each facet of $\Psi_p$ corresponds
to a
$p$-orbit of facet of $\Gamma_n$.

The rest of the paper is organized as follows.
Section \ref{udifb} gives the
preliminaries on convex analysis, matroid theory, partitions and group
theory that
are relavent to the discussion in this paper.
The problems 
are set up in Section \ref{bafdg}, where Theorem \ref{bfdp}, the main theorem is
stated. In Section
\ref{yuber}, we 
establish some symmetrical properties pertaining to $\Gamma_n$ and
$\Gamma^*_n$. The proof of this theorem is given in Section \ref{5555}. 
Discussions on applications to secret-sharing and further research
are in the last section. 
We close this section with the list of
notations in this paper.

{
\scriptsize 
\begin{longtable}{lllc}
    \hline
\hline
$\N$&$\{1,\cdots,n\}$& index set & Section \ref{sec:1}\\

$\hsp_n$&$\SetR^{2^\mc{N}}$& entropy space & \\

$\hsp^0_n$&$\{\h\in\hsp_n:\h(\emptyset)=0\}$& & \\

$\Gamma^*_n$& &entropy function region&\\

$\Gamma_n$& &polymatroidal region&\\
\hline
$\aff{A}$&&affine hull of $A$&Subsection \ref{sec:convex-cones}\\
 $E(i)$&$\{\h\in \pmr: \h(\mathcal{N}\setminus\{i\})=\h(\mathcal{N})\},\
        i\in\mathcal{N}$& facets of
        $\Gamma_n$, first type&\\
$E(ij,\mc{K})$&$\{\h\in \pmr:
\h(\mc{K}\cup\{i\})+\h(\mc{K}\cup\{j\})=\h(\mc{K})$&facets of
        $\Gamma_n$, second type
&\\
&$+\h(\mc{K}\cup\{i,j\})\}, i,j\in\mc{N},\mc{K}\subset
        \mathcal{N}\setminus\{i,j\}$&&\\
$E(\mc{I},\mc{K})$& $E(i)=E(\{i\},\emptyset)$,
$E(ij,\mc{K})=E(\{i,j\},\mc{K})$& &\\ 
$\mc{E}_n$& &the set of all facets of $\Gamma_n$& \\
  $\ball{\mb{c}}{r}$&$\{\mb{x}\in\SetR^d:\|\mb{x}-\mb{c}\|_2<r\}$&open
  ball&\\
$\ri{A}$&$\{\mb{x}\in\SetR^d:\exists
  \epsilon>0,\ball{\mb{x}}{\epsilon}\cap \aff{A}\subset A\}$&relative interior of $A$&\\
$\relbd{A}$&$\overline{A}\setminus\ri{A}$& relative boundary of $A$&\\
$U_{m,n}$& $U_{m,n}(\mc{A})=\min\{m,
  |\mc{A}|\}\ \mc{A}\subset\N$&uniform matroid&Subsection
  \ref{sec:matroid-theory}\\
$p$&$\{\N_1,\cdots,\N_t\}$ s.t. $\N_i$ disjoint and
$\N=\cup^t_{i=1}\N_i$&partition of $\N$&Subsection \ref{sdsb}
  \\
$\mc{P}_n$&&the set of all partitions of $\N$&\\
$\mc{P}_{t,n}$&&the set of all $t$-partitions of $\N$&\\
$\lambda_{\mc{A},p}$&$(|\mc{A}\cap\N_1|,\cdots,|\mc{A}\cap\N_t|)$&partition
vector of $\mc{A}$ under $p$&\\
$\lambda_{\mc{A},p}(i)$&$|\mc{A}\cap\N_i|$& the $i$-th entry of
$\lambda_{\mc{A},p}$&\\
  $\lambda_p$& $\lambda_{\N,p}$&the partition vector of $p$&\\
$\mb{n}$&$[n_1,\cdots,n_t]$ with
  $0<n_i\le n_j,1\le i<j\le t$&partition of $n$&\\
$\mb{n}_p$&&nondecreasing arrangement of $\lambda_{p}(i)$&\\
$\Sigma_n$&&symmetric group on $\N$&Section
  \ref{hidmx}\\
$\Sigma_p$ &$\{\sigma\in\Sigma_n: \sigma(j)\in\N_i,\  j\in\N_i,\ i=1,\cdots,t\}$&$p$-group&\\
$\mc{O}_p,\mc{O}_p(\h)$&&$p$-orbit&\\
$\fix_{p} (T)$&&fixed set of $\mc{T}$ under $p$-group&\\
 $S_p$&$\{\h\in\hsp_n:\ \h(\mc{A})=\h(\mc{B}),
  \text{ if } \lambda_{\mc{A},p}=\lambda_{\mc{B},p}\}$&$p$-symmetrical
  subspace&\\
\hline
$\mc{N}_p$&$\{(k_1,\cdots,k_t):k_i\in\{0,1,\cdots,n_i\},i=1,\cdots,t\}$&&Section \ref{bafdg}\\
$\s(\h,p)$&$(s_{k_1,\cdots,k_t})_{(k_1,\cdots,k_t)\in\N_p}$&&\\
$\Psi^*_p$&$\Gamma^*_n\cap S_p$&$p$-symmetrical entropy function
region&\\
$\Psi_p$&$\Gamma_n\cap S_p$&$p$-symmetrical polymatroidal
region&\\
$S^0_p$&$\{\h\in S_p:\h(\emptyset)=0\}$&&\\
$\mc{P}^*_n$&&representative set of partitions of $\N$&\\
$\mc{P}^*_{t,n}$&&representatives set of $t$-partitions of $\N$&\\
\hline
$\mf{F}_p$&&collection of all $p$-orbits of faces of $\Gamma_n$&Section \ref{yuber}\\
$\mf{E}_p$&&collection of all $p$-orbits of facets of $\Gamma_n$&\\
 $\mf{N}_p$&&set
of all distinct pairs of
$(\lambda_{\mc{I},p},\lambda_{\mc{K},p})$&\\
$\mc{E}_p(\lambda_{\mc{I},p},\lambda_{\mc{K},p})$&&the $p$-orbit that
$E(\mc{I},\mc{K})$ belongs to&\\
$\mb{0}_t$&&zero vector with dimension
$t$&\\
$\mb{1}_{t}(l)$&&$t$-dimensional vector with the $l$-th &\\
&&entry 1 and other
entries 0&\\
$\mb{2}_{t}(l)$&&$t$-dimensional vector with the $l$-th &\\
&&entry 2 and other
entries 0&\\
$\mb{1}_{t}(l_1,l_2)$&&$t$-dimensional vector with the $l_1$-th&\\
&&
and $l_2$-th entries 1 and other entries 0&\\
$\mc{G}_p$&&the collection of all facets of $\Psi_p$&\\
$\omega_p$&$E\mapsto E\cap S_p$ &&\\
$G_p(\lambda_{\mc{I},p},\lambda_{\mc{K},p})$&$
  E(\mc{I},\mc{K})\cap S_p$&facet of $\Psi_p$&\\
\hline
$\mf{E}_{\mc{E},p}$&& the family
$p$-orbits contained in $p'$-orbit $\mc{E}$&Appendix\\
\hline
\hline
  \caption{Notation List}
  \label{tab:1}
  \end{longtable}

}

\section{Preliminaries}
\label{udifb}
\subsection{Convex cone}
\label{sec:convex-cones}
A convex set $C\subset\SetR^d$ is called a \emph{convex cone} if for
any $\mb{c}\in C$ and $a\ge 0$, we have $a\mb{c}\in C$. 
A convex
cone which does not contain a line is called \emph{pointed}. In
this paper, convex cones are assumed to be pointed and closed unless
otherwise specified. From the definition, it can be seen that
$\Gamma_n$ is a convex cone in $\hsp_n$. It was shown in
\cite{ZY97} that $\overline{\Gamma^*_n}$ is also a convex cone.

A convex cone is called \emph{polyhedral} if it is the intersection
of a finite set of closed halfspaces. Since each closed halfspace is
induced by a linear inequality and the number of linear inequalities
in the set of polymatroidal axioms is finite for a fixed $n$,
$\Gamma_n$ is a polyhedral cone. On the contrary,
$\overline{\Gamma^*_n}$ is not polyhedral as proved
in \cite{M07a}.

A hyperplane $P$ is called a \emph{supporting hyperplane} of a
convex set $C\subset\SetR^d$
if one of its corresponding closed halfspaces\footnote{For a
hyperplane $P=\{\x\in\SetR^d: \mb{c}^T\x=a\}$, its two
corresponding closed halfspace are $\{\x\in\SetR^d: \mb{c}^T\x
\leq a\}$ and $\{\x\in\SetR^d: \mb{c}^T\x\geq a\}$, where
$\mb{c}\in\SetR^d$ and $a\in\SetR$.} 
$P^+\supset C$ and
$\mr{dist}(P,C)=0$.\footnote{Given $A, B\subset\SetR^d$,
$\mr{dist}(A,B)=\inf_{x\in A,y\in B}\|x-y\|_2.$
} 

\begin{definition}
[Face] 
\label{def:sd} 
A \emph{face} of a convex cone
$C\subset\SetR^d$
 is the cone $C$ itself or
$C\cap P$, where $P$ is a supporting hyperplane of $C$. A face that
is not $C$ or the origin is called a \emph{proper face} of the cone.
\end{definition}

Note that if $\mr{dim} C< d$,\footnote{
Given $A\subset\SetR^d$, $\mr{dim}A$
is defined by the dimension of $\aff{A}$, the affine hull of
$A$.
}  for any hyperplane $P\supset C$, we
have $C=C\cap P$. We can readily see that $P$ is a supporting hyperplane
of $C$.
Therefore, for convex cone $C$ with $\mr{dim}
C< d$, all faces $F$ of
$C$ can be written as $C\cap P$ for some supporting hyperplane $P$ of $C$.

A \emph{maximum proper face} of a convex cone is a proper face which
is not contained by any other proper face. Similarly, a
\emph{minimum proper face} is a proper face that does not contain
other proper faces.

\begin{definition}[Extreme ray] An \emph{extreme ray} $R$ of
a convex cone $C$ is a subset of $C$ and for any $\mb{r}\in R$ such
that $\mb{r}=\mb{c}_1+\mb{c}_2$ and $\mb{c}_1,\mb{c}_2\in C$, we
have $\mb{c}_1, \mb{c}_2\in R$, where $\mb{c}_1=a \mb{r}$ and
$\mb{c}_2=(1-a)\mb{r}$ for some $a\in\SetR$.
\end{definition}

For a polyhedral convex cone $C$ with $\mathrm{dim}
C=d'$, where $d'\ge 2$, 
the maximum proper faces, also called \emph{facets}, are the
$(d'-1)$-dimensional faces, while the minimum proper faces are the
$1$-dimensional faces and they coincide with the extreme rays
of the cone. Note that if $\mathrm{dim}
C=1$, then $C$ does not have a proper face.

The family $\mc{F}$ of faces of a convex cone $C$ form a lattice
called the \emph{face lattice} of the convex cone, which is
partially ordered by inclusion. That is, for any $F_1, F_2\in
\mc{F}$, $F_1\le F_2$ if and only if $F_1\subset F_2$. Furthermore,
the faces of a convex cone have the following properties.

\begin{proposition}
 Any non-origin face of a
convex cone is the convex combination of some extreme rays of the cone.
\end{proposition}

\begin{proposition}
\label{njkdf}
 The intersection of any collection of faces
of a convex cone is a face of the cone.
\end{proposition}

\begin{proposition}
\label{vssld}
  Any face of a convex cone that is not the cone itself is the
  intersection of some facets.
\end{proposition}

For convex sets and convex polyhedral cones, readers are referred to
\cite{G03,Z95,R70} for a detailed discussion.

Elemental information inequalities involving $X_\mc{N}$ have the
following two forms:
    \begin{enumerate}
        \item $H(X_\mathcal{N})\geq H(X_{\mathcal{N}\setminus\{i\}}),\
        i\in\mathcal{N}$;
        \item $I(X_i;X_j|X_\mc{K})\geq 0,\quad$ distinct $i,j\in\mc{N},\ \mc{K}\subset
        \mathcal{N}\setminus\{i,j\}$.
    \end{enumerate}
These inequalities are called elemental since every Shannon-type
information inequality can be written as a conic
combination of these inequalities, and they form the minimal set of
inequalities that has this property \cite[Section 14.6]{Y08}. In other
words, the two forms of elemental inequalities give the ``minimal''
characterization of $\pmr$.

Setting the elemental inequalities to equalities and intersecting
the corresponding hyperplanes with $\Gamma_n$, we
obtain the facets of the cone $\Gamma_n$ in the following forms:
    \begin{enumerate}
        \item $E(i)\triangleq\{\h\in \pmr: \h(\mathcal{N}\setminus\{i\})=\h(\mathcal{N})\},\
        i\in\mathcal{N}$;
        \item $E(ij,\mc{K})\triangleq\{\h\in \pmr: \h(\mc{K}\cup\{i\})+\h(\mc{K}\cup\{j\})=\h(\mc{K})+\h(\mc{K}\cup\{i,j\})\},\ \text{ distinct } i,j\in\mc{N},\mc{K}\subset
        \mathcal{N}\setminus\{i,j\}$.
    \end{enumerate}
The set of all facets of $\Gamma_n$ is denoted by $\mc{E}_n$. For
notational convenience, the members of $\mc{E}_n$ are denoted by
$E(\mc{I},\mc{K})$: $E(i)=E(\{i\},\emptyset)$,
$E(ij,\mc{K})=E(\{i,j\},\mc{K})$.

For $\mb{c}\in\SetR^d$ and $r>0$, let
$\ball{\mb{c}}{r}=\{\mb{x}\in\SetR^d:\|\mb{x}-\mb{c}\|_2<r\}$, the open
ball centered at $\mb{c}$ with radius $r$. 

\begin{definition}[Relative interior, relative boundary]
  For a set $A\subset\SetR^d$,
  $\ri{A}\triangleq\{\mb{x}\in\SetR^d:\exists
  \epsilon>0,\ball{\mb{x}}{\epsilon}\cap \aff{A}\subset A\}$ is called
  the \emph{relative interior} of $A$, where $\aff{A}$ is the affine
  hull of $A$, and $\relbd{A}\triangleq
  \overline{A}\setminus\ri{A}$ is called the \emph{relative boundary} of $A$.
\end{definition}

\begin{proposition}
  \label{coidl}
For any polyhedral cone, its relative boundary is the union of all its facets.
\end{proposition}

\subsection{Matroid}
\label{sec:matroid-theory} 
There exist various cryptomorphic
definitions of a matroid. Here we discuss matroid theory from the
perspective of rank functions and regard matroids as special cases
of polymatroids. For a detailed treatment of matroid theory,
readers are referred to \cite{W76,O92}.

\begin{definition}
  A \emph{matroid} $M$ is an ordered pair $\{\mc{N}, \mb{r}\}$, where the \emph{ground set}
        $\mc{N}$ and the \emph{rank function} $\mb{r}$ satisfy the conditions
        that: for any
        $\mc{A},\mc{B}\subset\mc{N}$,
     \begin{itemize}
        \item $0\le\mb{r}(\mc{A})\le|\mc{A}|$ and
        $\mb{r}(\mc{A})\in\SetZ$.
        \item $\mb{r}(\mc{A})\leq \mb{r}(\mc{B}),\ \text{if }\mc{A}\subseteq\mc{B}$,
        \item $\mb{r}(\mc{A})+\mb{r}(\mc{B}) \geq \mb{r}({\mc{A}\cup\mc{B}})+\mb{r}({\mc{A}\cap\mc{B}})$.
     \end{itemize}
\end{definition}

For a matroid $M=\{\mc{N}, \mb{r}\}$, let $\mc{N}'\subset\mc{N}$ and let
$\mb{r}'$ be a set function which is the restriction of $\mb{r}$ on the
power set of $\mc{N}'$. Then $\{\mc{N}', \mb{r}'\}$ is called a
\emph{submatroid} of $M$. For $e\in\mc{N}$, if $\mb{r}(\{e\})=0$, $e$
is called a \emph{loop} of $M$.

Note that for a polymatroid $\h\in\pmr$, if
$\h(\mc{A})\in\SetZ$ and $\h(\mc{A})\le|\mc{A}|$, then $\h$ is a matroid.
Therefore, matroids are special cases of polymatroids. With a slight 
abuse of terminology, we do not differentiate a matroid and its
rank function. So $M, \rb(M)$ and $\rb$ may all denote the rank
function of $M$ when there is no ambiguity. 

For a matrix $D$ over a field $\mathbb{F}$, we can define a matroid by
letting the ground set $\N$ be
the set of columns of $D$ and the rank function $\rb(\mc{A})$ for $\mc{A}\subset\N$ be the
rank of the submatrix of $D$ whose columns are those in $\mc{A}$. It
can be checked that $\{\N, \rb\}$ is indeed a matroid. Such a matroid is called
\emph{representable} over $\mathbb{F}$ or $\mathbb{F}$-representable.

\begin{proposition}
\label{prop:rpr_entr}
 A representable matroid is almost entropic.
\end{proposition}
\textbf{Remark} It is not difficult to show that an
$\mathbb{F}$-representable matroid is entropic if the base $b$ of the
logarithm defining entropy is taken to be 
$|\mathbb{F}|$ (see for example\cite[Theorem 7.3]{YLCZ}).  If $b$ is not
taken to be $|\mathbb{F}|$, since $\overline{\Gamma^*_n}$ is a
cone,
it follows that a representable matroid is almost entropic.

\begin{definition}[Free expansion, factor \cite{N78},\cite{W86}]
\label{def:fr_ex}
  Let $\h\in \pmr$ be an integer-valued polymatroid. Consider a set
$\mc{M}$ with cardinality $m\triangleq\sum_{i\in\N}\h(\{i\})$ and
any mapping $\phi:\N\rightarrow 2^{\mc{M}}$ such that $\phi(i)$ has
the cardinality $\h(\{i\}),i\in\N$ and
$\phi(i)\cap\phi(j)=\emptyset$ for $i\neq j$. Then the \emph{free
expansion} $\g\in\Gamma_m$ of $\h$ by $\phi$ is defined by
\begin{equation}
  \label{eq:2}
  \g(\mc{A})=\min_{\mc{B}\subset \N}\big(\h(\mc{B})+|\mc{A}\setminus \phi(\mc{B})|\big),\
  \mc{A}\subset \mc{M}.
\end{equation}
It is said that $\g$ \emph{factors to} $\h$ under $\phi$ or $\h$ is a
\emph{factor} of $\g$.
\end{definition}

It can be checked that $\g$ is also an
integer-valued polymatroid and furthermore, can be proved to be a
matroid \cite{N78}.
The fact that the integer-valued polymatroid $\h\in\pmr$ is a factor of $\g\in\Gamma_m$ under
some $\phi$ can be understood as that any $i\in\N$ is split into
$\h(\{i\})$ ``independent'' elements $j\in\phi(i)\subset\mc{M}$. We
now offer an ``information theoretic'' interpretation.
Consider any $\mc{A}\subset\mc{M}$ such that $\mc{A}=\phi(\mc{B}')$
for some $\mc{B}'\subset\mc{N}$. Then for all $\mc{B}\subset\N$,
\begin{align*}
  \h(\mc{B})+|\mc{A}\setminus\phi(\mc{B})|
=\ &\h(\mc{B})+|\phi(\mc{B}')\setminus\phi(\mc{B})|\\
=\ &\h(\mc{B})+|\phi(\mc{B}' \setminus\mc{B})|\\
=\ &\h(\mc{B})+\sum_{i\in \mc{B}' \setminus\mc{B}}\h(\{i\})\\
\ge\ &\h(\mc{B})+\h(\mc{B}' \setminus\mc{B})\\
\ge\ &\h(\mc{B}\cup \mc{B}')\\
\ge\ & \h(\mc{B}'),
\end{align*}
where the inequalities above follow from the polymatroidal axioms
because $\h\in\Gamma_n$. Togethor with \eqref{eq:2}, we see that $\g(\mc{A})=\h(\mc{B}')$.
 So
the ``inverse operation'' of free expansion can be written as
\begin{equation}
  \label{eq:7}
  \h(\mc{B})=\g(\phi(\mc{B})),\
  \mc{B}\subset \mc{N}.
\end{equation}
Note that Definition \ref{def:fr_ex} allows $\h(\{i\})$ to be equal to
zero. 

\begin{proposition}
\label{prop:fr_ex} Let $\h\in\Gamma_n$ be integer-valued and a
factor of $\g\in\Gamma_m$ under some $\phi$. Then
$\g\in\overline{\Gamma^*_m}$ if and only if $\h\in\efc$.
\end{proposition}
\begin{proof}
The ``if'' part is proved in \cite[Theorem 4]{M07b}. For the ``only
if'' part, by the continuity of free expansion, it suffices to prove
that $\h\in\ef$ if $\g\in\Gamma^*_m$. Let $\g$ be the entropy
function of the random vector
$Y_\mc{M}=(Y_j)_{j\in\mc{M}}$. Now define $X_\N=(X_i)_{i\in\N}$ by
$X_i=(Y_j)_{j\in\phi(i)}$ for all $i\in\N$. Then it can be checked that $\h$ is the
entropy function of $X_\N$. 
\end{proof}

\begin{definition}[Uniform matroid]
  A \emph{uniform matroid} $U_{m,n}$ is a matroid with ground set $\N$ and
  rank function
  \begin{equation*}
    U_{m,n}(\mc{A})=\min\{m, |\mc{A}|\}
  \end{equation*}
  for any $\mc{A}\subset\N$.
\end{definition}

Since uniform matroids are representable(e.g. by the Vandermonde matrix), they are almost entropic
by Proposition \ref{prop:rpr_entr}. For a matroid with ground set
$\N$, it can be shown it is also almost entropic if its submatroid on
some $\mc{M}\subset\N$ is uniform and elements in $\N\setminus\mc{M}$
are all loops. 
By Proposition \ref{prop:fr_ex},
a polymatroid which is a factor of a uniform matroid is almost
entropic.

\subsection{Partition}
\label{sdsb}
\begin{definition}[Partition of a set]
  For a set $\mc{S}$ and an index set $\mc{I}$, a collection $p=\{\mc{S}_i,i\in\mc{I}\}$ of
  disjoint subsets of $\mc{S}$ such that $\mc{S}=\cup_{i\in\mc{I}}\mc{S}_i$ and
  $\mc{S}_i\neq\emptyset$ is called a \emph{partition} of $\mc{S}$.  The sets $\mc{S}_i$ are called \emph{blocks}
  of the partition $p$.
\end{definition}

For two partitions of $\mc{S}$, $p_1=\{\mc{S}^{(1)}_i,i\in\mc{I}_1\}$ and $p_2=\{\mc{S}^{(2)}_i,i\in\mc{I}_2\}$,
$p_1\le p_2$ if for any $\mc{S}^{(2)}_i,i\in\mc{I}_2$, there
exists $\mc{J}_i\subset\mc{I}_1$ such that
$\mc{S}^{(2)}_i=\cup_{j\in\mc{J}_i}\mc{S}^{(1)}_j$. 
We say that
$p_1$ is a \emph{refinement} of $p_2$, $p_1$
is \emph{finer than} $p_2$ or $p_2$ is \emph{coarser than} $p_1$. For
$p_1,p_2\in\mc{P}_n$, $p_1<p_2$ if $p_1\le p_2$ and $p_1\neq p_2$. 

In this subsection, we discuss some properties of a partition of a
finite set that will be used later in this paper. Consider a partition
of $\N=\{1,\cdots,n\}$ with $t\le n$ blocks. Such a
partition $\{\N_1,\cdots,\N_t\}$ is also called a $t$-partition of $\N$. 
The set of all $t$-partitions of $\N$ is denoted by
$\mc{P}_{t,n}$. The cardinality of $\mc{P}_{t,n}$ is called the
Stirling number (of the second kind) with respect to $t$ and $n$. 
Let $\mc{P}_n=\cup^n_{t=1}\mc{P}_{t,n}$ be the set of
all partitions of $\N$. The cardinality of $\mc{P}_n$ is called
the Bell number with respect to $n$\cite{LW92}. 

For $p\in\mc{P}_{t,n}$ and $\mc{A}\subset\N$,
we call $\bs{\lambda}_{\mc{A},p}\triangleq(|\mc{A}\cap\N_1|,\cdots,|\mc{A}\cap\N_t|)$
the \emph{partition vector} of $\mc{A}$ under $p$. Let
$\bs{\lambda}_{\mc{A},p}(i)=|\mc{A}\cap\N_i|$ denote the $i$-th entry of $\bs{\lambda}_{\mc{A},p}$. 
In particular, when
$\mc{A}=\N$, we call $\bs{\lambda}_{\N,p}$ or simply $\bs{\lambda}_p$ the
partition vector of $p$. 

The set of all partitions of $\N$, denoted by
$\mc{P}_n$, is a partially ordered set with order ``$\le$''. It can be shown
that the partially ordered 
set $\mc{P}_n$ is a lattice with the set of all singletons
$\{\{i\}:i\in\N\}$ being the least element and the 1-partition $\{\N\}$ as the
greatest element. In a lattice $\mc{L}$, for $l_1,l_2\in\mc{L}$, $ l_2$
\emph{covers} $l_1$, if $l_1< l_2$ and for any
$l\in\mc{L}$ such that $l_1\le l\le l_2$, either $l=l_1$ or $l=l_2$. For
$p_1,p_2\in\mc{P}_n$, it can be shown that $p_2$ covers $p_1$
if and only if $p_1\le p_2$ and one of blocks of $p_2$ is the union of two blocks
of $p_1$ and all other blocks of $p_2$ are also blocks of
$p_1$. Therefore, any $p\in\mc{P}_{t,n}$ is covered by some
$p'\in\mc{P}_{t-1,n}$ for $2\le t\le n$ and covers some
$p''\in\mc{P}_{t+1,n}$ for $1\le t\le n-1$.


\begin{definition}[Partition of a number]
  For a positive integer $n$, a vector $\mb{n}=[n_1,\cdots,n_t]$ with
  $0<n_i\le n_j$ for $1\le i<j\le t$ such that $n=\sum^t_{i=1}n_i$ is called a \emph{partition} of $n$. 
\end{definition}

The number of the partitions of $n$ is called the partition function with
respect to $n$
\cite{LW92}. For $p\in\mc{P}_n$, let $\mb{n}_p$ be a vector whose
entries are a nondecreasing rearrangement of $\bsl_p(i),\ i=1,\cdots,t$.
It can be seen that $\mb{n}_p$
is a partition of $n$.

\subsection{Group action}
\label{hidmx}



To study symmetry, group theory is a regular tool. For the
basics of group theory, readers are referred to \cite{Rotman06}. In
this subsection, we discuss group actions and how they can be used to
study the symmetries in the entropy space $\hsp_n$. For a detailed
introduction to group actions, readers are referred to \cite{DM96}.

A bijection $\sigma:\N\rightarrow\N$ is called a \emph{permutation} of
$\N$. The set $\Sigma_n$ of all permutations of $\N$ is a
group with order $n!$ taking composition as its group operation. The group
$\Sigma_n$ is called the \emph{symmetric group} on $\N$.
For $p=\{\N_1,\cdots,\N_t\}\in\mc{P}_{n}$,
define 
\begin{align}
\label{feqbs}
  \Sigma_p= &\{\sigma\in\Sigma_n: \sigma(j)\in\N_i,\  j\in\N_i,\ i=1,\cdots,t\},
\end{align}
the set of permutations that keep the members of a block in the same block. It can be checked that $\Sigma_p$ is a 
subgroup of $\Sigma_n$ with order $\prod^t_{i=1}n_i!$. 
When $p=\{\N\}$, $\Sigma_{\N}$ coincides with $\Sigma_n$. Following the
notation simplification in Section \ref{bafdg}, we write $\Sigma_{\N}$ as $\Sigma_n$.

\begin{definition}[Group action]
  For a set $\mc{S}$, a group $\Sigma$ \emph{acts} on
  $\mc{S}$ if there exists a function $\Sigma\times \mc{S}\rightarrow
  \mc{S}$,  called an action, denoted by $(\sigma, s)\mapsto
  \sigma s$, such that 
  \begin{enumerate}
  \item $(\sigma_1 \sigma_2)s=\sigma_1(\sigma_2 s)$ for all
    $\sigma_1,\sigma_2\in \Sigma$ and $s\in\mc{S}$;
    \item $1 s=s$ for all $s\in \mc{S}$, where $1$ is the identity of $\Sigma$.
  \end{enumerate}
\end{definition}

For any $\sigma\in\Sigma_n$, define
$\sigma':\hsp_n\rightarrow\hsp_n$ by
\begin{equation}
  \label{eq:1}
  \sigma'(\h)(\mc{A})=\h(\sigma(\mc{A})),\ \mc{A}\subset\N.
\end{equation}
It can readily be verified that $\sigma: 2^\N\rightarrow 2^\N$
(defined by $\sigma(\mc{A})=\{\sigma(i):i\in\mc{A}\}$) is a
bijection and so a permutation of $2^\N$. It then follows that
$\sigma'(\h)$ is obtained from $\h$ by permutating the components of
$\h$, and hence $\sigma'$ is a bijection. With a slight abuse of
notation, we write $\sigma'(\h)$ as $\sigma(\h)$. 


It can be checked that $\sigma(\h)$ defines a group action $\Sigma_n$
on $\hsp_n$. By restricting the
action on a subgroup $\Sigma_p$, we obtain the group action $\Sigma_p$
on $\hsp_n$.

\begin{definition}[Orbit]
\label{uisdo}
  If group $\Sigma$ acts on $\mc{S}$, then for $s\in\mc{S}$, the
  \emph{orbit}
  of the action of $s$ is defined by 
  \begin{equation*}
    \mc{O}_{\Sigma}(s)=\{\sigma s:\sigma\in \Sigma\}.
  \end{equation*}
\end{definition}

It can be verified that 1) $s\in\mc{O}_\Sigma(s)$, 2) if $s_2\in\mc{O}_\Sigma(s_1)$, then
$s_1\in\mc{O}_\Sigma(s_2)$ and 3) if $s_2\in\mc{O}_\Sigma(s_1),s_3\in\mc{O}_\Sigma(s_2)$, then
$s_3\in\mc{O}_\Sigma (s_1)$. Therefore, orbits of an action defines an equivalence
relation on $\mc{S}$ which implies the following proposition. 
\begin{proposition}\cite[Proposition 2.142]{Rotman06}
\label{dfoib}
  If group $\Sigma$ acts on a set $\mc{S}$, then the orbits induced by the action
  of $\Sigma$ on $\mc{S}$ form
  a partition of $\mc{S}$.
\end{proposition}

For $p\in\mc{P}_n$, consider the action of $\Sigma_p$ on
$\hsp_n$. For any $\h\in\hsp_n$, we call the orbit
$\mc{O}_{\Sigma_p}(\h)$, or $\mc{O}_p$ in short,
a \emph{$p$-orbit}. If $\h_1$ and $\h_2$ are in the same
$p$-orbit, i.e., $\h_2\in\mc{O}_p(\h_1)$, we say that $\h_1$ and $\h_2$ are \emph{$p$-equivalent}.

\begin{definition}
  If a group $\Sigma$ acts on $\mc{S}$, for
  $\mc{T}\subset\mc{S}$, the \emph{fixed set} of $\mc{T}$ is defined by
  \begin{equation*}
    \fix_{\Sigma} (\mc{T})=\{s\in\mc{T}:\sigma s=s, \forall \sigma\in\Sigma\}.
  \end{equation*}
When $\mc{T}=\mc{S}$, we call $\fix_{\Sigma} (\mc{S})$
the fixed set of the action.
\end{definition}

If $\mc{T}\subset\mc{S}$, then $\fix_{\Sigma}
(\mc{T})=\fix_{\Sigma}(\mc{S})\cap\mc{T}$. Furthermore, for any $s\in
\fix_{\Sigma}(\mc{S})$, the singleton 
$\{s\}$ forms an orbit of the action.

For the action of $\Sigma_p$ on $\hsp_n$ and $T\subset\hsp_n$,
$\fix_{\Sigma_p} (T)$ is denoted by $\fix_{p} (T)$ for
simplicity. For $p=\{\N_1,\cdots,\N_t\}\in\mc{P}_{n}$, it can be checked that
\begin{equation}
\label{symt_c}
  \fix_{p}(\hsp_n)=\{\h\in\hsp_n:\ \h(\mc{A})=\h(\mc{B})
  \text{ if } \bs{\lambda}_{\mc{A},p}=\bs{\lambda}_{\mc{B},p}\}.
\end{equation}
Therefore, the set $\fix_{p}(\hsp_n)$ is a subspace of $\hsp_n$. We
denote this subspace by $S_p$ and call it the \emph{$p$-symmetrical
subspace}. Then $\fix_{p}(\Gamma_n)=\Gamma_n\cap S_p$ and
$\fix_{p}(\Gamma^*_n)=\Gamma^*_n\cap S_p$. 

\begin{proposition}
\label{o238}
  If a group $\Sigma$ acts on a set $\mc{S}$, then $\Sigma$ also acts on
  $2^{\mc{S}}$, where $\sigma \mc{T}=\{\sigma s: s\in\mc{T}\}$ for
  $\mc{T}\subset\mc{S}$. 
\end{proposition}
\begin{proof}
  This proposition can be readily proved by checking the definition of
  group action.
\end{proof}

For the induced group action on the power set as in Proposition \ref{o238}, we call an orbit of the
action a \emph{setwise orbit} of the original action. Specifically,
for any
$\mc{T}\subset\mc{S}$, the setwise orbit of $\mc{T}$ is given by
\begin{equation*}
  \mc{O}_\Sigma(\mc{T})=\{\sigma\mc{T}:\sigma\in\Sigma\}.
\end{equation*}
To distinguish a setwise orbit from an orbit of the
original action (cf. Definition \ref{uisdo}), we will refer to the
latter as a
\emph{pointwise orbit}.
Note that a setwise orbit of the original group action is a pointwise
orbit of the induced group action on the power set.
By Proposition \ref{dfoib}, all of the
setwise orbits form a partition of the power set.

For action $\Sigma_p$ on $\hsp_n$, we also call the setwise orbits of
the action $p$-orbits and denoted them by $\mc{O}_p$ if there is no
ambiguity. 
For $T_1, T_2\subset \hsp_n$, if they are in the same
$p$-orbit, i.e., $T_2\in\mc{O}_p(T_1)$, we also say that they are
$p$-equivalent. 

\begin{definition}[Invariance]
   If a group $\Sigma$ acts on a set $\mc{S}$, a subset
   $\mc{T}\subset\mc{S}$ is called \emph{invariant} if
   \begin{equation*}
     \mc{T}=\sigma(\mc{T}) \text{ for any }\sigma\in\Sigma.
   \end{equation*}
\end{definition}

It can be checked that $\mc{T}$ is invariant if and only if $\mc{T}$ is in the fixed set of the induced action
on the power set. If $\mc{T}$ is invariant, then
$\mc{T}=\sigma(\mc{T})=\{\sigma(s):s\in\mc{T}\}=\cup_{s\in\mc{T}}\sigma(s)$,
i.e., $\mc{T}$ is the union of pointwise orbits. 
 As any $s\in\fix_{\Sigma}(\mc{S})$ itself forms a pointwise
orbit, any subset of the fixed set is invariant. Note that an invariant set
$\mc{T}$ itself forms a setwise orbit.

For the action of $\Sigma_p$ on $\hsp_n$, if $T\subset\hsp_n$ is invariant,
we say that it is \emph{$p$-invariant}. For a $p$-invariant set, the
following propositions are straightforward.

\begin{proposition}
\label{sioab}
  $\h\in\hsp_n$ is $p$-invariant if and only if $\h\in S_p$.
\end{proposition}

\begin{proposition}
\label{badlm}
  If $T\subset S_p$, then $T$ is $p$-invariant.
\end{proposition}

Obviously, $\Gamma_n$, $\Gamma^*_n$ and $\overline{\Gamma^*_n}$ are all
$p$-invariant for any $p\in\mc{P}_n$.

\section{Problem Formulation}
\label{bafdg}

Let $\Sigma_n$ be the symmetric group on $\N$. For any
$\sigma\in\Sigma_n$ and any $\h\in \hsp_n$, define
\begin{equation}
  \label{eq:1}
  \sigma(\h)(\mc{A})=\h(\sigma(\mc{A})),\ \mc{A}\subset\N,
\end{equation}
i.e., $\sigma(\h)$ is obtained from $\h$ by permuting the components
$\mc{A}$ of
$\h$ to $\sigma(\mc{A})$ for all $\mc{A}\subset\N$. It can be checked that $\sigma(\h)$ defines a group action $\Sigma_n$
on $\hsp_n$. By restricting the
action on any subgroup $\Sigma$ of $\Sigma_n$, we obtain the group action $\Sigma$
on $\hsp_n$. For any $T\subset\hsp_n$, let
\begin{equation*}
  \fix_\Sigma(T)=\{\h\in T: \sigma(\h)=\h,\quad \forall \sigma\in \Sigma\}
\end{equation*}
be the \emph{fixed set} of $T$ for this action.
Let
\begin{align}
\label{feqbs}
  \Sigma_p= &\{\sigma\in\Sigma_n: \sigma(j)\in\N_i,\  j\in\N_i,\ i=1,\cdots,t\}
\end{align}
be the set of permutations that keep the members of a block in the same
block. Obviously, $\Sigma_p$ is a subgroup of $\Sigma_n$. It can be
readily seen that
\begin{equation}
\label{symt_cc}
  \fix_{\Sigma_p}(\hsp_n)=\{\h\in\hsp_n:\ \h(\mc{A})=\h(\mc{B}),
  \text{ if } \bs{\lambda}_{\mc{A},p}=\bs{\lambda}_{\mc{B},p}\}.
\end{equation}
We call $\fix_{\Sigma_p}$ the $p$-symmetrical subspace of $\hsp_n$ and write it as $\fix_p$
for simplicity. A vector $\h\in\fix_p$ is called $p$-symmetrical.
Then naturally, we define \emph{$p$-symmetrical entropy region}
       \begin{equation*}
            \Psi^*_p\triangleq\fix_{p}(\Gamma^*_n)=\Gamma^*_n\cap \fix_p
        \end{equation*}
and \emph{$p$-symmetrical polymatroidal region}
        \begin{equation*}
            \Psi_p\triangleq\fix_{p}(\Gamma_n)=\Gamma_n\cap \fix_p,
        \end{equation*}
respectively.

For $1\le n\le 3$, $\overline{\Psi^*_p}=\Psi_p$ for any
$p\in\mc{P}_n$ since $\overline{\Gamma^*_n}=\Gamma_n$ and
$\overline{\Psi^*_p}=\overline{\Gamma^*_n}\cap \fix_p$ by Theorem
\ref{lem:psi2} (See Subsection \ref{op1512}).

\begin{theorem}
\label{bfdp}
  For $n\ge 4$ and any $p\in\mc{P}_n$, 
  \begin{equation}
    \label{meq:1}
    \overline{\Psi^*_p}=\Psi_p,
  \end{equation}
  if and only if $p=\{\N\}$ or $p=\{\{i\},\N\setminus\{i\}\}$ for some
  $i\in\N$.
\end{theorem}

Theorem \ref{bfdp} says that $\overline{\Psi^*_p}$,
is completely characterized by
Shannon-type information inequalities if and only if $p$ is the
$1$-partition or a $2$-partition with one of its blocks being a
singleton. This theorem, the main result of this paper, will be established
through Corollary \ref{rel11} and Theorems \ref{qwekl}, \ref{fbad} and
\ref{qpads}. Specifically, 
we will prove \eqref{meq:1}
in Corollary \ref{rel11} and Theorem \ref{qwekl} for the cases $p=\{\N\}$ and $p=\{\{i\},\N\setminus\{i\}\}$,
respectively. In Theorems \ref{fbad} and \ref{qpads}, we will prove
that $\overline{\Psi^*_p}\subsetneq\Psi_p$ for all other cases. 

To facilitate our discussion in the rest of the paper, we now
introduce a simplification of the notations.
For $p\in\mc{P}_n$, let $\mb{n}_p$ be a vector whose
entries are a nondecreasing rearrangement of $\bsl_p(i),\ i=1,\cdots,t$.
It can be seen that $\mb{n}_p$
is a (number) partition of $n$. It can be checked that partitions $p$ with
same $\mb{n}_p$ form a equivalence class.
Therefore, for $p_1,p_2\in\mc{P}_n$, we say they are equivalent if
$\mb{n}_{p_1}=\mb{n}_{p_2}$.  For
each equivalence class determined by $\mb{n}_p=[n_1,\cdots,n_t]$, we choose a 
representative $\{\N_1,\cdots,\N_t\}$ such that
$\N_1=\{1,\cdots,n_1\},\N_2=\{n_1+1,\cdots,n_1+n_2\},\cdots,\N_t=\{\sum^{t-1}_{i=1}n_i+1,\cdots,\sum^{t}_{i=1}n_i\}$
and let $\mc{P}^*_n$ be the set of these representatives.
If $p_1$ and $p_2$ are equivalent, 
the characterization of $\Psi^*_{p_2}$ can be
obtained by the characterization of $\Psi^*_{p_1}$ by permuting
the indices. For example, consider $p_1=\{\{1\},\{2,3\}\}$ and
$p_2=\{\{2\},\{1,3\}\}$. Then $\mb{n}_{p_1}=\mb{n}_{p_2}=[1,2]$. 
If $\h\in \fix_{\{\{1\},\{2,3\}\}}$ is the entropy function of $\{X_1, X_2,
X_3\}$, then there exists $\h'\in \fix_{\{\{2\},\{1,3\}\}}$ which is the
entropy function of $\{X'_1, X'_2,
X'_3\}$ where $X'_1=X_2,X'_2=X_1$ and $X'_3=X_3$. Therefore, for each such
equivalence class of
partitions, we only need to consider one partition in the equivalence
class. Without loss of generality, we consider only those
$p\in\mc{P}^*_n$ and for the purpose of our discussion,
such a $p$ will be represented by $\mb{n}_p$ for simplicity. For example, partition $\{\{1\},\{2,3\}\}$ will be
represented by $[1,2]$. For this spirit, we write
$\Psi_{\{\{1\},\{2,3\}\}}$ as $\Psi_{[1,2]}$, which is further
simplified as $\Psi_{1,2}$. Then Theorem \ref{bfdp} can be restated
as for $n\ge 4$ and any $p\in\mc{P}^*_n$,
$\overline{\Psi^*_p}=\Psi_p$ if and only if $p=[n]$ or $p=[1,n-1]$.





\section{Symmetrical properties of $\Gamma_n$ and $\Gamma^*_n$}
\label{yuber}

\subsection{$p$-equivalent facets of $\Gamma_n$}
\label{aaa31}



For $p\in\mc{P}_n$, the orbit of the action $\Sigma_p$ on
$\hsp_n$, i.e., $\{\sigma(\h): \sigma\in\Sigma_p\}$
for some $\h\in\hsp_n$ is called a \emph{pointwise} $p$-orbit. Similarly,
$\{\sigma(T): \sigma\in\Sigma_p\}$
for some $T\subset\hsp_n$ is called a setwise $p$-orbit. Either of
them is called a $p$-orbit. 
For $\h_1,\h_2\in\hsp_n$ ($T_1, T_2\subset\hsp_n$), they are called
$p$-equivalent if they are in the same $p$-orbit.

\begin{lemma}
\label{iwds}
  For $p\in\mc{P}_n$, let $T_1,T_2\subset\mc{H}_n$ be $p$-equivalent.
  Then $T_1$ is a face (facet) of $\Gamma_n$ if and only
  if $T_2$ is a face (facet) of $\Gamma_n$. 
\end{lemma}
\begin{proof}
  Since $\dim{\Gamma_n}=2^n-1<\dim\hsp_n=2^n$, 
each face $F$ of $\Gamma_n$
  can be written as $F=\Gamma_n\cap P$ for some supporting hyperplane
  $P$ of $\Gamma_n$. 

Let $T_1,T_2\subset T$ be $p$-equivalent. Then $T_2=\sigma(T_1)$ for
some $\sigma\in \Sigma_p$.  
Now if $T_1$ is a face of $\Gamma_n$, there exists a supporting hyperplane $P$ of $\Gamma_n$ such that
  $T_1=\Gamma_n\cap P$. Then
  $T_2=\sigma(T_1)=\sigma(P\cap\Gamma_n)
  =\sigma(P)\cap\sigma(\Gamma_n)=\sigma(P)\cap\Gamma_n$. As
  $\sigma(P)$ is also a supporting hyperplane of $\Gamma_n$,
  $T_2$ is a face of $\Gamma_n$. The only if part is also true since $T_1=\sigma^{-1}(T_2)$.
  
  Furthermore, as $T_2=\sigma(T_1)$ has the same dimension as $T_1$,
  $T_2$ is a facet if and only if $T_1$ is a facet.
\end{proof}

By Lemma \ref{iwds}, for $p\in\mc{P}_n$, all the faces of $\Gamma_n$ are partitioned
into $p$-orbits, and some of them are families of facets of
$\Gamma_n$ which play a more important role because they correspond to
the elemental inequalities $E(\mc{I}_i,\mc{K}_i)$ defining $\Gamma_n$. The collection of all
$p$-orbits of facets is denoted by $\mf{E}_p$.

\begin{lemma}
\label{oiabd}
  For $p\in\mc{P}_n$ and $E_i\triangleq E(\mc{I}_i,\mc{K}_i)\in\mc{E}_n,
  i=1,2$, the following three statements are equivalent:
  \begin{enumerate}
    \item $E_1$ and $E_2$ are $p$-equivalent;
    \item there exists $\sigma\in\Sigma_p$ such that
  $\mc{I}_1=\sigma(\mc{I}_2)$ and $\mc{K}_1=\sigma(\mc{K}_2)$;
    \item $\bs{\lambda}_{\mc{I}_1,p}=\bs{\lambda}_{\mc{I}_2,p}$ and $\bs{\lambda}_{\mc{K}_1,p}=\bs{\lambda}_{\mc{K}_2,p}$.
  \end{enumerate}
\end{lemma}
\begin{proof}
  We first prove that the first two statements are equivalent. If $E_1$ and $E_2$ are $p$-equivalent,
  there exists $\sigma\in\Sigma_p$ such that $\sigma(E_1)=E_2$. 
 For $|\mc{I}_1|=1$, i.e., $E_1=E(i)$ and $\mc{I}_1=\{i\}$ for some
 $i\in\N$, since for any $\h\in E_2$, $\sigma^{-1}(\h)\in E_1$, we have 
 \begin{align*}
   E_2&\ =\{\h\in\Gamma_n:\sigma^{-1}(\h)(\N)=\sigma^{-1}(\h)(\N\setminus\{i\})\}\\
   &\ =\{\h\in\Gamma_n:\h(\sigma^{-1}(\N))=\h(\sigma^{-1}(\N\setminus\{i\}))\}\\
   &\ =\{\h\in\Gamma_n:\h(\N)=\h(\N\setminus\sigma^{-1}(\{i\}))\}.
 \end{align*}
Hence, $\mc{I}_2=\sigma^{-1}(\{i\})=\sigma^{-1}(\mc{I}_1)$ or
$\mc{I}_1=\sigma(\mc{I}_2)$. In this case,
$\mc{K}_1=\sigma(\mc{K}_2)=\emptyset$.
For $|\mc{I}_1|=2$, i.e., $E_1=E(ij,\mc{K})$ and $\mc{I}_1=\{ij\}$ for some
 $i,j\in\N$, similarly, we have
\begin{align*}
   E_2&\ =\{\h\in\Gamma_n:\sigma^{-1}(\h)(\{i\}\cup\mc{K}_1)+\sigma^{-1}(\h)(\{j\}\cup\mc{K}_1)\\
&\quad \quad \quad \quad \quad \quad=\sigma^{-1}(\h)(\mc{K}_1)+\sigma^{-1}(\h)(\{ij\}\cup\mc{K}_1)\}\\
   &\ =\{\h\in\Gamma_n:\h(\sigma^{-1}(\{i\})\cup
   \sigma^{-1}(\mc{K}))+\h(\sigma^{-1}(\{j\})\cup
   \sigma^{-1}(\mc{K}))\\
&\ =\h(\sigma^{-1}(\mc{K}))+\h(\sigma^{-1}(\{ij\})\cup
   \sigma^{-1}(\mc{K})).
 \end{align*}
Hence, $\mc{I}_1=\sigma(\mc{I}_2)$ and
$\mc{K}_1=\sigma(\mc{K}_2)$. As each step above is invertible, the
second statement also implies the first one.

We now prove that the second and third statements are equivalent. If there exists $\sigma\in\Sigma_p$ such that
  $\mc{I}_1=\sigma(\mc{I}_2)$ and $\mc{K}_1=\sigma(\mc{K}_2)$,
then for any block $\N_l$ of $p$,
  $|\mc{I}_1\cap\N_l|=|\sigma(\mc{I}_2)\cap\N_l|=|\mc{I}_2\cap\N_l|$,
  where the second equality holds since
  $\sigma\in\Sigma_p$. Therefore,
  $\bs{\lambda}_{\mc{I}_1,p}=\bs{\lambda}_{\mc{I}_2,p}$. Similarly,
  $\bs{\lambda}_{\mc{K}_1,p}=\bs{\lambda}_{\mc{K}_2,p}$.

  On the other hand, assume that
  $\bs{\lambda}_{\mc{I}_1,p}=\bs{\lambda}_{\mc{I}_2,p}$ and
  $\bs{\lambda}_{\mc{K}_1,p}=\bs{\lambda}_{\mc{K}_2,p}$.
Since $\mc{I}_i$ and $\mc{K}_i$ are
  disjoint for $i=1,2$, given any block $\N_l$ of $p$,
  $\N_l\cap\mc{I}_i$, $\N_l\cap\mc{K}_i$ and
  $\N_l\setminus(\mc{I}_i\cup\mc{K}_i)$ are disjoint.
As $|\mc{I}_1\cap\N_l|=\bs{\lambda}_{\mc{I}_1,p}(l)=\bs{\lambda}_{\mc{I}_2,p}(l)=|\mc{I}_2\cap\N_l|$ and
  $|\mc{K}_1\cap\N_l|=\bs{\lambda}_{\mc{K}_1,p}(l)=\bs{\lambda}_{\mc{K}_2,p}(l)=|\mc{K}_2\cap\N_l|$,
  then there exists 
  $\sigma\in\Sigma_p$ such that
  $\mc{I}_1\cap\N_l=\sigma(\mc{I}_2\cap\N_l)$ and
  $\mc{K}_1\cap\N_l=\sigma(\mc{K}_2\cap\N_l)$ for any $l$.
Then it can be seen for such a $\sigma$, $\mc{I}_1=\sigma(\mc{I}_2)$ and $\mc{K}_1=\sigma(\mc{K}_2)$. 
\end{proof}

For a facet $E(\mc{I},\mc{K})\in\mc{E}_n$ of $\Gamma_n$, denote the $p$-orbit it belongs to by
$\mc{E}_p([\bs{\lambda}_{\mc{I},p},\bs{\lambda}_{\mc{K},p}])$. Note that
$\mc{E}_p([\bs{\lambda}_{\mc{I},p},\bs{\lambda}_{\mc{K},p}])$ is well-defined in
light of Lemma \ref{oiabd}.
Let $\mf{N}_p$ be the set
of all possible distinct pairs of
$\bs{\lambda}=[\bs{\lambda}_{\mc{I},p},\bs{\lambda}_{\mc{K},p}]$. Then
$\mf{E}_p=\{\mc{E}_p(\bs{\lambda}):\bs{\lambda}\in\mf{N}_p\}$. Evidently,
$|\mf{E}_p|=|\mf{N}_p|$. Let $\mc{G}_p$ be the collection of all facets of $\Psi_p$.

\begin{example}
\label{jaodm}
  Let $p=[n]$. For any $i\in\N$, facet $E(i)=E(\{i\},\emptyset)$
  with $[\bs{\lambda}_{\mc{I},p},\bs{\lambda}_{\mc{K},p}]=[(1),(0)]$. 
 It follows from Lemma \ref{oiabd} that
  all $E(i)$ are in the $[n]$-orbit $\mc{E}_{[n]}([(1),(0)])$.
  For distinct $i,j\in\N$, $\mc{K}\subset\N\setminus\{i,j\}$
  with the same $k\triangleq|\mc{K}|$, facets $E(ij,\mc{K})=E(\{i,j\},\mc{K})$
  with $[\bs{\lambda}_{\mc{I},p},\bs{\lambda}_{\mc{K},p}]=[(2),(k)]$. So by Lemma
  \ref{oiabd}, all such facets are in the same $[n]$-orbit
  $\mc{E}_{[n]}([(2),(k)])$. It can be seen that 
$\mf{N}_{[n]}=\{[(1),(0)]\}\cup\{[(2),(k)]:k=0,\cdots,n-2\}$ with cardinality $n$
and $\mf{E}_{[n]}=\{\mc{E}_{[n]}(\bs{\lambda}):\bs{\lambda}\in\mf{N}_{[n]}\}$ is a partition
of $\mc{E}_n$.
\hfill\QQQ
\end{example}

\begin{theorem}
\label{bafoi}
  Every $p$-orbit of facets of $\Gamma_n$ corresponds to a
  facet of $\Psi_p$, i.e., mapping $\omega_p:E\mapsto E\cap \fix_p$ satisfies
  \begin{enumerate}
    \item $\omega_p(E_1)=\omega_p(E_2)$ if and only if $E_1$ and
      $E_2$ are $p$-equivalent and
    \item $\omega_p$ is a surjection from $\mc{E}_n$ onto $\mc{G}_p$.
  \end{enumerate}    
\end{theorem}
\begin{proof} See Appendix A. 
\end{proof}

By \eqref{symt_cc}, $\fix_p$ is an $n_p\triangleq\prod^t_{i=1}(n_i+1)$
dimensional subspace of $\hsp_n$, where $n_i\triangleq\bs{\lambda}_p(i)=|\N_i|$.
For notational convenience, in the subsequent discussions, we map
$\fix_p$ to an $n_p$-dimensional Euclidean space as follows.
Let
$\mc{M}_p=\{(k_1,\cdots,k_t):k_i\in\{0,1,\cdots,n_i\},i=1,\cdots,t\}$.
For $\h\in \fix_p$, let $s_{k_1,\cdots,k_t}$ be the common value taken by
$\h(\mc{A})$ for all $\mc{A}$ such that $\bs{\lambda}_{\mc{A},p}=(k_1,\cdots,k_t)$.
For a fixed $(k_1,\cdots,k_t)\in\mc{M}_p$, the total number of $\mc{A}$
such that $\h(\mc{A})=s_{k_1,\cdots,k_t}$ is precisely $\prod^{t}_{i=1}\binom{n_i}{k_i}$.
For every $\h\in \fix_p$, let
\begin{equation}
  \label{eq:3}
  \s(\h,p)=(s_{k_1,\cdots,k_t})_{(k_1,\cdots,k_t)\in\mc{M}_p},
\end{equation}
such that $s_{k_1,\cdots,k_t}=\h(\mc{A})$ if 
$\bs{\lambda}_{\mc{A},p}=(k_1,\cdots,k_t)$. Note that $\s(\h,p)$ is properly
defined because for all $\h\in \fix_p$, $\h(\mc{A})=\h(\mc{B})$ if
$\bs{\lambda}_{\mc{A},p}=\bs{\lambda}_{\mc{B},p}$. When there is no
ambiguity, $\s(\h,p)$ is simply written as $\s$.

For $E(\mc{I},\mc{K})\in\mc{E}_n$,
when
$|\mc{I}|=1$ or $\mc{I}=\{i\}$ for some $i\in\N$, we must have $\mc{K}=\emptyset$ which implies that
$\bs{\lambda}_{\mc{K},p}$ must be equal to $\mb{0}_t$, a zero vector with dimension
$t$. In this case, the number of 
  possible pairs of
  $[\bs{\lambda}_{\mc{I},p},\bs{\lambda}_{\mc{K},p}]$ is equal to
  on the numbers of possible values of $\bs{\lambda}_{\mc{I},p}$,
  because both $p$ and $\mc{K}$ are fixed. If
  $i\in\N_l$, the $l$-th block of $p$, it can
  be seen that $\bs{\lambda}_{\mc{I},p}=\mb{1}_{t}(l)$, a $t$-vector
  with the $l$-th entry equal to 1 and other
entries equal to 0. Hence there are $t$
possible such values. By Theorem \ref{bafoi}, for any
$E\in\mc{E}_p([\mb{1}_{t}(l),\mb{0}_t])$, $E\cap \fix_p$ are all
the same, i.e., 
\begin{align}
  E\cap \fix_p&=\{\h\in\Psi_p: s_{\bs{\lambda}_p}=
  s_{\bs{\lambda}_{\N\setminus\{i\},p}} \}\nonumber\\
\label{soidv}
&=\{\h\in\Psi_p: s_{n_1,\cdots,n_t}=
  s_{n_1,\cdots,n_l-1,\cdots,n_t}. \}
\end{align}
Note that $\h\in \fix_p$ and by \eqref{eq:3}, $s_{n_1,\cdots,n_t}$, $s_{n_1,\cdots,n_l-1,\cdots,n_t}$ are entries of $\s(\h,p)$.

When $|\mc{I}|=2$ or $\mc{I}=\{i,j\}$ for some distinct $i,j\in\N$, if
$i,j$ are in different blocks, i.e., $i\in\N_{l_1}$, $j\in\N_{l_2}$ or
$j\in\N_{l_1}$, $i\in\N_{l_2}$ with some $1\le l_1<l_2\le t$, 
$\bs{\lambda}_{\mc{I},p}=\mb{1}_{t}(l_1,l_2)$, a $t$-vector with the $l_1$-th
and $l_2$-th entries equal to 1 and other
entries equal to 0; else if $i,j$ are in the same $\N_l$,
$\bs{\lambda}_{\mc{I},p}=\mb{2}_{t}(l)$, a $t$-vector with the $l$-th
entry equal to 2 and other
entries equal to 0. For $\bs{\lambda}_{\mc{I},p}=\mb{1}_{t}(l_1,l_2)$,
the number of possible values of $\bs{\lambda}_{\mc{K},p}=(k_1,\cdots,k_t)$
is $\frac{n_{l_1}n_{l_2}}{(n_{l_1+1})(n_{l_2+1})}\prod^t_{m=1}(n_m+1)$.
It is because $k_m$ can be equal to $0,\cdots,n_m-1$ if $m=
l_1,l_2$, and $k_m$ can be equal to $0,\cdots,n_m$,
otherwise. In other words, $\bs{\lambda}_{\mc{K},p}$ can be equal to any
$(k_1,\cdots,k_t)\in\mc{M}_p$ such that $k_{l_1}\neq n_{l_1}$ and $k_{l_2}\neq n_{l_2}$. For $E\in\mc{E}_p([\mb{1}_{t}(l_1,l_2),(k_1,\cdots,k_t)])$,
\begin{align}
  E\cap \fix_p&=\{\h\in\Psi_p:
 s_{\bs{\lambda}_{\mc{K}\cup\{i\},p}}+s_{\bs{\lambda}_{\mc{K}\cup\{j\},p}}=
 s_{\bs{\lambda}_{\mc{K},p}}+s_{\bs{\lambda}_{\mc{K}\cup\{i,j\},p}}\}\nonumber\\
&=\{\h\in\Psi_p:
s_{k_1,\cdots,k_{l_1}+1,\cdots,k_t}+s_{k_1,\cdots,k_{l_2}+1,\cdots,k_t}=
s_{k_1,\cdots,k_t}+s_{k_1,\cdots,k_{l_1}+1,\cdots,k_{l_2}+1,\cdots,k_t}\}.
\label{kobix}
\end{align}
Similarly, for $\bs{\lambda}_{\mc{I},p}=\mb{2}_{t}(l)$, the
number of possible values of $\bs{\lambda}_{\mc{K},p}$ 
is $\frac{n_{l}-1}{n_{l}+1}\prod^t_{m=1}(n_m+1)$. 
It is because $\bs{\lambda}_{\mc{K},p}$ can be any
$(k_1,\cdots,k_t)\in\mc{M}_p$ such that $k_{l}\neq n_l-1,n_l$.
For $E\in\mc{E}_p([\mb{2}_{t}(l),(k_1,\cdots,k_t)])$,
\begin{align}
  E\cap \fix_p&=\{\h\in\Psi_p:
s_{\bs{\lambda}_{\mc{K}\cup\{i\},p}}+s_{\bs{\lambda}_{\mc{K}\cup\{j\},p}}=
 s_{\bs{\lambda}_{\mc{K},p}}+s_{\bs{\lambda}_{\mc{K}\cup\{i,j\},p}}\}\nonumber\\
&=\{\h\in\Psi_p:
 2s_{k_1,\cdots,k_{l}+1,\cdots,k_t}=
 s_{k_1,\cdots,k_t}+s_{k_1,\cdots,k_{l}+2,\cdots,k_t}\}.
\label{xiobd}
\end{align}

Let $\mf{N}_A=\{[\mb{1}_{t}(l),\mb{0}_t]:1\le l\le
  t\}$, $\mf{N}_B=\{[\mb{1}_{t}(l_1,l_2),(k_1,\cdots,k_t)]:(k_1,\cdots,k_t)\in\mc{M}_p,
k_{l_1}\neq n_{l_1}, k_{l_2}\neq n_{l_2}, 1\le l_1<l_2\le t\}$ and $\mf{N}_C=\{[\mb{2}_{t}(l),(k_1,\cdots,k_t)]:(k_1,\cdots,k_t)\in\mc{M}_p,
k_{l}\neq n_{l},1\le l\le t\}$.
From the above discussion, we see that $\mf{N}_p=\mf{N}_A\cup
\mf{N}_B\cup \mf{N}_C$. Since $\mf{E}_p=\{\mc{E}_p(\bs{\lambda}):\bs{\lambda}\in\mf{N}_p\}$,
we have
\begin{equation}
\label{bfkdg}
  |\mf{E}_p|=|\mf{N}_p|=t+\sum_{1\le l_1<l_2\le t}\frac{n_{l_1}n_{l_2}}{(n_{l_1+1})(n_{l_2+1})}\prod^t_{m=1}(n_m+1)+\sum_{1\le l\le t}\frac{n_{l}-1}{n_{l}+1}\prod^t_{m=1}(n_m+1).
\end{equation}

What we have
proved in Theorem \ref{bafoi} implies that there exists a bijection
between $\mf{E}_p$ and $\mc{G}_p$, and so $E\cap \fix_p$ as listed in 
\eqref{soidv}, \eqref{kobix} and \eqref{xiobd} are precisely all the members of
$\mc{G}_p$, and so $|\mc{G}_p|=|\mf{E}_p|$.

\textbf{Remark} When $p=\{\N\}$,  the first term of the right hand side of
\eqref{bfkdg} is 1, the second term vanishes and the last term is
$n-1$, and so $|\mf{E}_p|=n$.
 When $p=\{\{i\}:i\in\N\}$, the first term of the right hand side of
\eqref{bfkdg} is $n$, the
second term is $\binom{n}{2}2^{n-2}$ and the third term
vanishes. Therefore, $|\mf{E}_p|=n+\binom{n}{2}2^{n-2}$, which is equal
to $|\mc{E}_n|$, the
number of facets of $\Gamma_n$ \cite[(14.2)]{Y08}. In this case,
$\Psi_p=\Gamma_n$, so $|\mc{E}_n|=|\mc{G}_p|$ and by Theorem
\ref{bafoi}, $|\mc{G}_p|=|\mf{E}_p|$.
For $p=\{\N_1,\cdots,\N_t\}$, when $t\ll n$, it can be seen from
\eqref{bfkdg} that $|\mf{E}_p|\ll|\mc{E}_n|$, and then by Theorem
\ref{bafoi}, $|\mc{G}_p|\ll|\mc{E}_n|$. Therefore, for an information
theory problem with symmetrical structures induced by such a partition
$p$, the
complexity of this problem can be significantly reduced. 

By Theorem \ref{bafoi}, we denote the facets of $\Psi_p$ by
$G_p([\bs{\lambda}_{\mc{I},p},\bs{\lambda}_{\mc{K},p}])$ and thus
$\mc{G}_p=\{G_p(\bs{\lambda}):\lambda\in\mf{N}_p\}$. Specifically, 
\begin{align*}
  G_p([\bs{\lambda}_{\mc{I},p},\bs{\lambda}_{\mc{K},p}])&=
  E(\mc{I},\mc{K})\cap S_p\\
&=
\begin{cases}
  &G_p([\mb{1}_{t}(l),\mb{0}_t])=E(i)\cap S_p,\text{ if } \mc{I}=\{i\},\mc{K}=\emptyset,\ i\in\N_l,\ l=1,\cdots,t\\
&G_p([\mb{1}_t(l_1,l_2),\bs{\lambda}_{\mc{K},p}])=E(ij,\mc{K})\cap S_p,\text{ if } \mc{I}=\{i,j\},\ i\in\N_{l_1},
j\in\N_{l_2}, 1\le
 l_1<l_2\le t,\\  
&G_p([\mb{2}_t(l),\bs{\lambda}_{\mc{K},p}])=E(ij,\mc{K})\cap S_p,\text{ if } \mc{I}=\{i,j\},\ i,j\in\N_{l},\
l=1,\cdots,t.
\end{cases}
\end{align*}


\subsection{Other properties}
\label{op1512}

For $p=\{\N_1,\cdots,\N_t\}\in\mc{P}_{n}$, define $\psi_p:\hsp_n\rightarrow \hsp_n$ as follows. For any
$\mc{A}\subset\N$, 
\begin{equation*}
  \psi_p(\h)(\mc{A})=\bigg(\prod^t_{i=1}\binom{n_i}{a_i}\bigg)^{-1}\sum_{
      \mc{B}:\ \bs{\lambda}_{\mc{B},p}=\bs{\lambda}_{\mc{A},p}}
\h(\mc{B}).
\end{equation*}
where
$a_i\triangleq\bs{\lambda}_{\mc{A},p}(i)=|\mc{A}\cap\N_i|,i=1,\cdots,t$. It
can be shown that $\psi_p$ is a surjection from $\hsp_n$ onto
$\fix_p$. This is because it can be checked that for any $\h\in\hsp_n$,
$\psi_p(\h)\in \fix_p$, and for any $\h\in \fix_p$, $\h=\psi_p(\h)$.

\begin{theorem}
\label{iobad}
For any $p\in\mc{P}_{n}$,
  \begin{equation*}
  \psi_p(\h)=\frac{1}{|\Sigma_p|}\sum_{\sigma\in\Sigma_p}\sigma(\h).
\end{equation*}
\end{theorem}
\begin{proof}
  Let $p=\{\N_1,\cdots,\N_t\}$. By definition, as discussed above, $\Sigma_p$ is a group with composition as its group
  operation, which is a subgroup of the symmetrical group 
  $\Sigma_n$ and $|\Sigma_p|=\prod^t_{i=1}n_i!$. For a fixed
  $\mc{A}\subset\N$, define
  \begin{equation*}
    \Sigma_{\mc{A},p}=\{\sigma\in\Sigma_p:\sigma(\mc{A})=\mc{A}\}.
  \end{equation*}
It can be checked that $\Sigma_{\mc{A},p}$ is a subgroup of
$\Sigma_p$. Since for $\sigma\in\Sigma_p$, $\sigma_i(j)=j$ for all
$j\notin\N_i,i=1,\cdots,t$ according to \eqref{feqbs}, we have
$|\Sigma_{\mc{A},p}|=\prod^t_{i=1} a_i!(n_i-a_i)!$.
Furthermore, for any $\mc{B}\in\N$ such that $\bs{\lambda}_{\mc{B},p}=\bs{\lambda}_{\mc{A},p}$, 
  \begin{equation*}
    \Sigma_{\mc{B}|\mc{A},p}\triangleq\{\sigma\in\Sigma_p:\sigma(\mc{A})=\mc{B}\}
  \end{equation*}
is a left coset of $\Sigma_{\mc{A},p}$ in $\Sigma_p$. To see this, let $\sigma_b\in
\Sigma_{\mc{B}|\mc{A},p}$. Then it is routine to check that the mapping
$f:\Sigma_{\mc{A},p}\rightarrow\Sigma_{\mc{B}|\mc{A},p}$ defined by
$f(\sigma)=\sigma_b\circ\sigma$ is a bijection. Therefore,
\begin{equation}
  \label{eq:6}
  |\Sigma_{\mc{B}|\mc{A},p}|=|\Sigma_{\mc{A},p}|=\prod^t_ia_i!(n_i-a_i)!.
\end{equation}
Then
\begin{align}
\label{qdavd}
  \left(\frac{1}{|\Sigma_p|}\sum_{\sigma\in\Sigma_p}\sigma(\h)\right)(\mc{A})
=&\frac{1}{|\Sigma_p|}\sum_{\sigma\in\Sigma_p}\sigma(\h)(\mc{A})\\
\label{ladid}
=&\frac{1}{|\Sigma_p|}\sum_{\sigma\in\Sigma_p}\h(\sigma(\mc{A}))\\
\label{hhobr}
=&\frac{1}{|\Sigma_p|}\sum_{
      \mc{B}:\ \bs{\lambda}_{\mc{B},p}=\bs{\lambda}_{\mc{A},p}}\sum_{\sigma\in\Sigma_{\mc{B}|\mc{A},p}}\h(\sigma(\mc{A}))\\
=&\frac{1}{|\Sigma_p|}\sum_{
      \mc{B}:\ \bs{\lambda}_{\mc{B},p}=\bs{\lambda}_{\mc{A},p}}|\Sigma_{\mc{B}|\mc{A},p}|\h(\mc{B})\nonumber\\
\label{ciosd}
=&\frac{|\Sigma_{\mc{A},p}|}{|\Sigma_p|}\sum_{\sigma\in\Sigma_{\mc{B}|\mc{A},p}}\h(\mc{B})\nonumber\\
=&\frac{|\Sigma_{\mc{A},p}|}{|\Sigma_p|}\sum_{
      \mc{B}:\ \bs{\lambda}_{\mc{B},p}=\bs{\lambda}_{\mc{A},p}}\h(\mc{B})\\
\label{dfbio}
=&\frac{\prod^t_ia_i!(n_i-a_i)!}{\prod^t_{i=1}n_i!}\sum_{
      \mc{B}:\ \bs{\lambda}_{\mc{B},p}=\bs{\lambda}_{\mc{A},p}}\h(\mc{B})\\
\label{odfbi}
=&\left(\prod^t_{i=1}\binom{n_i}{a_i}\right)^{-1}\sum_{
      \mc{B}:\
      \bs{\lambda}_{\mc{B},p}=\bs{\lambda}_{\mc{A},p}}\h(\mc{B})\nonumber\\
=&\psi_p(\h)(\mc{A})\nonumber,
\end{align} which proves the theorem.
Eq. \eqref{qdavd} is valid since the summation is component-wise. Eq.
\eqref{ladid} is due to \eqref {eq:1}. For \eqref{hhobr}, we partition
$\Sigma_p$ into cosets $\Sigma_{\mc{B}|\mc{A},p}$
of $\Sigma_{\mc{A},p}$. Eqs. \eqref{ciosd} and \eqref{dfbio}
are due to \eqref{eq:6}.
\end{proof}
Therefore, we have proved that $\psi_p(\h)$ is the average of all $\sigma(\h)$,
$\sigma\in\Sigma_p$, i.e., all elements in $\mc{O}_p(\h)$.

\begin{lemma}
\label{oifdb}
  If $T\subset \hsp_n$ is $p$-invariant and convex, $\psi_p(T)=T\cap S_p$.
\end{lemma}
\begin{proof}
   Since $T$ is $p$-invariant and convex,
$\psi_p(T)\subset T$
and so
$\psi_p(T)\subset T\cap
\fix_p$.
On the other hand, for any $\s\in
T\cap \fix_p$, $\s=\psi_p(\s)\in
\psi_p(T)$ which implies that
$\psi_p(T)\supset T\cap
\fix_p$. 
\end{proof}
So, $\psi_p(\Gamma_p)=\Psi_p$ and
$\psi_p(\overline{\Gamma^*_n})=\overline{\Gamma^*_n}\cap \fix_p$. 
  
For $\mb{h}\in\hsp_n$ and $r>0$, let
$\ball{\mb{h}}{r}=\{\mb{x}\in\hsp_n:\|\mb{x}-\mb{h}\|_2<r\}$, the open
ball centered at $\mb{h}$ with radius $r$. 
For a set $A\subset\hsp_n$,
  $\ri{A}\triangleq\{\mb{x}\in\hsp_n:\exists
  \epsilon>0,\ball{\mb{x}}{\epsilon}\cap \aff{A}\subset A\}$ is called
  the \emph{relative interior} of $A$, where $\aff{A}$ is the affine
  hull of $A$.
Let $\hsp^0_n=\{\h\in\hsp_n:\h(\emptyset)=0\}$ and $\fix^0_p=\{\h\in\fix_p:\h(\emptyset)=0\}$.

\begin{theorem}
\label{lem:psi2}
For any $p\in\mc{P}_n$,
        \begin{equation*}
            \overline{\Psi^*_p}=\overline{\Gamma^*_n}\cap \fix_p.
        \end{equation*}
\end{theorem}
\begin{proof}
  Toward proving this theorem, we first prove that
        \begin{equation}
          \label{L:int}
            \overline{\ri{\Gamma^*_n}}=\overline{\Gamma^*_n}.
        \end{equation}
By \cite[Theorem 1]{M07b},
$\ri{\overline{\Gamma^*_n}}\subset\Gamma^*_n$. Since
$\ri{\overline{\Gamma^*_n}}$ is open,
$\ri{\overline{\Gamma^*_n}}\subset\ri{\Gamma^*_n}$. On the
other hand,
$\ri{\Gamma^*_n}\subset\ri{\overline{\Gamma^*_n}}$. Thus
$\ri{\Gamma^*_n}=\ri{\overline{\Gamma^*_n}}$. Taking the closure
of the both sides, we have
$\overline{\ri{\Gamma^*_n}}=\overline{\ri{\overline{\Gamma^*_n}}}$.
Hence $\overline{\ri{\Gamma^*_n}}=\overline{\Gamma^*_n}$ by
\cite[Theorem 6.3]{R70}.

Now we prove Theorem \ref{lem:psi2}. Since $\overline{\Gamma^*_n}\cap \fix_p\supset\Gamma^*_n\cap
\fix_p$, by taking the closure of both sides, we have $\overline{\Gamma^*_n}\cap
\fix_p\supset\overline{\Gamma^*_n\cap \fix_p}$. Then it suffices to prove
$\overline{\Gamma^*_n}\cap \fix_p\subset\overline{\Gamma^*_n\cap \fix_p}$, or for any
$\h\in\overline{\Gamma^*_n}\cap \fix_p$, $\h\in\overline{\Gamma^*_n\cap \fix_p}$. Let
$\ball{c}{r}$ denote an open ball centered at $c$ with radius $r$.
Given any $\epsilon>0$, let $B=\ball{\h}{\epsilon}\cap\ri{\Gamma^*_n}$.
Since $\h\in\overline{\Gamma^*_n}=\overline{\ri{\Gamma^*_n}}$ by \eqref{L:int},
$B\neq\emptyset$. Let $\g\in B$. Then there exits a set of 
random variables $X_\N$ whose entropy function is $\g$.
Then for $\sigma\in\Sigma_p$, $\sigma(\g)$ is the entropy function of
$X'_\N\triangleq\{X'_i=X_{\sigma(i)}:i\in\N\}$.
Since $\h\in \fix_p$, for
any $\sigma\in\Sigma_p$, $\|\h-$ $\sigma(\g)\|_2=\|\h-\g\|_2<\epsilon$, i.e., $\sigma(\g)\in
\ball{\h}{\epsilon}$. Let $\h^*=\frac{1}{|\Sigma_p|}\sum_{\sigma\in \Sigma_p}\sigma(\g)$.
Note that $\h^*=\psi_p(\h)\in \fix_p$. Due to the convexity of
$\ball{\h}{\epsilon}$ and $\ri{\Gamma^*_n}$, we have $\h^*\in
\ball{\h}{\epsilon}\cap\ri{\Gamma^*_n}$. Hence $\h^*\in
\ball{\h}{\epsilon}\cap\ri{\Gamma^*_n}\cap \fix_p$, that is,
$\ball{\h}{\epsilon}\cap\ri{\Gamma^*_n}\cap \fix_p\neq\emptyset$
for all $\epsilon>0$, which implies $\h\in\overline{\ri{\Gamma^*_n}\cap
\fix_p}\subset\overline{\Gamma^*_n\cap \fix_p}$. Therefore $\overline{\Gamma^*_n}\cap
\fix_p\subset\overline{\Gamma^*_n\cap \fix_p}$.
\end{proof}
Then Theorem
\ref{lem:psi2} implies that $\psi_p(\overline{\Gamma^*_n})=\overline{\Psi^*_p}$.
  For any
$p\in\mc{P}_n$, $\psi_p$ is linear by Theorem \ref{iobad} and so continuous. Hence
$\overline{\psi_p(\Gamma^*_n)}=\psi_p(\overline{\Gamma^*_n})$. 
By Lemma \ref{oifdb}, $\psi_p(\overline{\Gamma^*_n})=\overline{\Gamma^*_n}\cap
\fix_p$.
Together with Theorem \ref{lem:psi2}, we have $\overline{\psi_p(\Gamma^*_n)}=\overline{\Gamma^*_n\cap
  \fix_p}=\overline{\Psi^*_p}$. 

\section{Proof of Theorem 1}
\label{5555}

\subsection{Proof of the ``if'' part}
\label{ifpart}
In this subsection, we prove that $\overline{\Psi^*_p}=\Psi_p$ if
$p=[n]$ or $[1,n-1]$.

\emph{1) The case $p=[n]$:}

From\cite[Theorem 4.1]{Han78}, one can readily obtain
       \begin{equation}
\label{col1}
            \Psi_n=\overline{\Gamma^*_n}\cap\fix_n.
        \end{equation}
A more explicit proof of \eqref{col1} can be find in \cite[Section V]{P03}.  Then the following corollary follows immediately from \eqref{col1} and Theorem
\ref{lem:psi2}.

\begin{corollary}
\label{rel11}
    \begin{equation}
    \overline{\Psi^*_n}=\Psi_n.
  \end{equation}
\end{corollary}

\textbf{Remark}  In the 
proof of \eqref{col1} in \cite[Section V]{P03}, for every extreme ray of
$\Psi_n$, a set of random
variables whose entropy function on this ray is constructed. 
Indeed, these extreme rays of $\Psi_n$ are exactly those
rays containing the uniform matroids $U_{k,n},k\in\N$, where
$U_{k,n}(\mc{A})\triangleq\min\{k,|\mc{A}|\},\mc{A}\subset\N$. 
Since
$U_{k,n}$ are representable, they are also almost entropic, and so are those
rays containing them. The case
$p=[1,n-1]$ will be proved by a similar method.

\emph{2) The case $p=[1,n-1]$:}

According to \eqref{soidv}-\eqref{xiobd},
\begin{align}
\Psi_{1,n-1} =\{\h\in \fix_{1,n-1}:
\label{A6}
&s_{1,n-1}\ge s_{1,n-2},\\
\label{A7}
  &s_{1,n-1}\ge s_{0,n-1},\\
\label{A8}
 & s_{1,j-1}+ s_{0,j}\ge s_{0,j-1}+s_{1,j},\ 1\le j\le n-1,\\
\label{A9}
 & 2s_{i,j}\ge s_{i,j-1}+s_{i,j+1}, i=0,1,1\le j\le n-2
\}.
\end{align}

\begin{definition}[Free expansion, factor \cite{N78},\cite{W86}]
\label{def:fr_ex}
  Let $\h\in \pmr$ be an integer-valued polymatroid. Consider a set
$\mc{M}$ with cardinality $m\triangleq\sum_{i\in\N}\h(\{i\})$ and
any mapping $\phi:\N\rightarrow 2^{\mc{M}}$ such that $\phi(i)$ has
the cardinality $\h(\{i\}),i\in\N$ and
$\phi(i)\cap\phi(j)=\emptyset$ for $i\neq j$. Then the \emph{free
expansion} $\g\in\Gamma_m$ of $\h$ by $\phi$ is defined by
\begin{equation}
  \label{eq:2}
  \g(\mc{A})=\min_{\mc{B}\subset \N}\big(\h(\mc{B})+|\mc{A}\setminus \phi(\mc{B})|\big),\
  \mc{A}\subset \mc{M}.
\end{equation}
It is said that $\g$ \emph{factors to} $\h$ under $\phi$ or $\h$ is a
\emph{factor} of $\g$.
\end{definition}
It can be checked that $\g$ is also an
integer-valued polymatroid and furthermore, can be proved to be a
matroid \cite{N78}.

\begin{theorem}
\label{lem:1,n-1}
For $n\ge2$, the set of all extreme rays of $\Psi_{1,n-1}$ are the
rays containing
the polymatroids
  \begin{align}
\label{r1}
  &U^{\{1\},n}_{1,1},\\
\label{r2}
  &U^n_{1,n-1 },\cdots,U^n_{n-1,n-1 },\\
\label{r3}
  &U^n_{1,n},\cdots,U^n_{n-1,n},\\
\label{r4}
  &U^n_{2,n+1},\cdots,U^n_{n-1,n+1},\\
\label{r5}
  &\cdots,\\
\label{r6}
  &U^n_{n-2,2n-3},U^n_{n-1,2n-3},\\
\label{r7}
  &U^n_{n-1,2n-2}
\end{align}
\end{theorem}
to be defined next.

In \eqref{r1}, $U^{\{1\},n}_{1,1}$ is the matroid with ground set $\N$
which has submatroid $U_{1,1}$ on $\{1\}$ and loops $2,\cdots,n$. 
Specifically, for any $\mc{A}\subset\mc{N}$,
$U^{\{1\},n}_{1,1}(\mc{A})=|\{1\}\cap\mc{A}|$. Note that $U^{\{1\},n}_{1,1}\in
\fix_{1,n-1}$ and for
$(u_{j_1,j_2})_{(j_1,j_2)\in\N_{1,n-1}}\triangleq\s(U^{\{1\},n}_{1,1},[1,n-1])$, 
\begin{equation}
  \label{u}
  u_{j_1,j_2}=j_1.
\end{equation}
In \eqref{r2}-\eqref{r7}, for $n-1\le m\le 2n-2$ and $\max\{1,m-n+1\}\le k
\le n-1$, $U^n_{k,m}$ denotes an integer-valued polymatroid with ground
set $\N$ which is the factor of the uniform matroid $U_{k,m}$ under
$\phi_{m,n}:\N\rightarrow 2^\mc{M}$ with $\mc{M}=\{1,\cdots,m\}$ defined as follow:
\begin{equation*}
  \phi_{m,n}(i)=
  \begin{cases}
    \{1,\cdots,m-n+1\},&\text{ if } i=1,\\
   \{ i+m-n\}, &\text{ if } i\in\N\setminus\{1\}.
  \end{cases}
\end{equation*}
It can be seen that $U^n_{k,m}(\mc{A})=\min\{k,|\phi_{m,n}(\mc{A})|\}$
for any $\mc{A}\subset\mc{N}$ and $U^n_{k,m}\in \fix_{1,n-1}$. Then for
$(u_{j_1,j_2})_{(j_1,j_2)\in\N_{1,n-1}}\triangleq\s(U^n_{k,m},[1,n-1])$,
\begin{equation}
  \label{u_km}
  u_{j_1,j_2}=\min\{k,(m-n+1)j_1+j_2\}.
\end{equation}
Note that when $m=n-1$,
$\phi_{m,n}(1)=\emptyset$ and $U^n_{k,m}$ is
the matroid which has submatroid $U_{k,n-1}$ on $\{2,\cdots,n\}$ and loop
$1$; when $m=n$, $U^n_{k,m}$ coincides with the uniform matroid $U_{k,n}$.

The set of all polymatroids in \eqref{r1}-\eqref{r7} is denoted by $\mc{U}_n$.
We use $R$ in place of $U$ to denote the ray containing the corresponding
polymatroids, i.e., $R^{\{1\},n}_{1,1}$ is the ray 
containing $U^{\{1\},n}_{1,1}$ and $R^n_{k,m}$ is the ray containing
$U^n_{k,m}$. The set of all rays containing the
polymatroids in $\mc{U}_n$ is denoted by $\mc{R}_n$.

\begin{lemma}
  For $n\ge2$, any ray $R\in\mc{R}_n$ is an extreme ray of $\Psi_{1,n-1}$.
\end{lemma}
\begin{proof}
To prove this lemma, it suffices to prove that any $R\in\mc{R}_n$ is a
face of $\Psi_{1,n-1}$. Let $\mc{G}(R)=\{G\in\mc{G}_{1,n-1}:R\subset
G\}$. Specifically, we will prove that 
\begin{equation}
\label{eq:01}
  R=\bigcap_{G\in\mc{G}(R)} G,
\end{equation}
where the right hand side above is a face of $\Psi_{1,n-1}$
which by
definition contains $R$.

We now prove that $\bigcap_{G\in\mc{G}(R)} G$, the right hand side of \eqref{eq:01}, is a subset of
$R$. Since $\bigcap_{G\in\mc{G}(R)} G$ is a
face of $\Psi_{1,n-1}$, we only need to prove that it is $1$-dimensional
to conclude that $R$ is an extreme ray of $\Psi_{1,n-1}$. Toward this end,
for a specific $R$, we will consider a suitably chosen subset 
$\mc{G}'(R)$ of $\mc{G}(R)$ such that
$\bigcap_{G\in\mc{G}'(R)} G$ is $1$-dimensional. It then follows that
$\bigcap_{G\in\mc{G}(R)} G$ is also $1$-dimensional because $R\subset\bigcap_{G\in\mc{G}(R)} G\subset\bigcap_{G\in\mc{G}'(R)} G$.

First, we consider $R^{\{1\},n}_{1,1}$. It can be checked that \eqref{u}
implies that
\begin{align*}
  2u_{j_1,j_2+1}& =u_{j_1,j_2}+u_{j_1,j_2+2},\quad j_1=0,1,\text{ and
  } j_2=0,\cdots,n-3,\\
u_{1,n-1}& =u_{1,n-2},\\
u_{1,1}& =u_{1,0}+u_{0,1}.
\end{align*}
Then $R^{\{1\},n}_{1,1}$ is contained in the following three types of
facets with the respective constraints.
\begin{itemize}
\item $G_{[1,n-1]}([(0,2),(j_1,j_2)],\ j_1=0,1,$ and
  $j_2=0,\cdots,n-3$ (cf. \eqref{xiobd}):
  \begin{equation}
    \label{eqf91}
    2s_{j_1,j_2+1} =s_{j_1,j_2}+s_{j_1,j_2+2};
  \end{equation}
\item $G_{[1,n-1]}([(0,1),(0,0)])$ (cf. \eqref{soidv}):
  \begin{equation}
    \label{eqf92}
    s_{1,n-1} =s_{1,n-2};
  \end{equation}
\item $G_{[1,n-1]}([(1,1),(0,0)])$ (cf. \eqref{kobix}):
  \begin{equation}
    \label{eqf93}
    s_{1,1} =s_{1,0}+s_{0,1}.
  \end{equation}
\end{itemize}
Let $\mc{G}'(R^{\{1\},n}_{1,1})$ be the set of all the above facets, and we
now prove that $\bigcap_{G\in\mc{G}'(R^{\{1\},n}_{1,1})} G$ is $1$-dimensional.
By setting $j_1=1$ in \eqref{eqf91} and \eqref{eqf92}, we can show by
induction that $s_{1,n-1}=s_{1,n-2}=\cdots=s_{1,0}$, or
\begin{equation}
\label{eq0la}
  s_{1,j_2}=s_{1,0},\quad j_2=1,\cdots,n-1.
\end{equation}
Since $s_{1,1}=s_{1,0}$, together with \eqref{eqf93}, we have
\begin{equation}
\label{eq0lb}
  s_{0,1}=0.
\end{equation}
Using \eqref{eqf91} for $j_1=0$, \eqref{eq0lb}, and the fact that $s_{0,0}=\h(\emptyset)=0$, 
we can
show by induction that 
\begin{equation}
\label{eq0lc}
  s_{0,j_2}=0,\quad j_2=1,\cdots,n-1.
\end{equation}
Finally, by combining \eqref{eq0la} and \eqref{eq0lc}, we have 
\begin{equation*}
  s_{j_1,j_2}=s_{1,0}j_1,\quad (j_1,j_2)\in\N_{1,n-1},
\end{equation*}
 which implies that the right hand side of \eqref{eq:01} is
1-dimensional and so \eqref{eq:01} is valid for $R=R^{\{1\},n}_{1,1}$.

For $R^n_{k,m}$, we prove \eqref{eq:01} in the same way as we have proved
the case for $R^{\{1\},n}_{1,1}$.
For the convenience of discussion, we will specify the elements of 
$\mc{G}'(R^n_{k,m})$ as we
progress. Let $\N'=\{1,\cdots,n-1\}$. We first consider the facets and the
face (an intersection of two facets) below which contain $R^n_{k,m}$
for particular values of the indices (depending on the values of $k$
and $m$):
\begin{itemize}
\item  $G^n_{j_1,j_2}\triangleq G_{[1,n-1]}([(0,2),(j_1,j_2-1)])$, $ j_1=0,1$ and
 $ j_2=1,\cdots,n-2$ (cf. \eqref{xiobd}):
\begin{equation}
\label{eqg91}
  2s_{j_1,j_2}= s_{j_1,j_2-1}+s_{j_1,j_2+1};
\end{equation}
\item $G^n_{0,n-1}\triangleq G_{[1,n-1]}([(0,1),(0,0)])\cap
G_{[1,n-1]}([(1,1),(0,n-2)])$ (cf. \eqref{soidv},\eqref{kobix}):
\begin{equation}
\label{eqg93}
  s_{0,n-1}=s_{0,n-2};
\end{equation}
\item $G^n_{1,n-1}\triangleq G_{[1,n-1]}([(0,1),(0,0)])$ (cf. \eqref{soidv}):
  \begin{equation}
    \label{eqg92}
    s_{1,n-1} =s_{1,n-2};
  \end{equation}
\item  $G_{[1,n-1]}([(1,0),(0,0)])$ (cf. \eqref{soidv}):
\begin{equation}
  \label{eqg94}
   s_{1,n-1}=s_{0,n-1};
\end{equation}
\item $G_{[1,n-1]}([(1,1),(0,0)])$(cf. \eqref{kobix}):
\begin{equation}
  \label{eqg95}
   s_{1,1}=s_{0,1}+s_{1,0}.
\end{equation}
\end{itemize}

Using \eqref{u_km}, it can be checked that $R^n_{k,m}\subset
\bigcap_{j_1\in\N',j_2\neq k}G_{0,j_2}$. By letting $s_{0,n}=s_{0,n-1}$, we
can combine 
\eqref{eqg91} for $j_1=0$ and \eqref{eqg93} to obtain
\begin{equation*}
  2s_{0,j_2}=s_{0,j_2-1}+s_{0,j_2+1},\quad j_2\in\N',\ j_2\neq k.
\end{equation*}
Then we can readily show that
\begin{equation}
\label{eq0l}
  s_{0,j_2}=s_{0,1}\min\{k,j_2\},\quad j_2\in\N'.
\end{equation}

\begin{itemize}
\item If $k= m-n+1$, by \eqref{u_km},
$R^n_{k,m}\subset\bigcap_{j_2\in\N'}G^n_{1,j_2}$.
By letting $s_{1,n}=s_{1,n-1}$, we can combine 
\eqref{eqg91} for $j_1=1$ and \eqref{eqg92} to obtain
\begin{equation*}
  2s_{1,j_2}=s_{1,j_2-1}+s_{1,j_2+1},\quad j_2\in\N'.
\end{equation*}
Then we can show by induction that
$s_{1,n-1}=s_{1,n-2}=\cdots=s_{1,0}$, or
\begin{equation}
\label{eq1la}
  s_{1,j_2}=s_{1,n-1},\quad j_2\in\N'.
\end{equation}

Consider \eqref{eq0l} for $j_2=n-1$. Since $k\le n-1$, 
\begin{equation}
  \label{add1}
s_{0,n-1}=ks_{0,1}.
\end{equation}
Since $R^n_{k,m}\subset G_{[1,n-1]}([(1,0),(0,0)])$, we have
\eqref{eqg94}. Together with \eqref{add1}, 
\begin{equation}
  \label{add2}
s_{1,n-1}=ks_{0,1}.
\end{equation}
Then by \eqref{add2} and \eqref{eq1la}, 
\begin{equation}
  \label{add3}
s_{1,j_2}=ks_{0,1},\quad j_2\in\N'.
\end{equation}
Combining \eqref{eq0l} and \eqref{add3}, we have
\begin{equation}
  \label{eqla}
  s_{j_1,j_2}=s_{0,1}\min\{k,j_1k+j_2\},\quad(j_1,j_2)\in\mc{M}_{1,n-1}.
\end{equation}

\item If $k> m-n+1$, by \eqref{u_km},
$R^n_{k,m}\subset \bigcap_{j_2\in\N',j_2\neq k+n-m-1}G^n_{1,j_2}$. By letting $s_{1,n}=s_{1,n-1}$, we can combine 
\eqref{eqg91} for $j_1=1$ and \eqref{eqg92} to obtain
\begin{equation*}
  2s_{1,j_2} =s_{1,j_2-1}+s_{1,j_2+1},\quad j_2\in\N',\ j_2\neq
  k+n-m-1,
\end{equation*}
Note that $R^n_{k,m} \subset G_{[1,n-1]}([(1,1),(0,0)])$. Together with
\eqref{eqg95}, we can show that 
\begin{equation}
\label{eq1lb}
  s_{1,j_2}=s_{1,0}+s_{0,1}\min\{k+n-m-1,j_2\},\quad j_2\in\N'.
\end{equation}
Consider \eqref{eq1lb} for $j_2=n-1$. Since $k\le n-1\le m$, we have
\begin{equation}
  \label{bdd1}
s_{1,n-1}=s_{1,0}+s_{0,1}(k+n-m-1).
\end{equation}
Since $R^n_{k,m}\subset G_{[1,n-1]}([(1,0),(0,0)])$, we have
\eqref{eqg94}. Together with \eqref{bdd1}, 
\begin{equation}
  \label{bdd2}
s_{0,n-1}=s_{1,0}+s_{0,1}(k+n-m-1).
\end{equation}
By \eqref{add1} and \eqref{bdd2}, we have $s_{1,0}=s_{0,1}(m-n+1)$,
and so by \eqref{eq1lb},
\begin{equation}
\label{bdd3}
  s_{1,j_2}=s_{0,1}\min\{k,m-n+1+j_2\},\quad j_2\in\N'.
\end{equation}
Combining \eqref{eq0l} and \eqref{bdd3}, we have 
\begin{equation}
  \label{eqlb}
  s_{j_1,j_2}=s_{0,1}(m-n+1)j_1+s_{0,1}\min\{k+(n-m-1)j_1,j_2\}, \quad(j_1,j_2)\in\mc{M}_{1,n-1}.
\end{equation}

\end{itemize}

It can be seen that both \eqref{eqla} and \eqref{eqlb} can be written as
\begin{equation*}
  s_{j_1,j_2}=s_{0,1}\min\{k,(m-n+1)j_1+j_2\}, \quad(j_1,j_2)\in\mc{M}_{1,n-1},
\end{equation*}
which implies that the right hand side of \eqref{eq:01} is
1-dimensional and then \eqref{eq:01} is valid for $R=R^n_{k,m}$.
\end{proof}
\textbf{Remark} Note that upon deleting the loops, $U^{\{1\},n}_{1,1}$, $U^{n}_{k,n-1},k=1,\cdots,n-2$ and $U_{k,n},k=1,\cdots,n-1$ become connected
matroids. By \cite[Theorem 2.1.5]{N78}, the rays containing these matroids
are extreme rays of $\Gamma_n$ and so extreme rays of $\Psi_{1,n-1}$. 
This is an alternative method to prove that rays containing these
matroids are extreme rays of $\Psi_{1,n-1}$. However, this method
cannot handle the remaining polymatroids in Theorem \ref{lem:1,n-1}.

\begin{lemma}
  For $n\ge 2$, $\mc{R}_n$ contains all the extreme rays of $\Psi_{1,n-1}$.
\end{lemma}
\emph{Proof} It can be readily checked that $\mc{R}_2$ contains all
the extreme rays of $\Psi_{1,1}$. So it suffices to prove that if
$\mc{R}_n$ contains all the extreme rays of $\Psi_{1,n-1}$, then
$\mc{R}_{n+1}$ contains all the extreme rays of $\Psi_{1,n}$.
To this end, we will prove that for an arbitrary 
$\h_{n+1}\in\Psi_{1,n}$, it
can be written as a conic combination of the polymatroids in
$\mc{U}_{n+1}$.

Consider $\h_n\in\Psi_{1,n-1}$. Let
$\s_n\triangleq(s_{j_1,j_2})_{(j_1,j_2)\in\N_{1,n-1}}=\s(\h_n,[1,n-1])$. Let
$\h_{n+1}\in\Psi_{1,n}$ such that
$\s(\h_{n+1},[1,n])=\s_{n+1}\triangleq(s_{j_1,j_2})_{(j_1,j_2)\in\N_{1,n}}$,
i.e., $\s_{n+1}$ contains $\s_n$ as a subvector with two additional entries
$s_{0,n}$ and $s_{1,n}$. These two entries satisfy the following five constraints:
\begin{align}
\label{l1}
&s_{1,n}\ge s_{1,n-1},\\
\label{l2}
  &s_{1,n}\ge s_{0,n},\\
\label{l3}
 & s_{1,n-1}+ s_{0,n}\ge s_{0,n-1}+s_{1,n},\\
\label{l4}
 & 2s_{1,n-1}\ge s_{1,n-2}+s_{1,n},\\
\label{l5}
 & 2s_{0,n-1}\ge s_{0,n-2}+s_{0,n},
\end{align}
which can be rewritten as
\begin{align}
\label{t1}
  s_{1,n-1}\le &\ s_{1,n}\le 2s_{1,n-1}-s_{1,n-2},\\
\label{t2}
  s_{1,n}-s_{1,1-n}+s_{0,n-1}\le &\ s_{0,n}\le \min \{s_{1,n},2s_{0,n-1}-s_{0,n-2}\},
\end{align}
or
\begin{align}
\label{e1}
  s_{1,n}=&\ s_{1,n-1}+e_1,\\
\label{e2}
  s_{0,n}=&\ s_{0,n-1}+e_1+e_2,
\end{align}
where $0\le e_1\le s_{1,n-1}-s_{1,n-2}$ and $0\le e_2\le
\min\{s_{1,1-n}-s_{0,n-1},s_{0,n-1}-s_{0,n-2}-e_1\}$ which can be seen
by substituting \eqref{e1} into \eqref{t2}. Then $L\triangleq
(s_{1,n-1}-s_{1,n-2},s_{1,1-n}-s_{0,n-1},s_{0,n-1}-s_{0,n-2})$
bounds the range of auxiliary variables $e_1,e_2$, where the first
component is the upper bound of $e_1$ and the second and third
components together with $e_1$ define the upper bound on $e_2$. If $\s_{n}\in\mc{U}_n$, the entries of $L$ can only be 0 or
1. We classify the members of $\mc{U}_n$ into four classes according
to $L$ as follow.
\begin{enumerate}
\item $L=(0,1,0)$, $\mc{A}=\{U^{\{1\},n}_{1,1}\}$,
\item $L=(1,0,1)$, $\mc{B}=\{U^n_{n-1,n-1}\}$,
\item $L=(0,0,1)$, $\mc{C}=\{U^n_{n-1,m}:n\le m\le 2n-2\}$,
\item $L=(0,0,0)$, $\mc{D}=\mc{U}_n\setminus(\mc{A}\cup\mc{B}\cup\mc{C})$.
\end{enumerate}
These can be verified from \eqref{u} and \eqref{u_km}.

For $\tilde{\s}_{n+1} \in\s(\Psi_{1,n},[1,n])$ with subvector
$\tilde{\s}_{n}\in\s(\Psi_{1,n-1},[1,n-1])$ and two additional entries
$\tilde{s}_{0,n},\tilde{s}_{1,n}$, we write it as $(\tilde{\s}_{n}, \tilde{s}_{0,n},
\tilde{s}_{1,n})$. 
For
\begin{equation*}
  \h_n=a\ub_A+b\ub_B+\sum_{\ub_C\in\mc{C}} c_{\ub_C}\ub_C+\sum_{\ub_D\in\mc{D}} d_{\ub_D}\ub_D,
\end{equation*}
where $\ub_A\in\mc{A},\ub_B\in\mc{B}$, $a, b,
c_{\ub_C},d_{\ub_D}\ge0$, and $s_{1,n}=
s_{1,n-1}+e_1,s_{0,n}=s_{0,n-1}+e_1+e_2$, it can be checked that
\begin{align}
  \s_{n+1}=\ &(a-e_2)(\s(\ub_A,[1,n-1]),0,1)\nonumber\\
&+\left( b-e_1-e_2+\sum_{\ub_C\in\mc{C}} c'_{\ub_C}\right) (\s(\ub_B,[1,n-1]),n-1,n-1) \nonumber\\
&+e_1(\s(\ub_B,[1,n-1]),n,n) \nonumber\\
&+\left(e_2-\sum_{\ub_C\in\mc{C}} c'_{\ub_C}\right)(\s(\ub_A+\ub_B,[1,n-1]),n,n) \nonumber\\
&+\sum_{\ub_C\in\mc{C}} (c_{\ub_C}-c'_{\ub_C})(\s(\ub_C,[1,n-1]),n-1,n-1) \nonumber\\
&+\sum_{\ub_C\in\mc{C}} c'_{\ub_C}(\s(\ub_A+\ub_C,[1,n-1]),n,n) \nonumber\\
&+\sum_{\ub_D\in\mc{D}} d_{\ub_D}(\s(\ub_D,[1,n-1]),k_{\ub_D},k_{\ub_D}),\label{dec}
\end{align}
where $0\le c'_{\ub_C}\le \min\{c_{\ub_C},e_2\}$ and $\sum
c'_{\ub_C}\ge e_1+e_2$. For $\ub_A=U^{\{1\},n}_{1,1}$,
$\ub_B=U^n_{n-1,n-1}$, $\ub_C=U^n_{n-1,m_{\ub_C}}$ and
$\ub_D=U^n_{k_{\ub_D},m_{\ub_D}}$. It can be checked that
\begin{align*}
  (\s(\ub_A ,[1,n-1]),0,1)& =\s(U^{\{1\},n+1}_{1,1},[1,n]),\\
(\s(\ub_B ,[1,n-1]),n-1,n-1)& =\s(U^{n+1}_{n-1,n},[1,n]),\\
(\s(\ub_B ,[1,n-1]),n,n)& =\s(U^{n+1}_{n,n},[1,n]), \\
(\s(\ub_A+\ub_B ,[1,n]),n,n)& =\s(U^{n+1}_{n,n+1},[1,n]),\\
(\s(\ub_C ,[1,n-1]),n-1,n-1)& =\s(U^{n+1}_{n-1,m_{\ub_C}+1},[1,n]),\\ 
(\s(\ub_A+\ub_C ,[1,n]),n,n)& =\s(U^{n+1}_{n,,m_{\ub_C}+1},[1,n]),\\
(\s(\ub_D),k_{\ub_D},m_{\ub_D})& =\s(U^{n+1}_{k_{\ub_D},m_{\ub_D}+1},[1,n]). 
\end{align*}

Then by taking
$\s^{-1}$ on both sides of \eqref{dec}, $\h_{n+1}$ is a conic
combination of polymatroids in
$\mc{U}_{n+1}$. 

By exhausting all $\h_{n}\in\Psi_{1,n-1}$ and possible $e_1,e_2$, we
can write an arbitrary $\h_{n+1}\in\Psi_{1,n}$ as a conic combination of the
polymatroids in $\mc{U}_{n+1}$, which implies that there exist no other
extreme rays of $\Psi_{1,n}$. \hfill\QQQ

\begin{theorem}
\label{prop:fr_ex} Let $\h\in\Gamma_n$ be integer-valued and a
factor of $\g\in\Gamma_m$ under some $\phi$. Then
$\g\in\overline{\Gamma^*_m}$ if and only if $\h\in\overline{\Gamma^*_n}$.
\end{theorem}
\begin{proof}
The ``if'' part is proved in \cite[Theorem 4]{M07b}. For the ``only
if'' part, by the continuity of free expansion, it suffices to prove
that $\h\in\ef$ if $\g\in\Gamma^*_m$. Let $\g$ be the entropy
function of the random vector
$Y_\mc{M}=(Y_j)_{j\in\mc{M}}$. Now define $X_\N=(X_i)_{i\in\N}$ by
$X_i=(Y_j)_{j\in\phi(i)}$ for all $i\in\N$. Then it can be checked that $\h$ is the
entropy function of $X_\N$. 
\end{proof}

\begin{theorem}
  \label{qwekl}
  \begin{equation*}
    \overline{\Psi^*_{1,n-1}}=\Psi_{1,n-1}.
  \end{equation*}
\end{theorem}
\begin{proof}
As $\overline{\Psi^*_{1,n-1}}\subset\Psi_{1,n-1}$, we only need to
prove anther direction of the inclusion.
  Uniform matroids are representable and so almost entropic.
  According to Theorem \ref{lem:1,n-1}, all extreme rays of $\Psi_{1,n-1}$
  contains polymatroids which are
  factors of some uniform matroids. It follows from Theorem
  \ref{prop:fr_ex} that all of these polymatroids are almost
  entropic. Since $\Psi_{1,n-1}$ is a convex cone,
  $\Psi_{1,n-1}\subset \overline{\Gamma^*_n}$ and
  so $\Psi_{1,n-1}\subset \overline{\Gamma^*_n}\cap
  \fix_{1,n-1}$. By Theorem \ref{lem:psi2}, $\overline{\Psi^*_{1,n-1}}=\overline{\Gamma^*_n}\cap
  \fix_{1,n-1}$ and theorem follows.
\end{proof}

\subsection{Proof of the ``only if'' part}
\emph{1) $p=[n_1,n_2]$ with $n_1,n_2\ge 2$:}

\label{B999}
Consider $\wh\in \hsp^0_4$, where
\begin{equation}
  \label{eq:13}
  \wh(\mc{A})=
  \begin{cases}
    2\quad \text{ if }|\mc{A}|=1,\\
    3\quad \text{ if }|\mc{A}|=2 \text{ and } \mc{A}\neq\{1,2\},\\
    4\quad \text{ if } \mc{A}=\{1,2\} \text{ or }|\mc{A}|\ge 3.
  \end{cases}
\end{equation}
It can be checked that $\wh\in\Gamma_4$ and $\wh\in S_{2,2}$, so
$\wh\in\Psi_{2,2}$. On the other hand, $\wh$ violates the Zhang-Yeung
inequality \cite{ZY98} which implies that
$\wh\notin\overline{\Gamma^*_4}$, and so $\wh\notin\overline{\Psi^*_{2,2}}$.
It follows that 
\begin{equation}
  \label{eq:004}
  \overline{\Psi^*_{2,2}}\subsetneq\Psi_{2,2}.
\end{equation}

\textbf{Remark} The free expansion of $\wh$ is the V\'amos matroid, a well-known
non-representable matroid.

\begin{theorem}
\label{fbad}
  For any $p=[n_1,n_2]\in\mc{P}^*_{n}$ such that $n_1,n_2\ge2$, 
  \begin{equation*}
    \overline{\Psi^*_{n_1,n_2}}\subsetneq \Psi_{n_1,n_2}.
  \end{equation*}
\end{theorem}
\begin{proof}
  Consider $\wh_{n_1,n_2}\in \hsp^0_n$, where
\begin{equation}
  \label{eq:14}
  \wh_{n_1,n_2}(\mc{A})=
  \begin{cases}
    2\quad \text{ if }|\mc{A}|=1,\\
    3\quad \text{ if }|\mc{A}|=2\text{ and } \mc{A}\not\subset\N_1,\\
    4\quad \text{ if } |\mc{A}|=2\text{ and } \mc{A}\subset\N_1,\\
    4\quad \text{ if } |\mc{A}|\ge 3.
  \end{cases}
\end{equation}
It can be checked that $\wh_{n_1,n_2}\in\Gamma_n$ and $\wh\in \fix_{n_1,n_2}$, so
$\wh\in\Psi_{n_1,n_2}$. We claim that
$\wh_{n_1,n_2}\notin\overline{\Gamma^*_n}$. Assume otherwise. Let
$\N'_i\subset\N_i$, $i=1,2$ such that $|\N'_i|=2$. Let $\wh'$ be a
polymatroid with ground set $\N'_1\cup\N'_2$ such that for any $\mc{A}\subset \N'_1\cup\N'_2$,
\begin{equation*}
 \wh'(\mc{A})=\wh_{n_1,n_2}(\mc{A}).
\end{equation*}
Then
\begin{equation*}
  \wh'(\mc{A})=
  \begin{cases}
    2\quad \text{ if }|\mc{A}|=1,\\
    3\quad \text{ if }|\mc{A}|=2 \text{ and } \mc{A}\neq \N'_1\\
    4\quad \text{ if } \mc{A}=\N'_1\text{ or }|\mc{A}|\ge 3.
  \end{cases}
\end{equation*}
Since $\wh'$ is the restriction of $\wh_{n_1,n_2}$ on $\N_1\cup\N_2$, 
$\wh'$ is almost entropic. On the other hand, it is seen that 
$\wh'$ is not almost entropic because it violates the Zhang-Yeung
inequality (cf.\eqref{eq:13}). This leads to a contradiction.
 Therefore
$\wh_{n_1,n_2}\notin\overline{\Gamma^*_n}$ and so
$\wh_{n_1,n_2}\notin\overline{\Psi^*_{n_1,n_2}}$. Hence, $\overline{\Psi^*_{n_1,n_2}}\subsetneq \Psi_{n_1,n_2}$. 
\end{proof}

\emph{2) Multi-partition entropy functions:}

Now we consider the cases $p\in\mc{P}_{n}$ with more
than two partitions.
\begin{lemma}
\label{sdbpo}
  For $p_1,p_2\in\mc{P}_n$ such that
$p_1\le p_2$, $\fix_{p_2}\subset \fix_{p_1}$. Furthermore, if $p_1<p_2$, $\fix_{p_2}\subsetneq \fix_{p_1}$.
\end{lemma}
\begin{proof}
Let $p_1=\{\N^{(1)}_1,\cdots,\N^{(1)}_{t_1}\}$ and $p_2=\{\N^{(2)}_1,\cdots,\N^{(2)}_{t_2}\}$,
 where $p_1\le p_2$. Note that for any $\mc{A}\subset\N$ and any
 $i=1,\cdots,t_2$,
 $\bs{\lambda}_{\mc{A},p_2}(i)=|\mc{A}\cap\N^{(2)}_i|=|\mc{A}\cap(\cup_{j\in\mc{J}_i}\N^{(1)}_j)|=\sum_{j\in\mc{J}_i}|\mc{A}\cap
 \N^{(1)}_j|=\sum_{j\in\mc{J}_i}\bs{\lambda}_{\mc{A},p_1}(j)$.
That is, each entry of $\bs{\lambda}_{\mc{A},p_2}$ is the
summation of some entries of $\bs{\lambda}_{\mc{A},p_1}$. Hence, for any
$\mc{A},\mc{B}\subset\N$,
$\bs{\lambda}_{\mc{A},p_2}=\bs{\lambda}_{\mc{B},p_2}$ if
 $\bs{\lambda}_{\mc{A},p_1}=\bs{\lambda}_{\mc{B},p_1}$. Then for $\h\in
 \fix_{p_2}$, for any $\mc{A},\mc{B}\subset\N$,
 $\h(\mc{A})=\h(\mc{B})$ if
 $\bs{\lambda}_{\mc{A},p_2}=\bs{\lambda}_{\mc{B},p_2}$. It follows that $\h(\mc{A})=\h(\mc{B})$ if
 $\bs{\lambda}_{\mc{A},p_1}=\bs{\lambda}_{\mc{B},p_1}$. Hence, $\h\in \fix_{p_1}$.

For $p_1<p_2$, there exists $i\in\{1,\cdots,t_2\}$ with
$\N^{(2)}_i=\cup_{j\in\mc{I}_i}\N^{(1)}_j$ such that $|\mc{J}_i|\ge 2$. Let $\mc{A}=\{l_1\},\mc{B}=\{l_2\}$ with
$l_1\in\N^{(1)}_{j_1}$ and $l_2\in\N^{(1)}_{j_2}$ such that
$j_1,j_2\in\mc{J}_i$. Let $\h\in \fix_{p_1}$ be such that $\h(\mc{A})\neq
\h(\mc{B})$. Then $\h\notin \fix_ {p_2}$.
\end{proof}

\begin{theorem}
\label{qpads}
  For $n\ge4,t\ge3$ and any $t$-partition $p\in\mc{P}_{n}$, $\overline{\Psi^*_p}\subsetneq\Psi_p$.  
\end{theorem}
\begin{proof}
   When $n\ge4,t\ge3$, for any $t$-partition $p\in\mc{P}_{n}$, there
exists a $2$-partition $p'\in\mc{P}_{n}$ such that $p\le p'$ and the cardinality of each
block of
$p'$ is at least 2.  Then by Lemma \ref{sdbpo}, $\fix_{p}\supset
\fix_{p'}$ which implies that $\Psi_p\supset\Psi_{p'}$.  

Now by Theorem
\ref{fbad}, $\overline{\Psi^*_{p'}}\subsetneq\Psi_{p'}$, i.e., there
exists $\h\in\Psi_{p'}$ but $\h\notin
\overline{\Psi^*_{p'}}$. Consider
$\h\neq\overline{\Psi^*_{p'}}=\overline{\Gamma^*_n\cap
  \fix_{p'}}=\overline{\Gamma^*_p}\cap \fix_{p'}$ (cf. Theorem \ref{lem:psi2}). Since
$\h\in\Psi_{p'}=\Gamma_n\cap \fix_{p'}$ implies $\h\in \fix_{p'}$, we see that
$\h\notin \overline{\Gamma^*_n}$ which in turn implies that $\h\notin
\overline{\Gamma^*_n}\cap \fix_{p}=\overline{\Psi^*_p}$. On the other
hand, $\h\in\Psi_{p'}\subset\Psi_p$. Therefore
$\h\in\Psi_p\setminus \overline{\Psi^*_p}$ or $\overline{\Psi^*_{p}}\subsetneq\Psi_{p}$.
\end{proof}

\section{Discussion}
\label{baddf}
\subsection{Applications to secret-sharing}

In this section, we discuss the application of our results to secret sharing. Consider
$p=[1,n-1]$. A \emph{secret-sharing} problem involves a dealer who has a secret, indexed by $\N_1=\{1\}$,
a set of $n-1$ parties, indexed by $\N_2=\{2,\cdots,n\}$, and
a collection $\mf{A}$ of sets of parties, i.e., $\mf{A}\subset 2^{\N_2}$,
called the access structure. Let random variable $X_1$ be the secret and $X_i,i\in\N_2$ be
the share distributed to party $i$.
For a \emph{perfect} secret-sharing problem, any authorized set $\mc{A}\in\mf{A}$ of parities can
reconstruct the secret, i.e.,
\begin{equation*}
H(X_1,X_{\mc{A}})=H(X_{\mc{A}}) \quad (\text{correctness criteria})
\end{equation*}
 and any unauthorized set $\mc{A}\notin\mf{A}$ of parities cannot
 reconstruct any information, i.e.,
 \begin{equation*}
 H(X_1,X_{\mc{A}})=H(X_1)+H(X_{\mc{A}}) \quad (\text{privacy criteria} ).
 \end{equation*}
Note that $\mf{A}$ is monotone: if $\mc{A}\subset\mc{B}$ and
$\mc{A}\in\mf{A}$, then $\mc{B}\in\mf{A}$. We say the set of random
variables $X_\N$ is a secret-sharing scheme realizing $\mf{A}$ if it satisfies both the correctness criteria and the privacy
criteria. For a secret-sharing scheme $X_\N$ realizing $\mf{A}$, let
$\h\in\hsp_n$ be its entropy function. The information ratio of $X_\N$ is defined
by 
\begin{equation*}
  \rho_{\mf{A}} (\h)=\frac{\max_{i\in\N_2}\h(\{i\})}{\h(\{1\})}.
\end{equation*}

Our goal is to determine the optimal information ratio of $\mf{A}$, i.e., the infimum
of $\rho_{\mf{A}}$ over all possible scheme
$X_{\N}$ realizing $\mf{A}$. In other words, we need to solve the following
optimization problem:
\begin{equation}
\label{66}
  \rho^{\mr opt}_{\mf{A}}\triangleq\inf_{\h\in\Gamma^*_n\cap C({\mf{A}})} \rho_{\mf{A}} (\h)=\inf_{\h\in\Gamma^*_n\cap C({\mf{A}})} \frac{\max_{i\in\N_2}\h(\{i\})}{\h(\{1\})},
\end{equation}
where 
\begin{align*}
  C({\mf{A}})=\{\h\in\hsp_n:\ &\h(\{1\}\cup\mc{A})=\h(\mc{A}),\quad\mc{A}\in\mf{A}\\
&\h(\{1\}\cup\mc{A})=\h(\{1\})+\h(\mc{A}),\quad\mc{A}\notin\mf{A}\}.
\end{align*}

In Shamir's threshold secret-sharing with threshold $1\le t\le
n-1$, the  access structure is
\begin{equation*}
  \mf{A}_t=\{\mc{A}\subset\N_2: |\mc{A}|\ge t\}.
\end{equation*}
It is also called the $t$-threshold secret-sharing.
Then
\begin{equation*}
  \rho^{\mr opt}_{\mf{A}_t}=\inf_{\h\in\Gamma^*_n\cap C(\mf{A}_t)} \frac{\max_{i\in\N_2} \h(\{i\})}{\h(\{1\})}.
\end{equation*}

\begin{lemma}
\label{l1}
  \begin{equation*}
  \rho^{\mr opt}_{\mf{A}_t}=\inf_{\h\in\Psi^*_{1,n-1}\cap C(\mf{A}_t)} \frac{\max_{i\in\N_2} \h(\{i\})}{\h(\{1\})}.
\end{equation*}
\end{lemma}
\begin{proof}
Consider the permutation group $\Sigma_{1,n-1}=\{\sigma\in\Sigma_n: \sigma(1)=1, \sigma(i)\in\N_2, i\in\N_2\}$.
Note that for any $\sigma\in\Sigma_{1,n-1}$,
$\sigma(\Psi^*_{1,n-1}\cap C(\mf{A}_t))=\Psi^*_{1,n-1}\cap C(\mf{A}_t)$
and $\frac{\max_{i\in\N_2} \sigma(\h)(\{i\})}{\sigma(\h)(\{1\})}=\frac{\max_{i\in\N_2} \h(\{i\})}{\h(\{1\})}$.
Suppose the sequence $\h^{(k)}\in\Gamma^*_n\cap C(\mf{A}_t)$
achieves the infimum. Then $\sigma(\h^{(k)})\in\Gamma^*_n\cap
C(\mf{A}_t)$ also achieves the infimum. 
By \cite[Corollary
15.4]{Y08},  $\sum_{\sigma\in\Sigma_{1,n-1}}\sigma(\h^{(k)})\in \Gamma^*_n\cap C(\mf{A}_t)$.
Now for $\h_j\in\hsp_n, j=1,\cdots,t$, if $\frac{\max_{i\in\N_2}
  \h_j(\{i\})}{\h_j(\{1\})}=c$ for all $j$, then it can be seen that $\frac{\sum^t_{j=1}\max_{i\in\N_2} \h_j(\{i\})}{\sum^t_{j=1}\h_j(\{1\})}=c$.
Therefore $\sum_{\sigma\in\Sigma_{1,n-1}}\sigma(\h^{(k)})$ also achieves the infimum. Note that $\sum_{\sigma\in\Sigma_{1,n-1}}\sigma(\h^{(k)})\in \fix_{1,n-1}$.
Hence the optimization problems over $\Gamma^*_n \cap C(\mf{A}(t))$ and $\Gamma^*_n\cap \fix_{1,n-1}\cap C(\mf{A}_t)=\Psi^*_{1,n-1}\cap(\mf{A}_t)$ give the same value. The lemma is proved.
\end{proof}

By Lemma \ref{l1}, the computation of the optimal information ratio can be confined to $\Psi^*_{1,n-1}\cap C(\mf{A}_t)$. Let
\begin{align}
    \fix(\mf{A}_t)=C(\mf{A}_t)\cap \fix_{1,n-1}=\{\h\in\hsp_n: \ & s_{1,k}=s_{0,k}, \quad\text{ if } k\ge t;\label{a2}\\
   &s_{1,k}=s_{1,0}+s_{0,k}, \quad \text{ if }  k< t\}.\label{a3}
\end{align}
Then 
\begin{equation*}
  \rho^{\mr opt}_{\mf{A}_t}=\inf_{\h\in\Psi^*_{1,n-1}\cap\fix(\mf{A}_t)}\frac{s_{0,1}}{s_{1,0}}.
\end{equation*}

\begin{lemma}
\label{l10}
  \begin{equation*}
    \inf_{\h\in\Psi_{1,n-1}\cap \fix(\mf{A}_t)} \frac{s_{0,1}}{s_{1,0}}=\min_{\h\in\Psi_{1,n-1}\cap \fix(\mf{A}_t)} \frac{s_{0,1}}{s_{1,0}}=1.
  \end{equation*}
This lemma asserts that the infimum of the problem can be achieved by any non-origin $\h\in R_{t,n}$,
i.e., the ray containing the uniform matroid $U_{t,n}$.
\end{lemma}
\begin{proof}
By \eqref{A6}-\eqref{A9} and \eqref{a2}-\eqref{a3}, it can be seen
that $\Psi_{1,n-1}\cap \fix(\mf{A}_t)$ is a face of
$\Psi_{1,n-1}$. Checking \eqref{r1}-\eqref{r7}, it follows that
$\Psi_{1,n-1}\cap \fix(\mf{A}_t)$ is the convex hull of the rays
$R_{t,n}$ and $R^{n}_{i,n-1},i=1,\cdots,n-1$, i.e., the rays 
containing the polymatroids $U_{t,n}$ and $U^{n}_{i,n-1},i=1,\cdots,n-1$.
Note that $\Psi_{1,n-1}\cap \fix(\mf{A}_t)$ is a convex cone. For a ray $R\subset \Psi_{1,n-1}\cap \fix(\mf{A}_t)$, for any non-origin $\h\in R$, $s_{0,1}/s_{1,0}$ has the same value. 
Furthermore, for $\h_1,\h_2\in \hsp_n$ and $\h_3=\h_1+\h_2$, if
$s^{(1)}_{0,1}/s^{(1)}_{1,0}\le s^{(2)}_{0,1}/s^{(2)}_{1,0}$, then
$s^{(1)}_{0,1}/s^{(1)}_{1,0}\le s^{(3)}_{0,1}/s^{(3)}_{1,0}\le
s^{(2)}_{0,1}/s^{(2)}_{1,0}$. Hence the infimum of
$s_{0,1}/s_{1,0}$ must be on the extreme rays $R_{t,n}$ or 
$R^{n}_{i,n-1}$ for some $i\le i\le n-1$. By checking that $s_{0,1}/s_{1,0}=1$
for any non-origin $\h\in
R_{t,n}$, and $s_{0,1}/s_{1,0}$ goes to infinity for any non-origin $\h\in
R^n_{i,n-1},i=1,\cdots,n-1$, we
conclude that 
\begin{equation*}
    \inf_{\h\in\Psi_{1,n-1}\cap \fix(\mf{A}_t)} \frac{s_{0,1}}{s_{1,0}}=\min_{\h\in\Psi_{1,n-1}\cap \fix(\mf{A}_t)} \frac{s_{0,1}}{s_{1,0}}=1,
  \end{equation*}
and the minimum of the problem can be achieved by any non-origin $\h\in R_{t,n}$.
\end{proof}
\begin{theorem}
\label{t1}
  \begin{equation*}
 \rho^{\mr opt}_{\mf{A}_t}=\min_{\h\in\Psi^*_{1,n-1}\cap C(\mf{A}_t)} \frac{s_{1,0}}{s_{0,1}}=\min_{\h\in\Psi_{1,n-1}\cap C(\mf{A}_t)} \frac{s_{1,0}}{s_{0,1}}=1.
\end{equation*}
\end{theorem}
\begin{proof}
  Since $\Psi^*_{1,n-1}\subset\Psi_{1,n-1}$, $ \rho^{\mr opt}_{\mf{A}_t}\ge
  \min_{\h\in\Psi_{1,n-1}\cap C(\mf{A}_t)}
  \frac{s_{1,0}}{s_{0,1}}$. By Lemma \ref{l10}, we know that the
  optimal value of $\min_{\h\in\Psi_{1,n-1}\cap C(\mf{A}_t)}
  \frac{s_{1,0}}{s_{0,1}}$ can be achieved by any non-origin $\h\in
  R_{t,n}$. When $t=1$, $R_{1,n}\subset \Psi^*_n$. When $2\le t\le
  n-1$, There exists non-origin $\h\in R_{t,n}$ such that $\h\in \Psi^*_n$ \cite{M94a}. 
Hence,
  \begin{equation*}
  \rho^{\mr opt}_{\mf{A}_t}=\min_{\h\in\Psi_{1,n-1}\cap C(\mf{A}_t)} \frac{s_{1,0}}{s_{0,1}}=1.
\end{equation*}
The theorem is proved.\hfill
\end{proof}

The result $\rho^{\mr
  opt}_{\mf{A}_t}=1$ in Theorem \ref{t1} first appeared in Shamir's seminal
paper\cite{S79}. Here, we recover this result by meanings of
a systematic approach that makes use of the fundamental structure of
the problem.
With our approach, the above result can be generalized to the \emph{non-perfect}
secret-sharing problem. For this problem, in addition to the authorized sets of parties that
can recover the secret and the unauthorized sets of parties that know
nothing about the secret, there are some sets $\mc{A}\subset\N_2$
that know partial information about the secret, specifically,
$H(X_1,X_{\mc{A}})=H(X_1)+aH(X_{\mc{A}})$ for some $0<a<1$. Now for the
uniform secret-sharing problem\cite{IKS13}\cite{FHKP14}, the amount
information known by a set of parties depends only on the cardinality
of the set. To determine the optimal ratio of the problem,
we replace
$C(\mf{A}_t)$ in the $t$-threshold secret-sharing schemes with the
following region,
\begin{align*}
  C(\mb{a})=\{\h\in\hsp:\h(\{1\}\cup\mc{A})=\h(\{i\})+a_{|\mc{A}|}\h(\mc{A}),
  \mc{A}\subset\N_2\},
\end{align*}
where $\mb{a}\triangleq(a_0,\cdots,a_n-1)$ with $0\le
a_{n-1}\le\cdots\le a_0=0$. Note that $t$-threshold secret-sharing
is a special case of uniform secret-sharing with $a_i=1$
for $i\ge t$ and $a_i=0$ otherwise. The optimal information ratio for
uniform secret-sharing can be computed similarly as we do for
the case of $t$-threshold secret-sharing in Theorem \ref{t1}. The main result of
this paper, i.e.,  $\Psi_{1,n-1}=\overline{\Psi^*_{1,n-1}}$ asserts
that for the purpose of computing the optimal information ratio for
uniform secret-sharing,
Shannon-type information inequalities suffice.

\subsection{Further research}

We have proved in Theorem \ref{bfdp}, the main theorem, that $\overline{\Psi^*_p}=\Psi_p$ 
if and only if $p=[n]$ or $p=[1,n-1]$ for $p\in\mc{P}_n$, $n\ge 4$, i.e.,
$\overline{\Psi^*_p}$
is completely characterized by
Shannon-type information inequalities if and only if $p$ is the
$1$-partition or a $2$-partition with one of its blocks being a
singleton.
For those $p\in\mc{P}_n$ such
that $\overline{\Psi^*_p}=\Psi_p$, the characterization of 
$\overline{\Psi^*_p}$ is complete. However, further work is needed to
characterize 
$\Psi^*_p$ for those $p\in\mc{P}_n$ such that
$\overline{\Psi^*_p}\subsetneq\Psi_p$. For example, Theorem \ref{bfdp}
asserts that $\N=\{1,2,3,4\}$ together with
$p=\{\{1,2\},\{3,4\}\}$ gives the smallest example for which
$\overline{\Psi^*_p}\subsetneq\Psi_p$. Here, both $\overline{\Psi^*_p}$
and $\overline{\Gamma^*_4}$, where
$\overline{\Psi^*_p}\subset\overline{\Gamma^*_4}$, cannot be
completely characterized. Nevertheless, characterizing $\overline{\Psi^*_p}$ can be
regarded as an intermediate step toward characterizing $\overline{\Gamma^*_4}$.
The characterizations of $\overline{\Psi^*_p}$ may also be
useful for tackling other information theory problems with symmetrical
structures\cite{TCS13,AW15}. 

In the definition of $\Psi^*_p=\Gamma^*_n\cap \fix_p$ and
$\Psi_p=\Gamma_n\cap \fix_p$, the constraint $\fix_p$ is the fixed set
of a group action of a subgroup $\Sigma_p$ of symmetric subgroup
$\Sigma_n$ induced by the partition $p$. For further research, we can also
define for
any subgroup $\Sigma$ of $\Sigma_n$, the corresponding group action and its fixed set $\fix_\Sigma$, and then study
whether $\overline{\Psi^*_\Sigma}=\Psi_\Sigma$, where $\Psi^*_\Sigma=\Gamma^*_n\cap \fix_\Sigma$ and
$\Psi_\Sigma=\Gamma_n\cap \fix_\Sigma$.

For the entropy function $\h\in\Gamma^*_n$ of a secret-sharing scheme, according to
\cite[Propositions 2.1-2.3]{Cm97},
by the operation $f(\mc{A})\triangleq\h(\mc{A})/\h(\{1\})$ for any $\mc{A}\subset\N$, one
may absorb the random variable representing the secret to obtain a
polymatroid $f$ with ground set $\N_2$, and apply
Theorem 1 to the $[1,n-1]$-bipartite random variable, e.g. \cite{MBjr08}. 
In general, though $\Psi_p\supsetneq\overline{\Psi^*_p}$, when $p=[n_1,\cdots,n_t]$
with $n_1=1$, $t\ge3$ and $n\ge4$, the characterization of
$\overline{\Psi^*_p}$, which is easier than the characterization of
$\overline{\Gamma^*_n}$, would be useful for obtaining a tighter bound on the
information ratio for
an arbitrary multipartite secret-sharing\cite{PS00,FFP07}. 

\section*{Appendix}
\subsection*{A. Proof of Theorem \ref{bafoi}}

\begin{lemma}
\label{lfdbp}
  For $E_1,E_2\in\mc{E}_n$, $E_1\cap \fix_p=E_2\cap \fix_p$ if
  they are $p$-equivalent.
\end{lemma}
\begin{proof}
If $E_1$ and $E_2$ are $p$-equivalent, there exists $\sigma\in
\Sigma_p$ such that $E_2=\sigma(E_1)$. Then $E_2\cap 
\fix_p=\sigma(E_1)\cap \fix_p=\sigma(E_1)\cap\sigma(\fix_p)=\sigma(E_1\cap
\fix_p)=E_1\cap \fix_p$.
\end{proof}
It can be seen from Lemma \ref{vysid} that the converse of Lemma
\ref{lfdbp} is also true.
\begin{lemma}
  \label{vysid}
For $E_1,E_2\in\mc{E}_n$, if they
  are not $p$-equivalent, then neither $E_1\cap \fix_p\subset E_2\cap
  \fix_p$ nor $E_2\cap \fix_p\subset E_1\cap \fix_p$.
\end{lemma}
\begin{proof}
 See Appendix B.
\end{proof}

For $A\subset\hsp_n$, let $\relbd{A}\triangleq
  \overline{A}\setminus\ri{A}$ be the \emph{relative boundary} of $A$.
\begin{lemma}
\label{josd}
 For any facet $G\in\mc{G}_p$ of $\Psi_p$, there exists a facet $E\in\mc{E}_n$ of
  $\Gamma_n$ such that $G\subset E$.
\end{lemma}
\begin{proof}
We claim that $\relbd{\Psi_p}\subset \relbd{\Gamma_n}$. 
Assume
the contrary. Then there exists $\h\in \relbd{\Psi_p}$ but $\h\in
\ri{\Gamma_n}$ because $\Psi_p\subset\Gamma_n$. Since $\h\in
\ri{\Gamma_n}$, there exists $\epsilon>0$, such that
$\ball{\h}{\epsilon}\cap \fix^0_p\subset \ball{\h}{\epsilon}\cap
\hsp^0_n\subset \ri{\Gamma_n}\subset\Gamma_n$. 
Then $\ball{\h}{\epsilon}\cap \fix^0_p\subset\Gamma_n\cap \fix^0_p=\Psi_p$.
On the other hand, since $\h\in \relbd{\Psi_p}$,
$\ball{\h}{\epsilon}\cap \fix^0_p\not\subset\Psi_p$, a contradiction.

Note that
$\relbd{\Gamma_n}=\cup_{E\in\mc{E}_n}E$ and
$\relbd{\Psi_p}=\cup_{G\in\mc{G}_p}G$. As $\relbd{\Psi_p}\subset
\relbd{\Gamma_n}$, $G\subset
\relbd{\Gamma_n}$ for any $G\in\mc{G}_p$. 
For a fixed $G\in\mc{G}_p$, we now prove that for any $E\in\mc{E}_n$,
$G\cap E$ is a face of $G$. Let $P$ be the supporting hyperplane of
$\Gamma_n$ such that $E=\Gamma_n\cap P$. As the origin $O\in G\cap P$ and
$P^+\supset\Gamma_n\supset G$, $P$ is also a supporting
hyperplane of $G$. It follows that $G\cap P$ is a face of $G$. Since
$G\subset\Gamma_n$, $G\cap P=(G\cap\Gamma_n)\cap P=G\cap(\Gamma_n\cap
P)=G\cap E$. Therefore $G\cap E$ is a face of $G$. 

Finally, we prove that for any $G\in\mc{G}_p$, there must exist $E\in\mc{E}_n$ such that $G\cap
E=G$. Assume the contrary. Then for all $E\in\mc{E}_n$, $G\cap
E\subsetneq G$. Since we have shown that $G\cap E$ is a face of
$G$, it is a proper face of $G$ and $G\cap E\subset\relbd{G}$. It
follows that for any $E\in\mc{E}_n$ and
for any 
$\h\in\ri{G}$,  $\h\not\in G\cap E
$. Therefore $\h\not\in \cup_{E\in\mc{E}_n}(G\cap E)= G\cap
(\cup_{E\in\mc{E}_n}E)=G\cap\relbd{\Gamma_n}=G$, a contradiction. 
Hence there must exist $E$ such that $G\cap
E=G$ or $G\subset E$. 
\end{proof}

\begin{lemma}
\label{hd9db}
  For any $E\in\mc{E}_n$, $E\cap \fix_p$ is a facet of
      $\Psi_p$,
\end{lemma}
\begin{proof}
To prove $E\cap \fix_p$ is a facet of $\Psi_p$, we first prove that it is
a face of $\Psi_p$.  As $E$ is a face of $\Gamma_n$, there
exists a supporting hyperplane $P$ of $\Gamma_n$ such that
$E=\Gamma_n\cap P$. So
$E\cap \fix_p=\Gamma_n\cap P\cap \fix_p=\Psi_p\cap
P$. Since $P^+\supset\Psi_p$, $P$ is also a
supporting hyperplane of $\Psi_p$. Then $E\cap \fix_p$ is a
face of $\Psi_p$. We now prove that the face $E\cap \fix_p$ of $\Psi_p$ is indeed a facet. 
Assume $E\cap \fix_p$ is not a facet of $\Psi_p$, i.e., there exists a facet
$G$ of $\Psi_p$ such that $G\supsetneq(E\cap \fix_p)$. By Lemma~\ref{josd}, there exists a facet $E'$ of $\Gamma_n$ such that
$G\subset E'$. As $G\subset \fix_p$, $G\subset E'\cap \fix_p$ which implies
that $E\cap \fix_p\subset E'\cap \fix_p$. By Lemma \ref{vysid}, $E$ and $E'$
are $p$-equivalent. Hence by Lemma \ref{lfdbp}, $E\cap \fix_p=
E'\cap \fix_p$. It follows that
$G=E\cap \fix_p$ which contradicts the fact that $G\supsetneq(E\cap
\fix_p)$. Therefore $E\cap \fix_p$ is facet of $\Psi_p$. 
\end{proof}

\emph{Proof of Theorem \ref{bafoi}}
1) follows immediately from Lemma \ref{lfdbp} and  Lemma
  \ref{vysid}. We now prove 2).
By Lemma \ref{hd9db}, $\omega_p$ is a mapping from $\mc{E}_n$ to $\mc{G}_p$. For a facet $G\in\mc{G}_p$ of $\Psi_p$, by Lemma \ref{josd}, there exists a facet
$E$ of $\Gamma_n$ such that $G\subset E$. As $G\subset \fix_p$, $G\subset
E\cap \fix_p$. Since $E\cap \fix_p$ is a facet of $\Psi_p$, $G=E\cap \fix_p$
(because for a polyhedral cone, no facet can contain another facet), 
i.e., for each $G\in\mc{G}_p$, there exists $E\in\mc{E}_n$ such that
$\omega_p(E)=E\cap \fix_p=G$. Therefore $\omega_p$ is surjective.
\hfill\QQQ

\subsection*{B. Proof of Lemma \ref{vysid}}
For $p\in\mc{P}_n$, the set $\mathfrak{E}_p$ of all $p$-orbits of
facets of $\Gamma_n$ is a partition of $\mc{E}_n$, the set of all facets of
$\Gamma_n$. Therefore, there exists a partial order on $\{\mf{E}_p:p\in\mc{P}_n\}$.
\begin{lemma}
\label{joasd}
  For $p,p'\in\mc{P}_n$, if $p\le p'$, then 
  $\mf{E}_p\le\mf{E}_{p'}$.
\end{lemma}
\begin{proof}
  To prove this lemma, it suffices to prove that two $p$-equivalent
  facets are $p'$-equivalent if $p\le p'$. Let facets $E_i=E(\mc{I}_i,\mc{K}_i), i=1,2$ of $\Gamma_n$ be
  $p$-equivalent. By
  Lemma \ref{oiabd},
$\bs{\lambda}_{\mc{I}_1,p}=\bs{\lambda}_{\mc{I}_2,p}$ and $\bs{\lambda}_{\mc{K}_1,p}=\bs{\lambda}_{\mc{K}_2,p}$.
 Since $p\le p'$, any entry of  $\bs{\lambda}_{\mc{I}_i,p'},i=1,2$ is the
 summation of the entries of  $\bs{\lambda}_{\mc{I},p}$ as we discussed
 in Lemma \ref{sdbpo}. Therefore
 $\bs{\lambda}_{\mc{I}_1,p'}=\bs{\lambda}_{\mc{I}_2,p'}$. Similarly, $\bs{\lambda}_{\mc{K}_1,p'}=\bs{\lambda}_{\mc{K}_2,p'}$.
 It follows that $E_1$ and $E_2$ are $p'$-equivalent. Hence $\mf{E}_{p}\le\mathfrak{E}_{p'}$.
\end{proof}

For notational convenience, let $p_0$ be a virtual partiton such that
$\{\N\}\le p_0$ and $\mc{E}_n$ be the only $p_0$-orbit. 
Hence all facets of $\Gamma_n$ are $p_0$-equivalent.
For $p,p'\in\mc{P}_n\cup\{p_0\}$ such that $p\le p'$, by Lemma
\ref{joasd}, each $\mc{E}\in
\mf{E}_{p'}$ can be partitioned into some $p$-orbits. For a particular
$\mc{E}\in\mf{E}_{p'}$, let $\mf{E}_{\mc{E},p}$ be the family
of all such $p$-orbits. Note that $\mf{E}_{\mc{E},p}$ is a subset of
$\mf{E}_{p}$ and the union of all members of $\mf{E}_{\mc{E},p}$ is $\mc{E}$. 
\begin{definition}[Isolation]
\label{kodf8}
Let $p, p'\in\mc{P}_n\cup\{p_0\}$ such that $p\le p'$. For
$\mc{F}\in\mf{E}_{p'}$ and $\mc{E}\in\mf{E}_{\mc{F},p}$, 
$\mb{i}\in\Psi_p$ is called an \emph{isolation} of $\mc{E}$ in $\mf{E}_{\mc{F},p}$ if for any
$E\in\mc{E}$, $\mb{i}\notin E\cap \fix_p$ but for all
$E'\in(\mc{F}\setminus\mc{E})$, $\mb{i}\in E'\cap \fix_p$. We also say
$\mc{E}$ has an isolation $\mb{i}$ in $\mf{E}_{\mc{F},p}$.
\end{definition}

\begin{lemma}
\label{pr0e8}
If each $\mc{E}\in\mf{E}_{\mc{F},p}$ has an isolation in
$\mf{E}_{\mc{F},p}$, for any non $p$-equivalent $E_1, E_2\in\mc{F}$, 
neither $E_1\cap \fix_p\subset
E_2\cap \fix_p$ nor $E_2\cap \fix_p\subset E_1\cap \fix_p$.
\end{lemma}
\begin{proof}
  By the definition, if $\mc{E}$ has an isolation in
$\mf{E}_{\mc{F},p}$, for all $E\in\mc{E}$ and all
$E'\in\mc{F}\setminus\mc{E}$, $E'\cap \fix_p\not\subset E\cap
\fix_p$. Then the lemma follows.
\end{proof}

\begin{example}
  \label{nlas9} 
Let $p'=p_0, p=\{\N\}$. By Example \ref{jaodm}, 
$\mf{N}_p=\{[(1),(0)]\}\cup\{[(2),(k)]:k=0,\cdots,n-2\}$.
It can be checked that for the uniform matroid $U_{n,n}$, whose rank
function is $U_{n,n}(\mc{A})=|\mc{A}|,\mc{A}\subset\N$, $\s(U_{n,n},p)$ satisfies,
\begin{align*}
  s_n&>s_{n-1},\\
  2s_{i+1}&=s_{i}+s_{i+2},\ i=0,\cdots,n-2.
\end{align*}
Hence, 
for any $E\in\mc{E}_p([(1),(0)])$, the
uniform matroid $U_{n,n}\notin E\cap \fix_{[n]}$, but for any $E\in\mc{E}_p(\bs{\lambda})$
with $\bs{\lambda}\in\mf{N}_p$ and $\bs{\lambda}\neq [(1),(0)]$,
$U_{n,n}\in E\cap \fix_p$. Therefore $U_{n,n}$ is an isolation of $\mc{E}_p([(1),(0)])$
in $\mf{E}_{\mc{E}_n,n}$. 

Similarly, for a particular
$k=0,\cdots,n-2$, uniform matroid $U_{k+1,n}$ with rank function $U_{k+1,n}(\mc{A})=\min\{k+1,
|\mc{A}|\},\mc{A}\subset\N$, $\s(U_{k+1,n},p)$ satisfies,
\begin{align*}
 2 s_{k+1}&>s_{k}+s_{k+2},\\
  s_n&=s_{n-1},\\
  2s_{i+1}&=s_{i}+s_{i+2},\ i=0,\cdots,n-2,\ i\neq k.
\end{align*}
So $U_{k+1,n}\notin
E\cap \fix_p$, 
for any $E\in\mc{E}_p([(2),(k)])$ but for any $E\in\mc{E}_p(\bs{\lambda})$
with $\bs{\lambda}\in\mf{N}_p$ and $\bs{\lambda}\neq [(2),(k)]$,
$U_{k+1,n}\in E\cap \fix_p$. Therefore $U_{k+1,n}$ is an isolation of $\mc{E}_p([(2),(k)])$
in $\mf{E}_{\mc{E}_n,n}$.\hfill\QQQ
\end{example}

Example \ref{nlas9} discussed the case when $p$ is the one-partition. The
following lemma studies the case when $p$ is a two-partition which is
covered by the one-partition $p'=\N$.
\begin{lemma}
\label{u89sd}
  Let $p'=\N$ and $p=\{\N_1,\N_2\}\in\mc{P}_{n}$. 
  For a particular $\bs{\lambda}'\in\mf{N}_{p'}$ and each
  $\bs{\lambda}\in\mf{N}_p$ such that
  $\mc{E}_p(\bs{\lambda})\in\mf{E}_{\mc{E}_{p'}(\bs{\lambda}'),p}$, $\mc{E}_p(\bs{\lambda})$ has
  an isolation in $\mf{E}_{\mc{E}_{p'}(\bs{\lambda}'),p}$.
\end{lemma}
The proof of Lemma \ref{u89sd} is given after Lemmas \ref{xlkds}.

The following discussion facilitates the proof of Lemma \ref{u89sd}. For $p'=\N$ and $p=\{\N_1,\N_2\}\in\mc{P}_{n}$, for
$\mc{E}_p(\bs{\lambda})\in\mf{E}_p$ with $\bs{\lambda}\in\mf{N}_p$, let
$J_p(\bs{\lambda})\subset\Gamma_n$ be the set of all isolations of
$\mc{E}_p(\bs{\lambda})$ in $\mf{E}_{\mc{E}_{p'}(\bs{\lambda}'),p}$ for some
$\bs{\lambda}'\in\mf{N}_{p'}$. Note that because each 
$\mc{E}_p(\bs{\lambda})$ belongs to a unique
$\mf{E}_{\mc{E}_{p'}(\bs{\lambda}'),p}$, we do not need to specify 
which $\mf{E}_{\mc{E}_{p'}(\bs{\lambda}'),p}$ the $p$-orbit
$\mc{E}_p(\bs{\lambda})$ belongs to in the notation $J_p(\bs{\lambda})$.

For $\mc{E}_{p'}([(1),(0)])$, it can be seen that
$\mf{E}_{\mc{E}_{p'}([(1),(0)]),p}$ contains two members, 
\begin{itemize}
\item $\mc{E}_p([(1,0),(0,0)])$, the $p$-orbit of all $E(i)$ such that $i\in\N_1$
and 
\item $\mc{E}_p([(0,1),$ $(0,0)])$, the $p$-orbit of all $E(i)$ such that
$i\in\N_2$.
\end{itemize}
For $\mi\in  J_p([(1,0),(0,0)])$, by the definition of an isolation,
$\mi\not\in E\cap S_p$ for $E\in\mc{E}_p([(1,0),(0,0)])$. Then by
\eqref{soidv}, $\s(\mi,p)$ satisfies
$s_{n_1,n_2}>s_{n_1-1,n_2}$. Similarly, as $\mi\in E\cap S_p$ for
$E\in\mc{E}_p([(0,1),(0,0)])$, $\ s_{n_1,n_2}=s_{n_1,n_2-1}$. Therefore,
\begin{align}
  J_p([(1,0),(0,0)])=&\{\mi\in\Psi_p:
s_{n_1,n_2}>s_{n_1-1,n_2},\ s_{n_1,n_2}=s_{n_1,n_2-1}\}.\label{eq:5}
\end{align}
Region $ J_p([(0,1),(0,0)])$ can be obtained from \eqref{eq:5} by symmetry.
For $\mc{E}_p([(2),(k)]),k=0,\cdots,n-2$, the members in
$\mf{E}_{\mc{E}_p([(2),(k)]),p}$ are 
\begin{itemize}
\item $\mc{E}_p([(1,1),(k_1,k_2)])$, $(k_1,k_2)\in\mc{M}_p$ with
$k_1+k_2=k$ and $k_i\neq
n_i,i=1,2$; 
\item $\mc{E}_p([(2,0),(k_1,k_2)])$, $(k_1,k_2)\in\mc{M}_p$ with
$k_1+k_2=k$ and $k_1\neq
n_1-1,n_1$;
\item  $\mc{E}_p([(0,2),(k_1,k_2)])$, $(k_1,k_2)\in\mc{M}_p$ with
$k_1+k_2=k$, $k_2\neq
n_2-1,n_2$. 
\end{itemize}
Therefore, by the definition of an isolation and \eqref{kobix}
and \eqref{xiobd}, for fixed $(k_1,k_2)$, we have
\begin{align}
J_p&([(1,1),(k_1,k_2)])\nonumber\\
&=\{\h\in\Psi_p:
s_{k_1+1,k_2}+s_{k_1,k_2+1}>s_{k_1,k_2}+s_{k_1+1,k_2+1},\nonumber\\
    &s_{i+1,j}+s_{i,j+1}=s_{i,j}+s_{i+1,j+1},\ (i,j)\in\mc{M}_p,i\neq
    n_1,j\neq n_2, i+j=k_1+k_2,\ (i,j)\neq
    (k_1,k_2),\nonumber\\
    &2s_{i+1,j}=s_{i,j}+s_{i+2,j},\ (i,j)\in\mc{M}_p, i\neq n_1-1,n_1,i+j=k_1+k_2,\nonumber\\
    &2s_{i,j+1}=s_{i,j}+s_{i,j+2},\ (i,j)\in\mc{M}_p,j\neq n_2-1,n_2, i+j=k_1+k_2\}.\label{eq:10}
\end{align}
Similarly, for fixed $(k_1,k_2)$, we have
\begin{align}
J_p&([(2,0),(k_1,k_2)]) \nonumber\\
 &=\{\h\in\Psi_p:
2s_{k_1+1,k_2}>s_{k_1,k_2}+s_{k_1+2,k_2},\nonumber\\
    &s_{i+1,j}+s_{i,j+1}=s_{i,j}+s_{i+1,j+1},\ (i,j)\in\mc{M}_p,i\neq
    n_1, j\neq n_2, i+j=k_1+k_2,\nonumber\\
    &2s_{i+1,j}=s_{i,j}+s_{i+2,j},\ (i,j)\in\mc{M}_p,i\neq n_1-1,n_1, i+j=k_1+k_2,(i,j)\neq(k_1,k_2),\nonumber\\
    &2s_{i,j+1}=s_{i,j}+s_{i,j+2},\ (i,j)\in\mc{M}_p,j\neq n_2-1,n_2, i+j=k_1+k_2.\label{eq:81}\}
\end{align}
Region $ J_p([(0,2),(k_1,k_2)])$ can be obtained from
\eqref{eq:81} by symmetry. 

Thus, to prove Lemma \ref{u89sd} is indeed to prove
that  $J_p(\bs{\lambda})$ for all $\bs{\lambda}\in\mf{N}_p$ are
nonempty. Before we prove Lemma \ref{u89sd}, we first present a
technical lemma.  
\begin{lemma}
\label{xlkds}
Let $p=\{\N_1,\N_2\}\in\mc{P}_{n}$. For a fixed $(k_1,k_2)\in\mc{M}_p$ such that
$k_i\neq n_i$, let $l_m=k_m+1,m=1,2$. Then $\mi\in \fix_p$ with
$\s(\mi,p)$ satisfying 
\begin{equation}
  \label{eq:8}
  s_{i,j}=
  \begin{cases}
    &in_2+jn_1-ij\quad 0\le i\le l_1, 0\le j\le l_2 \text{ or } l_1+1\le
    i\le n_1, l_2+1\le
    j\le n_2\\
    &jl_1-(j-l_2)\max\{0,l_1-i-1\}+i(n_2-j)+j(n_1-l_1)\quad 0\le i\le
    l_1,l_2+1\le j\le n_2\\
    &il_2-(i-l_1)\max\{0,l_2-j-1\}+j(n_1-i)+i(n_2-l_2)\quad l_1+1\le i\le
    n_1,0\le j\le l_2
  \end{cases}
\end{equation}
is in $ J_p([(1,1),(k_1,k_2)])$.
\end{lemma}
The proof of Lemma \ref{xlkds} will be deferred to the end of this
appendix. With this lemma, we are now ready to prove Lemma \ref{u89sd}.

\emph{Proof of Lemma \ref{u89sd}}
  We first consider $\mi\in J_p([(1,0),(0,0)])$. Let 
  \begin{equation}
\label{d037h}
    \mi(\mc{A})=|\mc{A}\cap\N_1|,\ \mc{A}\subset\N.
  \end{equation}
  Then $\mi$ is a matroid on $\N$ with $U_{n_1,n_1}$ as its submatroid on
  $\N_1$ and elements in $\N_2$ as loops.
As the entries of
  $\s(\mi,p)$, by \eqref{d037h}, $s_{i,j}=i$ for all $(i,j)\in\N_p$.
 It can be checked by
  \eqref{eq:5} that $\mi\in J_p([(1,0),(0,0)])$. Similarly,
  we can prove $ J_p([(0,1),(0,0)])$ is also nonempty.

Now we consider $J_p([(2,0),(k_1,k_2)])$ with $(k_1,k_2)\in\mc{M}_p$ and
$k_1\neq n_1-1,n_1$.
Let 
  \begin{equation}
\label{siud}
    \mi(\mc{A})=\min\{|\mc{A}\cap\N_1|,k_1+1\},\ \mc{A}\subset\N.
  \end{equation}
  Then $\mi$ is a matroid on $\N$ with $U_{k_1+1,n_1}$ as its submatroid on
  $\N_1$ and elements in $\N_2$ as loops. As the entries of
  $\s(\mi,p)$, by \eqref{siud}, $s_{i,j}=\min\{i,k_1+1\}$ for all $(i,j)\in\mc{M}_p$.
It can be checked by \eqref{eq:81}
  that
  $\mi\in J_p([(2,0),(k_1,k_2)])$. Similarly,
  we can prove $ J_p([(0,2),(k_1,k_2)])$ with $(k_1,k_2)\in\mc{M}_p$ and
$k_2\neq n_2-1,n_2$ is also nonempty.

Finally, by Lemma \ref{xlkds}, for any $(k_1,k_2)\in\mc{M}_p$ with
$k_i\neq n_i,i=1,2$, $J_p([(1,1),(k_1,k_2)])$ is nonempty. 
\hfill\QQQ

Lemma \ref{u89sd} can be generalized from the two-partition case to the
multi-partition case, which will be stated in 
the next lemma whose
proof will also be deferred to the end of this appendix.

\begin{lemma} 
\label{fd802}
Let $p,p'\in\mc{P}_n$ and $p'$ covers $p$. 
For any
$\mc{E}_{p'}(\bs{\lambda}')\in\mf{E}_{p'}$ with
$\bs{\lambda}'\in\mf{N}_{p'}$, for each $\bs{\lambda}\in\mf{N}_{p}$,
$\mc{E}_{p}(\bs{\lambda})\in
\mf{E}_{\mc{E}_{p'}(\bs{\lambda}'),p}$ has
an isolation in $\mathfrak{E}_{\mc{E}_{p'}(\bs{\lambda}'),p}$. 
\end{lemma}

\emph{Proof of Lemma \ref{vysid}} We prove the lemma by induction
on the lattice $\mc{P}_n$. First, we show that this lemma is
true for $p=\{\N\}$. 
From Example \ref{nlas9}, we see that for any $\bs{\lambda}\in\mf{N}_n$, $\mc{E}_n(\bs{\lambda})$ has an isolation
$U_{k,n}$ in $\mathfrak{E}_n$. Then the lemma follows immediately from
Lemma \ref{pr0e8}. 

Then it is sufficient to prove that if the lemma is true for $(t-1)$-partition
$p'\in\mc{P}_{n}$, it is also true for any $t$-partition $p\in\mc{P}_{n}$ covered by
$p'$. Now assume that
the lemma is true for a fixed $(t-1)$-partition
$p'$, and consider any $t$-partition
$p$ covered by $p'$.
Let $E_1,E_2$ be two facets of $\Gamma_n$ which are not
$p$-equivalent. If $E_1,E_2$ are not $p'$-equivalent, by
the induction hypothesis, $E_1\cap \fix_{p'}\not\subset E_2\cap \fix_{p'}$. As $p\le
p'$, by Proposition \ref{sdbpo}, $\fix_{p'}\subset \fix_{p}$. So if
$E_1\cap \fix_{p}\subset E_2\cap \fix_{p}$, then $E_1\cap \fix_{p}\cap
\fix_{p'}\subset E_2\cap \fix_{p}\cap \fix_{p'}$, which implies that
$E_1\cap \fix_{p'}\subset E_2\cap \fix_{p'}$, a contradiction. Therefore
$E_1\cap \fix_{p}\not\subset E_2\cap \fix_{p}$. Similarly, $E_2\cap
\fix_{p}\not\subset E_1\cap \fix_{p}$. 

Then it remains to prove the case that
$E_1$ and $E_2$ are $p'$-equivalent but not $p$-equivalent. 
To this end, in light of Lemma \ref{pr0e8}, it suffices to prove that, for any
$\mc{E}_{p'}(\bs{\lambda}')\in\mf{E}_{p'}$ with
$\bs{\lambda}'\in\mf{N}_{p'}$, for each $\bs{\lambda}\in\mf{N}_{p}$,
$\mc{E}_{p}(\bs{\lambda})\in
\mf{E}_{\mc{E}_{p'}(\bs{\lambda}'),p}$ has
an isolation in $\mathfrak{E}_{\mc{E}_{p'}(\bs{\lambda}'),p}$. This is
what we have proved in Lemma \ref{fd802}. Then this lemma follows.
\hfill\QQQ

\emph{Proof of Lemma \ref{xlkds}}
To prove that $\mi\in \fix_p$ with $\s(\mi,q)$ satisfying \eqref{eq:8} is
in $J_p([(1,1),(k_1,k_2)])$, we need to check that, first
$\mi\in\Gamma_n$, i.e., $\s(\mi,q)$ satisfies
\begin{align}
  \label{eq:12}
  &s_{n_1,n_2}\ge s_{n_1-1,n_2}\\
 &s_{n_1,n_2}\ge s_{n_1,n_2-1}\label{eq:121}\\
&2s_{i+1,j}\ge s_{i,j}+s_{i+2,j}\quad 0\le i\le n_1-2, 0\le j\le n_2 \label{eq:122}\\
&2s_{i,j+1}\ge s_{i,j}+s_{i,j+2}\quad 0\le i\le n_1, 0\le j\le n_2-2 \label{eq:123}\\
&s_{i+1,j}+s_{i,j+1}\ge s_{i,j}+s_{i+1,j+1}\quad 0\le i\le n_1, 0\le j\le n_2-2 \label{eq:124}
\end{align}
(cf. \eqref{soidv}-\eqref{xiobd}) and second, it satisfies \eqref{eq:10}, i.e.,
\begin{align}
\label{eq:101}
&s_{l_1,l_2-1}+s_{l_1-1,l_2}>s_{l_1-1,l_2-1}+s_{l_1,l_2}\\
 \label{eq:102}   &s_{i+1,j}+s_{i,j+1}=s_{i,j}+s_{i+1,j+1}\quad (i,j)\in\mc{M}_p;i\neq
    n_1;j\neq n_2; i+j=l_1+l_2-2;\ (i,j)\neq
    (l_1-1,l_2-1)\\
    \label{eq:103}&2s_{i+1,j}=s_{i,j}+s_{i+2,j}\quad (i,j)\in\mc{M}_p; i\neq n_1-1,n_1;i+j=l_1+l_2-2\\
 \label{eq:104}  &2s_{i,j+1}=s_{i,j}+s_{i,j+2}\quad (i,j)\in\mc{M}_p; j\neq n_2-1,n_2;i+j=l_1+l_2-2
\end{align}
Note that the above are obtained by replacing $k_m$ by $l_m-1,m=1,2$
in \eqref{eq:10}.

Now we do the checking. According to \eqref{eq:8}, let 
\begin{itemize}
\item $\mc{A}=\{(i,j): 0\le i\le l_1, 0\le  j\le
l_2\}$, 
\item $\mc{B}=\{(i,j):0\le i\le
    l_1,l_2+1\le j\le n_2\}$,
\item $\mc{C}=\{(i,j): l_1+1\le i\le
    n_1,0\le j\le l_2\}$ and
\item $\mc{D}=\{(i,j):l_1+1\le
    i\le n_1, l_2+1\le
    j\le n_2\}$
\end{itemize}
be the four blocks of a partition of $\mc{M}_p$.

For an inequality in \eqref{eq:12}-\eqref{eq:124}, if the indices of
all the terms are in a particular block, say $\mc{A}$, we say that the
inequality is in $\mc{A}$; otherwise, if the indices are in more than one
block, say $\mc{A}$ and $\mc{B}$, we say the inequality is between
$\mc{A}$ and $\mc{B}$. If the condition that guarantees the existence
of the inequalities does not hold, then the inequality does not need
to be considered for the chosen set of parameter.

We will check that \eqref{eq:12}-\eqref{eq:124} hold, which proves that
$\mi\in\Gamma_n$. During the process, by noting which inequality hold
with equality and which hold strictly, we can verify the
equalities or inequalities in \eqref{eq:101}-\eqref{eq:104} along the
way. The details are given below, where the checking of  \eqref{eq:12}-\eqref{eq:124} are organized according to the indices involved in the inequalities.  Table II indicates how the inequalities and equalities in (67)-(70) are verified in the process.  For example, \eqref{eq:101} is verified under Inequalities in $\cal A$, and \eqref{eq:102} is verified under Inequalities in $\cal B$, Inequalities in $\cal C$, Inequalities in $\cal A$ and $\cal B$, and Inequalities in $\cal A$ and $\cal C$ for the corresponding ranges of $(i,j)$, respectively.

The result of the check are listed in Table \ref{tab:2}. Then we
can prove that $\mi\in J_p([(1,1),(k_1,k_2)])$. The details are given below.

\noindent\underline{Inequalities in $\mc{A}$:}

\begin{itemize}
\item \eqref{eq:122} holds with equality since the first case in
  \eqref{eq:8} gives
\begin{equation}
\label{eq:131}
 s_{i,j}= in_2+jn_1-ij
\end{equation}
which is linear in $i$ for a fixed $j$.
When $i=l_1-2,j=l_2$, it proves \eqref{eq:103} holds for such $i,j$.

\item \eqref{eq:123}: it holds with equality since $s_{i,j}$
  above is linear in $j$ for a fixed $i$.
When $i=l_1$ and $j=l_2-2$, it proves that \eqref{eq:104} holds for
such $i$ and $j$.

\item \eqref{eq:124} holds strictly since
  \begin{align*}
    (s_{i+1,j}&+s_{i,j+1})- (s_{i,j}+s_{i+1,j+1})\\
=\ &((-(i+1)j)+(-i(j+1)))-(-ij-(i+1)(j+1))\\
=\ &1.
\end{align*}
When $i=l_1-1,j=l_2-1$, it proves that \eqref{eq:101} holds.

\item If $l_1=n_1$ and $l_2=n_2$, \eqref{eq:12} is in $\mc{A}$ and it holds with equality since
$s_{n_1,n_2}=s_{n_1-1,n_2}=n_1n_2$.

\item If $l_1=n_1$ and $l_2=n_2$, \eqref{eq:121} is in $\mc{A}$ and it holds with equality since
$s_{n_1,n_2}=s_{n_1-1,n_2}=n_1n_2$.

\end{itemize}

\noindent\underline{Inequalities in $\mc{B}$:}

The second case in \eqref{eq:8} can be written in two subcases ($0\le
i\le l_1-1$ and $i=l_1$)
as
\begin{align}
  \label{eq:132}
  s_{i,j}=\ & l_1l_2+j-l_2+i(n_2-l_2)+j(n_1-l_1)\quad\text{if } 0\le i
  \le l_1-1,\\
\label{eq:133}
s_{l_1,j}=\ & l_1n_2+j(n_1-l_1).
\end{align}

\begin{itemize}
\item \eqref{eq:122} with $i\le l_1-3$ holds with equality since \eqref{eq:132}
is linear in $i$ for a fixed $j$.
When $i\le l_1-3$ and $i+j=l_1+l_2-2$, it proves 
\eqref{eq:103} holds for such $i,j$.

\item \eqref{eq:122} with $i=l_1-2$ holds strictly since by
  \eqref{eq:132} and \eqref{eq:133}
\begin{align*}
  2s_{l_1-1,j}&-s_{l_1-2,j}-s_{l_1,j}\\
=\ &2( l_1l_2+j-l_2+(l_1-1)(n_2-l_2)+j(n_1-l_1))\\ \ &-(
l_1l_2+j-l_2+(l_1-2)(n_2-l_2)+j(n_1-l_1))-(l_1n_2+j(n_1-l_1))\\
=\ &j-l_2>0.
\end{align*}

\item \eqref{eq:123} holds with equality since
  \eqref{eq:132} is linear in $j$ for a fixed $i$ and \eqref{eq:133} is linear in $j$.
When $i\le l_1-3$ and $i+j=l_1+l_2-2$, it proves 
\eqref{eq:104} holds for such $i,j$.

\item \eqref{eq:124} with $i\le l_1-2$ holds with equality since
by \eqref{eq:132} 
\begin{align*}
  (s_{i+1,j}&+s_{i,j+1})- (s_{i,j}+s_{i+1,j+1})\\
=\ & (s_{i,j+1}-s_{i,j})- (s_{i+1,j+1}-s_{i+1,j})\\
=\ & (1+n_1-l_1)-(1+n_1-l_2)=0.
\end{align*}
When $i\le l_1-3$ and $i+j=l_1+l_2-2$, it proves 
\eqref{eq:102} holds for such $i,j$.

\item \eqref{eq:124} with $i= l_1-1$ holds strictly since by \eqref{eq:132} and \eqref{eq:133}
\begin{align*}
  (s_{l_1,j}&+s_{l_1-1,j+1})- (s_{l_1-1,j}+s_{l_1,j+1})\\
=\ & (s_{l_1-1,j+1}-s_{l_1-1,j})- (s_{l_1,j+1}-s_{i_1,j})\\
=\ & (1+n_1-l_1)-(n_1-l_2)=1.
\end{align*}

\item If $l_1=n_1$, \eqref{eq:12} is in this block and it holds with equality since by
  \eqref{eq:132} and \eqref{eq:133},
  $s_{n_1,n_2}=s_{n_1-1,n_2}=n_1n_2$.

\item If $l_1=n_1$ and $l_2\le n_2-2$, \eqref{eq:121} is in this block and it holds with equality since by
\eqref{eq:133},
  $s_{n_1,n_2}=s_{n_1,n_2-1}=n_1n_2$.
\end{itemize}

\noindent\underline{Inequalities in $\mc{C}$:}

These inequalities are symmetrical to those inequalities in
$\mc{B}$. The details are omitted here.

\noindent\underline{Inequalities in $\mc{D}$:}

\begin{itemize}
\item \eqref{eq:122}, \eqref{eq:123} and \eqref{eq:124} hold by the
  same reason as they in $\mc{A}$.

\item If $l_1\le n_1-2$ and $l_2\le n_2-1$ \eqref{eq:12} is in this
  block and it holds by the same reason in $\mc{A}$.

\item If $l_1\le n_1-1$ and $l_2\le n_2-2$, \eqref{eq:121} is in this
  block and it holds by the same reason in $\mc{A}$.
\end{itemize}

\noindent\underline{Inequalities between $\mc{A}$ and $\mc{B}$:}

\begin{itemize}
\item \eqref{eq:123} with $i\le l_1-1$, $j=l_2-1$ holds since by \eqref{eq:131} and \eqref{eq:132},
  \begin{align*}
     2s_{i,l_2}&-s_{i,l_2-1}-s_{i,l_2+1}\\
=\ &
2(in_2+l_2n_1-il_2)-(in_2+(l_2-1)n_1-i(l_2-1))-(l_1l_2+1+i(n_2-l_2)+(l_2+1)(n_1-l_1))\\
=\ &l_1-i-1 \ge 0.
  \end{align*}
Furthermore, when $i= l_1-1$ and $j=l_2-1$, it holds with equality which proves
\eqref{eq:104} for such $i$ and $j$.

\item \eqref{eq:123} with $i= l_1$, $j=l_2-1$ holds with equality since by \eqref{eq:131} and \eqref{eq:133},
  \begin{align*}
     2s_{l_1,l_2}&-s_{l_1,l_2-1}-s_{l_1,l_2+1}\\
=\ &
2(l_1n_2+l_2n_1-l_1l_2)-(l_1n_2+(l_2-1)n_1-l_1(l_2-1))-(l_1n_2+(l_2+1)(n_1-l_1))\\
=\ &0.
  \end{align*}

\item \eqref{eq:123} with $i\le l_1-1$, $j=l_2$ holds with equality since by in \eqref{eq:131} and \eqref{eq:132},
  \begin{align*}
     2s_{i,l_2+1}&-s_{i,l_2+2}-s_{i,l_2}\\
=\ & 2(l_1l_2+1+i(n_2-l_2)+(l_2+1)(n_1-l_1))\\ \ &-(l_1l_2+2+i(n_2-l_2)+(l_2+2)(n_1-l_1))-(in_2+l_2n_1-il_2)\\
=\ &0;
  \end{align*}

When $i=l_1-2$ and $j=l_2$, it proves
\eqref{eq:104} for such $i$ and $j$.

\item \eqref{eq:123} with $i= l_1$, $j=l_2$ holds with equality since by in \eqref{eq:131} and \eqref{eq:133},
  \begin{align*}
     2s_{l_1,l_2+1}&-s_{l_1,l_2+2}-s_{l_1,l_2}\\
=\ & 2(l_1n_2+(l_2+1)(n_1-l_1))-(l_1n_2+(l_2+2)(n_1-l_1))-(l_1n_2+l_2n_1-l_1l_2)\\
=\ &0;
  \end{align*}

\item \eqref{eq:124} with $i\le l_1-2$, $j=l_2$ holds with equality since by \eqref{eq:131} and \eqref{eq:132},
  \begin{align*}
     (s_{i+1,l_2}&+s_{i,l_2+1})-(s_{i,l_2}+s_{i+1,l_2+1}) \\
=\ & (s_{i+1,l_2}-s_{i,l_2})-(s_{i+1,l_2+1}-s_{i,l_2+1}) \\
=\ & (n_2-l_2)-(n_2-l_2)=0.
  \end{align*}

When $i=l_1-1$ and $i+j=l_1+l_2-2$, it proves \eqref{eq:102} holds for
such $i$ and $j$.

\item \eqref{eq:124} with $i= l_1-1$, $j=l_2$ holds strictly since by \eqref{eq:131} and \eqref{eq:133},
  \begin{align*}
   (s_{l_1,l_2}&+s_{l_1-1,l_2+1})-(s_{l_1-1,l_2}+s_{l_1,l_2+1}) \\
=\ & (s_{l_1,l_2}-s_{l_1-1,l_2})-(s_{l_1,l_2+1}-s_{l_1-1,l_2+1}) \\
=\ & (n_2-l_2)-(n_2-l_2-1)=1.
  \end{align*}

\item If $l_1=n_1$ and $l_2=n_2-1$, \eqref{eq:121} is in this case and it holds with
  equality since by \eqref{eq:131} and \eqref{eq:133}, $s_{n_1,n_2}=s_{n_1,n_2-1}=n_1,n_2$.
\end{itemize}

\begin{table}
  \centering
  \begin{tabular}{|c|c|l|l|l|l|}
    \hline
    \hline
    &Inequalities&Inequalities in $\mc{B}$&Inequalities in
    $\mc{C}$&Inequalities between&Inequalities between
    \\
    &in $\mc{A}$&&& $\mc{A}$ and $\mc{B}$& $\mc{A}$ and $\mc{C}$
    \\
\hline
\hline
  \eqref{eq:101}&$\cdot$& & & &\\
\hline
  \eqref{eq:102}& &$\max\{0,n_1-l_1-l_2+1\}$&$(i,j):l_1+1\le i\le $
  &$(i,j)=$&$(i,j)=$\\
& &$\le i\le l_1-3$&$ \min\{n_1-1,l_1+l_2-2\}$
  &$(l_1-2,l_2)$&$(l_1,l_2-2)$\\
& &$j=l_1+l_2-2-i$&$j=l_1+l_2-2-i$& &\\

\hline
  \eqref{eq:103}& &$\max\{0,n_1-l_1-l_2+2\}$&$l_1+1\le i\le $ &
  &$(i,j)=$\\
&$(l_1-2,l_2)$ &$\le i\le
  l_1-3$&$\min\{n_1-2,l_1+l_2-2\}$ & &$(l_1-1,l_2-1)$\\
& &$j=l_1+l_2-2-i$&$j=l_1+l_2-2-i$& &or $(l_1,l_2-2)$\\

\hline
  \eqref{eq:104}& &$\max\{0,n_1-l_1-l_2\}$&$l_1+1\le i\le $&$(i,j)=$  &\\
&$(l_1,l_2-2)$&$\le i\le
  l_1-3$&$\min\{n_1,l_1+l_2-2\}$&$(l_1-2,l_2)$ or&\\
& &$j=l_1+l_2-2-i$&$j=l_1+l_2-2-i$&$(l_1-1,l_2-1)$  &\\
\hline
  \end{tabular}
  \caption{Verification of \eqref{eq:101}-\eqref{eq:104}}
  \label{tab:2}
\end{table}


\noindent\underline{Inequalities between $\mc{A}$ and $\mc{C}$:}

These inequalities are symmetrical to those inequalities between
$\mc{A}$  and
$\mc{B}$. The details are omitted here.

\noindent\underline{Inequalities between $\mc{B}$ and $\mc{D}$:}
\begin{itemize}
\item \eqref{eq:122} with $i=l_1-1$ holds with equality since by
  \eqref{eq:131}, \eqref{eq:132} and \eqref{eq:133},      
  \begin{align*}
    2s_{l_1,j}&-s_{l_1-1,j}-s_{l_1+1,j}\\
=\ &2(l_1n_2+j(n_1-l_1))-(l_1l_2+j-l_2+(l_1-1)(n_2-l_2)+j(n_1-l_1))\\ &-((l_1+1)n_2+jn_1-(l_1+1)j)\\
=\ &0.
  \end{align*}

\item \eqref{eq:122} with $i=l_1$ holds with equality since by
  \eqref{eq:131} and \eqref{eq:133},      
  \begin{align*}
    2s_{l_1+1,j}&-s_{l_1,j}-s_{l_1+2,j}\\
=\ &2((l_1+1)n_2+jn_1-(l_1+1)j)-(l_1n_2+j(n_1-l_1))-((l_1+2)n_2+jn_1-(l_1+2)j)\\
=\ &0.
  \end{align*}

\item \eqref{eq:124} with $i=l_1$ holds strictly since by \eqref{eq:131} and \eqref{eq:133}, 
  \begin{align*}
    (s_{l_1+1,j}&+s_{l_1,j+1})-(s_{l_1,j}+s_{l_1+1,j+1})\\
=\ &(s_{l_1,j+1}-s_{l_1,j})-(s_{l_1+1,j+1}-s_{l_1,j+1})\\
=\ &(n_1-l_1)-(n_1-l_1-1)=1.
  \end{align*}

\item If $l_1=n_1-1$, \eqref{eq:12} belongs to this case and it holds since
  by \eqref{eq:131} and \eqref{eq:133}, $s_{n_1,n_2}=s_{s_{n_1-1,n_2}}=n_1n_2$.
\end{itemize}

\noindent\underline{Inequalities between $\mc{C}$ and $\mc{D}$:}

These inequalities are symmetrical to those inequalities between 
$\mc{B}$ and $\mc{D}$. The details are omitted here.

\noindent\underline{Inequalities between $\mc{A}$, $\mc{B}$, $\mc{C}$ and $\mc{D}$:}
\begin{itemize}
\item \eqref{eq:124} with $i=l_1$ and $j=l_2$ holds strictly since by \eqref{eq:8}, 
  \begin{align*}
    (s_{l_1+1,l_2}&+s_{l_1,l_2+1})-(s_{l_1,l_2}+s_{l_1+1,l_2+1})\\
=\ &(l_2n_1+(l_1+1)(n_2-l_2)+l_1n_2+(l_2+1)(n_1-l_1))\\ &-(l_1n_2+l_2n_1-l_1l_2+(l_1+1)n_2+(l_2+1)n_1-(l_1+1)(l_2+1))\\
=\ &1.
  \end{align*}
\end{itemize}
\hfill\QQQ

\emph{Proof of Lemma \ref{fd802}}
Without
loss of generality, we assume $p,p'\in\mc{P}^*_n$ and the first block $p'$ is the union
of the first two blocks of $p$
and the other blocks of $p'$ are the other corresponding blocks of $p$,
i.e., $\N'_1=\N_1\cup\N_2=\{1,\cdots,n'_1\}$ and
$\N'_i=\N_{i+1},l=2,\cdots,t-1$ and furthermore $\N_1=\{1,\cdots,n_1\},\N_2=\{n_1+1,\cdots,n'_1\}$.
 With
this assumption, we have $n'_1=n_1+n_2$ and $n'_i=n_{i+1}$,
$i=2,\cdots,t-1$.

By the discussion above \eqref{bfkdg},
$\mf{E}_{p'}=\{\mc{E}_{p'}(\bs{\lambda}'):\bs{\lambda}'\in\mf{N}_{p'}\}$ with
\begin{align*}
  \mf{N}_{p'}= &\mf{N}_A\cup\mf{N}_B\cup\mf{N}_C\\
= &\{[\mb{1}_{t-1}(l),\mb{0}_{t-1}]:1\le l\le
  t-1\}\\ &\cup\{[\mb{1}_{t-1}(l_1,l_2),(k'_1,\cdots,k'_{t-1})]:1\le l_1<l_2\le t-1, (k'_1,\cdots,k'_t)\in\N_{p'},
k'_{l_1}\neq n'_{l_1}, k'_{l_2}\neq n'_{l_2}\}
\\ &\cup\{[\mb{2}_{t-1}(l),(k'_1,\cdots,k'_{t-1})] :1\le l\le t-1,(k'_1,\cdots,k'_{t-1})\in\N_{p'},
k'_{l}\neq n'_{l}-1, n'_{l}\}.
\end{align*}
For the convenience of our discussion, we partition $\mf{N}_{p'}$ into
the following five disjoint subsets,
\begin{enumerate}
\item 
$\mf{N}_1\triangleq\{[\mb{1}_{t-1}(1),\mb{0}_{t-1}]\}
\cup\{[\mb{2}_{t-1}(1),(k'_1,\cdots, k'_{t-1})]:(k'_1,\cdots,k'_{t-1})\in\N_{p'},
k'_1\neq n'_1-1,n'_1\}$;
\item
$\mf{N}_2\triangleq\{[\mb{1}_{t-1}(1,l),(k'_1,\cdots,k'_{t-1})]: 2\le
l\le t-1, (k'_1,\cdots,k'_t)\in\N_{p'}, k'_{1}\neq
n'_{1}, k'_{l}\neq n'_{l}\}$;
\item 
$\mf{N}_3\triangleq\{[\mb{1}_{t-1}(l),\mb{0}_{t-1}]:2\le l\le
  t-1\}$;
\item 
$\mf{N}_4\triangleq\{[\mb{2}_{t-1}(l),(k'_1,\cdots,k'_{t-1})]: 2\le l\le t-1, (k'_1,\cdots,k'_{t-1})\in\N_{p'},
k'_{l}\neq n'_{l}-1, n'_{l}\}$;
\item
$\mf{N}_5\triangleq
\{[\mb{1}_{t-1}(l_1,l_2),(k'_1,\cdots,k'_{t-1})]:2\le l_1<l_2\le t-1, (k'_1,\cdots,k'_t)\in\N_{p'},
k'_{l_1}\neq n'_{l_1}, k'_{l_2}\neq n'_{l_2}\}$.
\end{enumerate}
Note that $\mf{N}_1$ is composed of those $\bsl$ in $\N_A$ or $\N_C$
such that $l=1$, that is, for $\bsl\in\mf{N}_1$, $\mc{E}_p(\bsl)$ contains
$E(\mc{I},\mc{K})$ such that $\mc{I}\subset\N'_1$. Subset
$\mf{N}_2\subset\mf{N}_B$ and for 
$\bsl\in\mf{N}_2$, $\mc{E}_p(\bsl)$ contains
$E(\mc{I},\mc{K})$ such that $|\mc{I}\cap\N'_1|=1$. For the remaining, we have $\mf{N}_3=\mf{N}_A\setminus \mf{N}_1$,
$\mf{N}_4=\mf{N}_C\setminus \mf{N}_1$ and $\mf{N}_5=\mf{N}_B\setminus \mf{N}_2$.

\noindent\underline{$\bsl'\in\mf{N}_1$:}

Consider $q=\{\N_1,\N_2\}$ and $q'=\{\N'_1\}$ in
$\mc{P}_{n'_1}$. By Lemma \ref{u89sd}, for every particular $\hat{\bs{\lambda}}'\in\mf{N}_{q'}$ and each
  $\hat{\bs{\lambda}}\in\mf{N}_q$ such that
  $\mc{E}_q(\hat{\bs{\lambda}})\in\mf{E}_{\mc{E}_{[q']}(\hat{\bs{\lambda}}'),q}$, $\mc{E}_q(\hat{\bs{\lambda}})$ has
  an isolation in $\mf{E}_{\mc{E}_{q'}(\hat{\bs{\lambda}}'),q}$. For
$E(\mc{I},\mc{K})\in\mc{E}_{p'}(\bs{\lambda}')$ with $\bs{\lambda}'\in\mf{N}_1$,
since $\mc{I}\subset\N'_1$, they can be treated as
$E(\mc{I},\mc{K}')\in \mc{E}_{n'_1}$ with
$\mc{K'}=\mc{K}\cap\N'_1$. The details are explained in the following.

For $\bs{\lambda}'\in\mf{N}_{p'}$, let $\bs{\lambda}'(1)$ denote
the pair of the first entry of the two vectors in $\bsl'$. For
$\bs{\lambda}\in\mf{N}_p$, let $\bs{\lambda}(1,2)$ denote the pair of the
first two entries of the two 
vectors in $\bsl$. For example,
$[\mb{1}_{t-1}(1),\mb{0}_{t-1}](1)=[(1),(0)]$ and
$[\mb{1}_{t}(1),\mb{0}_{t}](1,2)=[(1,0),(0,0)]$. 
Now for each
$\bs{\lambda}'\in \mf{N}_1$, it can be seen that there exists
$\hat{\bs{\lambda}}'\in\mf{N}_{q'}$ such that
$\bs{\lambda}'(1)=\hat{\bs{\lambda}}'$. Observe that there exists a
bijection
$\omega_{\lambda'}:\mf{E}_{\mc{E}_{p'}(\bs{\lambda}'),p}\to\mf{E}_{\mc{E}_{q'}(\hat{\bs{\lambda}}'),q}$
defined by 
\begin{equation*}
  \omega_{\lambda'}(\mc{E}_p(\bs{\lambda}))=\mc{E}_q(\hat{\bs{\lambda}})
  \text{ if }\bs{\lambda}(1,2)=\hat{\bs{\lambda}}.
\end{equation*}
For example, for
$[\mb{1}_{t-1}(1),\mb{0}_{t-1}]\in\mf{N}_1$,
$[\mb{1}_{t-1}(1),\mb{0}_{t-1}](1)=[(1),(0)]\in\mf{N}_{q'}$. Note that
$\mf{E}_{\mc{E}_{p'}([\mb{1}_{t-1}(1),\mb{0}_{t-1}]),p}=\{\mc{E}_{p}([\mb{1}_{t}(1),\mb{0}_{t}]),\mc{E}_{p}([\mb{1}_{t}(2),\mb{0}_{t}])\}$
and $\mf{E}_{\mc{E}_{q'}([(1),(0)]),q}=\{\mc{E}_{q}([(1,0),(0,0)]),\mc{E}_{q}([(0,1),(0,0)])\}$.
It can be checked that $[\mb{1}_{t}(1),\mb{0}_{t}](1,2)=[(1,0),(0,0)]$
and $[\mb{1}_{t}(2),\mb{0}_{t}](1,2)=[(0,1),(0,0)]$.

Let $\mi_q$ be an isolation of
$\mc{E}_q(\hat{\bs{\lambda}})$
in $\mf{E}_{\mc{E}_{q'}(\hat{\bs{\lambda}}'),q}$. Define $\mi_p$ by
\begin{equation}
\label{io8dn}
  \mi_p(\mc{A})=\mi_q(\mc{A}\cap\N'_1),\ \mc{A}\subset\N.
\end{equation}
It can be checked that $\mi_p$ is an isolation of
$\mc{E}_p(\bs{\lambda})$ in
$\mf{E}_{\mc{E}_{p'}(\bs{\lambda}'),p}$. For example, let $\mi_q$ be
an isolation of $\mc{E}_{q}([(1,0),(0,0)])$ in
$\mf{E}_{\mc{E}_{q'}([(1),(0)]),q}$. Then by \eqref{eq:5},
$\s^q\triangleq\s(\mi_q,q)$ satisfies
\begin{align*}
  s^q_{n_1,n_2}>s^q_{n_1-1,n_2},\ s^q_{n_1,n_2}=s^q_{n_1,n_2-1}.
\end{align*}
By \eqref{io8dn}, $\s^p\triangleq\s_p(\mi_p,p)$ satisfies
$s^p_{n_1,n_2,n_3,\cdots,n_t}=s^q_{n_1,n_2}$,
$s^p_{n_1-1,n_2,n_3,\cdots,n_t}=s^q_{n_1-1,n_2}$ and
$s^p_{n_1,n_2-1,n_3,\cdots,n_t}=s^q_{n_1,n_2-1}$ which implies that
\begin{align*}
  s^p_{n_1,\cdots,n_t}>s^p_{n_1-1,n_2,\cdots,n_t},\ s^p_{n_1,\cdots,n_t}=s^p_{n_1,n_2-1,n_3,\cdots,n_t}.
\end{align*}
Hence $\mi_p$ is an isolation of
$\mc{E}_p([\mb{1}_t,\mb{0}_t])$ in
$\mf{E}_{\mc{E}_{p'}([\mb{1}_{t-1},\mb{0}_{t-1}]),p}$.

\noindent\underline{$\bsl'\in\mf{N}_2$:}

For
this case, we have 
$\mf{E}_{\mc{E}_{p'}(\bs{\lambda}'),p}=\{\mc{E}_p(\bs{\lambda}):\bs{\lambda}\in\{[\mb{1}_{t}(1,l+1),(k_1,\cdots,k_{t})]:(k_1,\cdots,k_{t})\in\N_p,
k_1\neq n_1,\ k_1+k_2=k'_1,k_i=k'_{i-1},i=3,\cdots,t\}\cup\{
[\mb{1}_{t}(2,l+1),(k_1,\cdots,k_{t})]:(k_1,\cdots,k_{t})\in\N_p, k_2\neq n_2,\
k_1+k_2=k'_1,k_i=k'_{i-1},i=3,\cdots,t\}\}$. We only need to treat the
case that $\bsl$ has the form
$[\mb{1}_{t}(1,l+1),(k_1,\cdots,k_{t})]$; the other case follows by
symmetry. Fix $\bsl=[\mb{1}_{t}(1,l+1),(k_1,\cdots,k_{t})]$ where $(k_1,\cdots,k_{t})\in\N_p,
k_1\neq n_1,\ k_1+k_2=k'_1,k_i=k'_{i-1},i=3,\cdots,t$.
Let $\mi$ be a matroid on $\N$ with $U_{k_1+k_{l+1}+1,n_1+n_{l+1}}$ as its submatroid on $\N_1\cup\N_{l+1}$
and other elements as loops, i.e.,
\begin{equation*}
  \mi(\mc{A})=\min\{k_1+k_{l+1}+1,|\mc{A}\cap (\N_1\cup\N_{l+1})|\},\ \mc{A}\subset\N.
\end{equation*}
It can be checked that $\s=\s(\mi,p)$ satisfies that 
\begin{align*}
  s_{k_1+1,k_2,\cdots,k_t}+s_{k_1,k_2,\cdots,k_l,k_{l+1}+1,k_{l+2},\cdots,k_t}&>s_{k_1,k_2,\cdots,k_t}+s_{k_1+1,k_2,\cdots,k_l,k_{l+1}+1,k_{k+2},\cdots,k_t},\\
  s_{i+1,j,k_3,\cdots,k_t}+s_{i,j,k_3,\cdots,k_l,k_{l+1}+1,k_{l+2},\cdots,k_t}&=s_{i,j,k_3,\cdots,k_t}+s_{i+1,j,k_3,\cdots,k_{l+1}+1,\cdots,k_t},\
  i+j=k_1+k_2,(i,j)\neq(k_1,k_2),\\
  s_{i,j+1,k_3,\cdots,k_t}+s_{i,j,k_3,\cdots,k_l,k_{l+1}+1,\cdots,k_{l+2},k_t}&=s_{i,j,k_3,\cdots,k_t}+s_{i,j+1,k_3,\cdots,k_l,k_{l+1}+1,k_{l+2},\cdots,k_t},\ i+j=k_1+k_2,
\end{align*}
which implies that $\mi$ is an isolation of
$\mc{E}_{p}(\bs{\lambda})$ with
$\bs{\lambda}=[\mb{1}_{t}(1,l+1),(k_1,\cdots,k_{t})]$ in
$\mathfrak{E}_{\mc{E}_{p'}(\bs{\lambda}'),p}$. 

\noindent\underline{$\bsl'\in\mf{N}_3$:}
 
For this case, we have $\mathfrak{E}_{\mc{E}_{p'}(\bs{\lambda}'),p}=\{\mc{E}_{p}([\mb{1}_{t}(l+1),\mb{0}_{t}])\}$, containing only one
element. Therefore, all $E\in
\mc{E}_{p'}(\bs{\lambda}')$ are $p$-equivalent. 
Then, any $\mi\in\Psi_p$
such that $\mi\notin E\cap S_p$ for all $E\in
\mc{E}_{p}(\bs{\lambda})$ is by definition an isolation of $\mc{E}_{p}(\bs{\lambda})$ in $\mathfrak{E}_{\mc{E}_{p'}(\bs{\lambda}'),p}$.

\noindent\underline{$\bsl'\in\mf{N}_4$:}

 For this case,
we have
$\mathfrak{E}_{\mc{E}_{p'}(\bs{\lambda}'),p}=\{\mc{E}_p([\mb{2}_{t}(l+1),(k_1,\cdots,k_{t})]):(k_1,\cdots,k_{t})\in\N_p, k_1+k_2=k'_1,k_i=k'_{k-1},i=3,\cdots,t\}$. 
Fix $\bsl=[\mb{2}_{t}(l+1),(k_1,\cdots,k_{t})]$.
Let $\mi$ be a matroid on $\N$ with $U_{k_1+k_{l+1}+1,n_1+n_{l+1}}$ as its submatroid on $\N_1\cup\N_{l+1}$
and other elements as loops, i.e.,
\begin{equation*}
  \mi(\mc{A})=\min\{k_1+k_{l+1}+1,|\mc{A}\cap (\N_1\cup\N_{l+1})|\},\ \mc{A}\subset\N.
\end{equation*}
It can be checked that $\s=\s(\mi,p)$ satisfies
 \begin{align*}
  2s_{k_1,k_2,\cdots,k_l,k_{l+1}+1,k_{l+2},\cdots,k_t}&>s_{k_1,k_2,\cdots,k_t}+s_{k_1,k_2,\cdots,k_l,k_{l+1}+2,k_{l+2},\cdots,k_t},\\
  2s_{i,j,k_3,\cdots,k_l,k_{l+1}+1,k_{l+2},\cdots,k_t}&=s_{i,j,k_3,\cdots,k_t}+s_{i,j,k_3,\cdots,k_l,k_{l+1}+2,k_{l+2},\cdots,k_t},\
  i+j=k_1+k_2,(i,j)\neq(k_1,k_2),
\end{align*}
which implies that $\mi$ is an isolation of $\mc{E}_{p}([\mb{2}_{t}(l+1),(k_1,\cdots,k_{t})])$ in
$\mf{E}_{\mc{E}_{p'}([\mb{2}_{t-1}(l),(k'_1,\cdots,k'_{t-1})]),p}$.

\noindent\underline{$\bsl'\in\mf{N}_5$:}

For this case,
$\mf{E}_{\mc{E}_{p'}(\bs{\lambda}'),p}=\{\mc{E}_p([\mb{1}_{t}(l_1+1,l_2+1),(k_1,\cdots,k_{t})]):(k_1,\cdots,k_{t})\in\N_p,
k_1+k_2=k'_1,k_i=k'_{k-1},i=3,\cdots,t\}$.
Fix $\bsl=[\mb{1}_{t}(l_1+1,l_2+1),(k_1,\cdots,k_{t})]$.
 Let $\mi$ be a matroid on $\N$ with submatroid
$U_{k_1+k_{l_1+1}+k_{l_2+1}+1,n_1+n_l}$ 
on $\N_1\cup\N_{l+1}\cup\N_{l+2}$ and other elements as loops, i.e., 
\begin{equation*}
  \mi(\mc{A})=\min\{k_1+k_{l_1+1}+k_{l_2+1}+1,|\mc{A}\cap (\N_1\cup\N_{l+1}\cup\N_{l+2})|\},\ \mc{A}\subset\N.
\end{equation*}
It can be checked that $\s=\s(\mi,p)$ satisfies
 \begin{align*}
  s_{k_1,k_2,\cdots,k_l,k_{l+1}+1,k_{l+2},\cdots,k_t}+s_{k_1,k_2,\cdots,k_{l_2},k_{l_2+1}+1,k_{k_2+2}\cdots,k_t}&>s_{k_1,k_2,\cdots,k_t}+s_{k_1,k_2,\cdots,k_{l_1},k_{l_1+1}+1,k_{l_1+2},\cdots,k_{l_2}, k_{l_2+1}+1,k_{l_2+2},\cdots,k_t},\\
   s_{i,j,k_3,\cdots,k_l,k_{l+1}+1,k_{l+2},\cdots,k_t}+s_{i,j,k_3,\cdots,k_{l_2},k_{l_2+1}+1,k_{k_2+2}\cdots,k_t}&=s_{i,j,k_3,\cdots,k_t}+s_{i,j,k_3,\cdots,k_{l_1},k_{l_1+1}+1,k_{l_1+2},\cdots,k_{l_2}, k_{l_2+1}+1,k_{l_2+2},\cdots,k_t},\\
  i+j=k_1+k_2,(i,j)\neq(k_1,k_2),
\end{align*}
which implies that
$\mc{E}_{p}(\bs{\lambda})$ in
$\mf{E}_{\mc{E}_{p'}(\bs{\lambda}'),p}$.

Now, we have proved that 
For any
$\mc{E}_{p'}(\bs{\lambda}')\in\mf{E}_{p'}$ with
$\bs{\lambda}'\in\mf{N}_{p'}$, for each $\bs{\lambda}\in\mf{N}_{p}$,
$\mc{E}_{p}(\bs{\lambda})\in
\mf{E}_{\mc{E}_{p'}(\bs{\lambda}'),p}$ has
an isolation in $\mathfrak{E}_{\mc{E}_{p'}(\bs{\lambda}'),p}$. 
Therefore, this
lemma is true. 
\hfill\QQQ

\section*{Acknowledgments}
We thank Dr.\ Tarik Kaced and Dr.\ Satyajit Thakor for the insightful
discussions. We also thank the reviewers for their useful suggestions
for the revision. The first round revision was done when Qi Chen was
working at ECE department, Drexel University. Qi Chen thanks
Prof. Walsh for giving him suggestions on the revision. He also thanks
Prof. Mat\v{u}\'s for pointing out that \eqref{col1} was also proved
in \cite{P03} and that our results have possible applications to the
$[1,n-1]$-bipartite secret-sharing problem \cite{MBjr08}.

This work was partially supported by a grant from the University
Grants Committee of the Hong Kong Special Administrative Region, China
(Project No. AoE/E-02/08) and 
partially supported by grants from the Shenzhen Key Laboratory of Network Coding Key Technology and Application, Shenzhen, China (ZSDY20120619151314964).


\begin{thebibliography}{1}

\bibitem{Y97}
R. W. Yeung, ``A framework for Linear Information Inequalities,''
\emph{IEEE Trans. Inform. Theory}, vol.
  43, no. 11, pp. 1924-1934, Nov. 1997.

\bibitem{F78}
S. Fujishige, ``Polymatroidal dependence structure of a set of
random variables,'' \emph{Info. Contr.}, 39: pp. 55-72,1978.

\bibitem{Y08}
R. W.~Yeung, \emph{Information Theory and Network Coding,} Springer,
2008.  

\bibitem{ZY97}
Z. Zhang and R. W. Yeung, ``A non-Shannon type conditional inequality
of information quantities,'' \emph{IEEE Trans. Info. Theory,} vol.
  43, no. 11 pp. 1982-1986, Nov. 1997.

\bibitem{M06}
F.~Mat\'u\v{s}, ``Piecewise linear conditional information
inequality,'' \emph{IEEE Trans. Info. Theory}, vol. 44, no. 1, pp. 236-238, Jan. 2006.

\bibitem{CY12}
Q. Chen and R. W. Yeung, ``Characterizing the entropy function region
via extreme rays,'' IEEE Info. Theory Workshop, Lausanne
Switzerland, Sept. 2012. 

\bibitem{ZY98}
Z. Zhang and R. W. Yeung, ``On characterization of entropy function
via information inequalities,'' \emph{IEEE Trans. Info. Theory}, vol.
  44, pp. 1440-1452, Nov. 1998.

\bibitem{YYZ01}
X. Yan, R. W. Yeung and Z. Zhang, ``A class of non-Shannon-type
information inequalities and their applications,'' IEEE Int. Symp. Info. Theory, Washington DC, June
  2001.

\bibitem{MMRV02}
Makarychev K, Makarychev Y, Romashchenko A, et al, ``A new class of
non-Shannon-type inequalities for entropies,'' \emph{Communications in
Information and Systems}, 2002, 2(2): 147-166.


\bibitem{Z03}
Z. Zhang ``On a new non-Shannon-type information
inequality,'' \emph{Communications in Information and Systems,}
vol. 3, no. 1, pp. 47-60, June 2003.


\bibitem{DFZ06}
R. Doughterty, C. Freiling and K. Zeger, ``Six new non-Shannon
information inequalities,'' IEEE Int. Symp. Info. Theory, Seattle WA June
  2006.

\bibitem{XWS08}
W. Xu, J. Wang and J. Sun, ``A projection method for derivation of
non-Shannon-type information inequalities,'' IEEE Int. Symp. Info. Theory, Toronto, Canada June
  2008.

\bibitem{DFZ11}
Dougherty R, Freiling C, Zeger K, ``Non-Shannon information inequalities in four random variables,'' arXiv preprint arXiv:1104.3602, 2011.


\bibitem{M07b}
F.~Mat\'u\v{s}, ``Two constructions on limits of entropy
functions,'' 
\emph{IEEE Trans. Info. Theory}, vol. 53, no. 1
  pp. 320-330, Jan. 2007.

\bibitem{M07a}
F.~Mat\'u\v{s}, ``Infinitely many information inequalities,'' 
IEEE Int. Symp. Info. Theory, Nice, France, June
  2007.

\bibitem{NC00}
M. A. Nielsen and I. L. Chuang, \emph{Quantum Computation and Quantum Information,} Cambridge University Press, 2000.

\bibitem{CGY12}
T. Chan, D. Guo, and R. Yeung, ``Entropy functions and determinant
inequalities,'' 2012 IEEE Int. Symp. Info. Theory, Cambridge, MA, Jul. 2012.


\bibitem{CG08}
T. Chan and A. Grant, ``Dualities between entropy functions and
network codes'', \emph{IEEE Trans. Info. Theory,} vol. 54, 
no.10, 4470-4487, Oct. 2008.


\bibitem{CY02}
T. H. Chan and R. W. Yeung, ``On a relation between information
inequalities and group theory,'' \emph{IEEE Trans. Info. Theory},
vol.48, 1992-1995, July 2002.


\bibitem{HRSV00}
D. Hammer, A. Romashchenko, A. Shen and N. Vereshchagin, ``Inequalities for Shannon Entropy and Kolmogorov Complexity,'' \emph{J. Comp. and Syst. Sci.,} 60: 442-464, 2000.


\bibitem{C01}
T. H. Chan, ``A combinatorial approach to information inequalities,'' \emph{Comm. Info. and Syst.,} 1: 241-253, 2001.

\bibitem{Y12}
R. W. Yeung, ``Facets of entropy,'' \emph{IEEE Information Theory Society
Newsletter}, 2012: 6-15.

\bibitem{C11}
T. Chan, ``Recent progresses in characterizing information inequalities,'' \emph{Entropy,} 13: 379-401, 2011.

\bibitem{KR11}
T. Kaced and A. Romashchenko, ``On essentially conditional
information inequalities,'' IEEE Int. Symp.
Info. Theory, Saint-Petersburg, Russia, July 2011.

\bibitem{KR12}
T. Kaced and A. Romashchenko, ``On the non-robustness of
essentially conditional information inequalities,''  IEEE Info. Theory Workshop, Lausanne
Switzerland, Sept. 2012.

\bibitem{KR13}
T. Kaced and A. Romashchenko, ``Conditional Information Inequalities for
Entropic and Almost Entropic Points,''  \emph{IEEE
  Trans. Info. Theory} vol. 59, pp. 7149-7167, Nov. 2013.

\bibitem{G03}
B. Gr\"{u}nbaum, \emph{Convex polytopes,} Springer, 2003.

\bibitem{Z95}
G. M. Ziegler, \emph{Lectures on polytopes,} Springer-Verlag, 1995.

\bibitem{R70}
R. T. Rockafellar, \emph{Convex analysis,} Princeton Univ. Press,
1970.


\bibitem{O92}
J. G. Oxley, \emph{Matroid theory,} Oxford Univ. Press, 1992.

\bibitem{W76}
D. J. A. Welsh. \emph{Matroid theory,} Academic Press, 1976.


\bibitem{YLCZ}
R. W. Yeung, S-Y. R. Li , N. Cai, Z. Zhang, ``Network coding theory,'' Communications and Information Theory, 2005, 2(4): 241-329.

\bibitem{N78}
H. Q. Nguyen, ``Semimodular functions and combinatorial geometries,''
\emph{Trans. AMS.},vol. 238, pp. 355-383, April 1978.


\bibitem{W86}
N. White, ``Theory of matroids,''
Cambridge University Press, 1986.


\bibitem{LW92}
J. H. van Lint and R. M. Wilson, \emph{A course in combinatorics,}
Cambridge University Press, 1992.




\bibitem{Rotman06}
J. Rotman, \emph{A first course in Abstract Algebra with
  applications}, Pearson/Prentice Hall,
2006.

\bibitem{DM96}
J. D. Dixon and B. Mortimer, \emph{Permutation groups,} Springer,
1996.

\bibitem{CY14}
Q. Chen and R. W. Yeung, ``Partition-symmetrical entropy functions,''
http://arxiv.org/abs/1407.7405


\bibitem{Han78}
T. S. Han, ``Nonnegative entropy mesures of multivariate
symmetric correlations,''  \emph{Info. Contr.,} 36, 133-156,
1978.


\bibitem{CY13}
Q. Chen and R. W. Yeung, ``Two-partition-symmetrical entropy function
regions,'' IEEE Info. Theory Workshop  Seville, Spain, Sept. 2013.

\bibitem{P03}
N. Pippenger, ``The inequalities of quantum information theory,'' \emph{IEEE Trans. Info. Theory,} vol.49, no.4,
April 2003. 




\bibitem{S79}
A. Shamir. ``How to share a secret,'' \emph{Communications of the
  ACM}, 22(11): 612-613, 1979.

\bibitem{IKS13}
Y. Ishai, E. Kushilevitz, and O. Strulovich, ``Lossy chains and
fractional secret sharing,'' Symposium on Theoretical Aspects of Computer Science, 2013.

\bibitem{FHKP14}
O. Farr\`{a}s, T. Hansen, T. Kaced and C. Padr\'{o}, ``Optimal
Non-Perfect Uniform Secret Sharing Schemes,''
CRYPTO 2014, Santa Barbara, August 17-21, 2014.

\bibitem{PS00}
C. Padr{\'o} and G. S{\'a}ez ``Secret sharing schemes with bipartite
access structure,'' \emph{IEEE Trans. Info. Theory,} vol.46, no.7,
July 2000. 

\bibitem{FFP07}
O. Farr{\`a}s, and J. Mart{\'\i}-Farr{\'e}, and C. Padr{\'o} ``Ideal Multipartite Secret
Sharing Schemes,'' \emph{ J. Cryptology} 25 (2012) 434-463, 2007.





\bibitem{M94a}
F. Mat\'u\v{s}, ``Probabilistic conditional independence structures and matroid theory: background,'' \emph{Int. J. of General Systems,} 22 185-196, 1994.

\bibitem{MBjr08}
J. R. Metcalf-Burton, ``Information rates of minimal
non-matroid-related access structures,'' http://arxiv.org/abs/0801.3642


\bibitem{Cm97}
L. Csirmaz, ``The size of a share must be large,''
\emph{J. Cryptology} 10:223-231, 1997.


\bibitem{TCS13}
S. Thakor, T. Chan and K. W. Shum, ``Symmetry in distributed data
storage systems,'' 2013 IEEE Int. Symp. Info.
Theory, Istanbul, Turkey, July 2013.

\bibitem{AW15}
J. Apte, J. W. Walsh, ``Symmetry in network coding,'' IEEE Int. Symp. Info.
Theory, Hong Kong, June 2015.

\end{thebibliography}
\end{document}